\documentclass[14pt]{article}

\usepackage[T1]{fontenc}
\usepackage[utf8]{inputenc}
\usepackage{times}
\usepackage[a4paper,margin=1in]{geometry}
\usepackage[round,authoryear]{natbib}

\usepackage{graphicx}
\usepackage{amsmath,amssymb,amsfonts,mathtools}
\usepackage{dsfont}
\usepackage{cancel}
\usepackage{amsthm}

\usepackage{caption}
\usepackage{subcaption}
\usepackage{float}
\usepackage{rotating}
\usepackage{array,makecell}

\usepackage{algorithm}
\usepackage{algpseudocode}

\usepackage{pgffor}
\usepackage{authblk}
\usepackage{hyperref}

\graphicspath{{picture/}}
\algblock{ParFor}{EndParFor}
\algnewcommand\algorithmicparfor{\textbf{parfor}}
\algnewcommand\algorithmicpardo{\textbf{do}}
\algnewcommand\algorithmicendparfor{\textbf{end\ parfor}}
\algrenewtext{ParFor}[1]{\algorithmicparfor\ #1\ \algorithmicpardo}
\algrenewtext{EndParFor}{\algorithmicendparfor}
\newcommand{\stir}{\genfrac{\lbrace}{\rbrace}{0pt}{}}

\renewcommand{\d}{\mathrm{d}}

\newcommand{\leb}[1]{ \lvert #1 \rvert }
\DeclarePairedDelimiter\floor{\lfloor}{\rfloor}
\newcommand{\giventhat}{\; \Big \vert \;}
\foreach \x in {a,...,z}{
	\expandafter\xdef\csname bf\x
	\endcsname{
		\noexpand\ensuremath{\noexpand\mathbf{\x}}
	}
}
\foreach \x in {A,...,Z}{
	\expandafter\xdef\csname bf\x
	\endcsname{
		\noexpand\ensuremath{\noexpand\mathbf{\x}}
	}
	\expandafter\xdef\csname cal\x
	\endcsname{
		\noexpand\ensuremath{\noexpand\mathcal{\x}}
	}
	\expandafter\xdef\csname bb\x
	\endcsname{
		\noexpand\ensuremath{\noexpand\mathbb{\x}}
	}
	\expandafter\xdef\csname sf\x
	\endcsname{
		\noexpand\ensuremath{\noexpand\mathsf{\x}}
	}
	\expandafter\xdef\csname frak\x
	\endcsname{
		\noexpand\ensuremath{\noexpand\mathfrak{\x}}
	}
}
\def\toolboxmcrppy{\href{https://github.com/dhawat/MCRPPy}{MCRPPy}}
\newcommand{\Var}{\bbV \text{ar}}

\theoremstyle{plain}
\newtheorem{theorem}{Theorem}[section]
\newtheorem{proposition}[theorem]{Proposition}
\newtheorem{corollary}[theorem]{Corollary}
\newtheorem{lemma}[theorem]{Lemma}
\theoremstyle{remark}
\newtheorem{remark}[theorem]{Remark}
\newcommand{\bmsection}[1]{\par\medskip\noindent\textbf{#1.}\ }
\newcommand{\bmsubsection}[1]{%
  \par\smallskip
  \noindent\textbf{#1.}\hspace{0.5em}%
}
\title{Repelled point processes with application to numerical integration}

\author[1,2]{Diala Hawat}
\author[3]{Gabriel Mastrilli}
\author[1]{Rémi Bardenet}
\author[2,4]{Raphaël Lachièze-Rey}

\affil[1]{Université de Lille, CNRS, Centrale Lille, UMR 9189 -- CRIStAL, Lille, France}
\affil[2]{Université Paris Cité, MAP5, Paris, France}
\affil[3]{Univ Rennes, ENSAI, CNRS, CREST-UMR 9194, Rennes, France}
\affil[4]{Inria Paris, DYOGENE, Paris, France}

\affil[*]{\texttt{dialahawat7@gmail.com}}

\begin{document}
\maketitle
\begin{abstract}
    We look at Monte Carlo numerical integration from a stochastic geometry point of view.
    While crude Monte Carlo estimators relate to linear statistics of a homogeneous Poisson
    point process (PPP), linear statistics of more regularly spread point processes can yield
    unbiased estimators with faster-decaying variance, and thus lower integration error.
    Following this intuition, we introduce a Coulomb repulsion operator, which reduces
    clustering by slightly pushing the points of a configuration away from each other.
    Our empirical findings show that applying the repulsion operator to a PPP as well as,
    intriguingly, to more regular point processes, preserves unbiasedness while reducing
    the variance of the corresponding Monte Carlo estimator, thus enhancing the method.
    We prove this variance reduction when the initial point process is a PPP.
    On the computational side, the complexity of the operator is quadratic and the
    corresponding algorithm can be parallelized without communication across tasks.
    \end{abstract}
	
% main content
\section{Introduction} % (fold)
\label{sec:Introduction}

Numerical integration --~the task of approximating integrals using pointwise evaluations of the integrand~-- has a rich history,  with (subjective) milestones such as the quadrature of \cite{Gauss1815} or the Metropolis algorithm \citep{MRRTT53}.
Among methods that scale best to larger dimensions, Monte Carlo methods may be the most popular in applications; see e.g. \citep{Owen2013}.
In its simplest form, crude Monte Carlo amounts to writing the target integral as an expectation and estimating it by a sample average formed with $N$ i.i.d.\ samples.
The performance of the estimator is typically measured by the mean squared error, which classical probability results show decreases as $O(1/N)$.
This rate is slow, and variance reduction for Monte Carlo integration has been a rich topic of research.
We summarize here a few leading approaches.

One popular variance reduction technique is to leverage auxiliary integrands with known integrals, a method known as \emph{control variates}, see \citep[Chapters 8 and 10]{Owen2013} for classical results, and \citep{SOMD22} and references therein for more recent work.
A second family of methods directly improve the rate of convergence of the mean squared error by introducing sophisticated dependence through weights in the sample average, while the quadratures nodes remain independent  \citep{DelPor2016,AzaDelPort2018,Lelal2023}.
The cost of the seminal approach of \cite{DelPor2016}, for instance, is quadratic in $N$.
A third recent line of research focuses on the direct minimization of a measure of similarity between the target measure and the empirical measure of the quadrature nodes, with algorithms such as Stein variational gradient descent \citep{liu2016stein} or kernel Stein discrepancy minimization \citep{korba2021kernel}.  
In particular, the kernel Stein discrepancy is the worst-case integration error for integrands belonging to a space of smooth functions, namely a specific reproducing kernel Hilbert space (RKHS) chosen for numerical tractability.
Because of the relative recency of these approaches, theoretical results that back up the experimentally observed improvement over the Monte Carlo rate are still rather scarce. 
One key progress is the two-stage thinning algorithm of \cite{dwivedi2021generalized, dwivedi2024kernel}: the authors show how to reduce $N^2$ classical Monte Carlo samples to a sample of cardinality $N$ with a worst-case integration error over an RKHS essentially of order $1/N$ in probability.
Moreover, the thinning can be done in quadratic time with respect to $N$ \cite{shetty2022distributioncompressionnearlineartime}.

Finally, a fourth approach considers expectations under more regularly spread probability distributions than i.i.d.\ or Markov chain draws, as in randomized quasi-Monte Carlo (RQMC) \citep{Owen2008} or with determinantal point processes (DPPs) \citep{BarHar2020,CoeMazAmb2021,BeBaCh20}.
Our contribution is of the latter kind, taking an expectation under a well-spread random set of points.
But while the cost of sampling the DPP in \citep{BarHar2020} is (at least) cubic in $N$ \citep{GaBaVa19}, we aim at keeping the sampling cost quadratic, to match the costs of the aforementioned works of \cite{DelPor2016} and \cite{shetty2022distributioncompressionnearlineartime},
while still introducing enough dependence to force variance reduction.

Our inspiration comes from a set of results in stochastic geometry, on gravitational allocation \citep{ChaPelal2010,NazSodVol2007}, hyperuniform point processes \citep{Tor2018, Cos2021, Kla2019}, and systems of points with Coulomb interactions \citep{Ser2019}.
Loosely speaking and for our purpose, a point process is a random locally finite set of points in $\mathbb{R}^d$.
A natural reference is the homogeneous Poisson point process (PPP), whose points in a compact set are i.i.d.\ uniform random variables, once we condition on the number of points in that compact set.
At the other end of the spectrum, the counting statistics of \emph{hyperuniform} point processes yield Monte Carlo estimators with mean squared error decaying at a faster rate than the PPP.
For instance, \cite{Kla2019} study the candidate hyperuniform point process that results from applying Lloyd's algorithm iteratively to a PPP.
In a similar vein but leaving aside any iterative scheme for mathematical tractability, we propose to use gravitational allocation to yield a provable variance reduction compared to the PPP.
%\dhawat{Moreover, we intend to get a variance reduction even when applying the same method to (stationary) ''regular'' point processes, to yields a variance reduction. But this may be harder to prove theoretically.}

We introduce the \emph{repulsion operator}, a tool designed to mitigate clustering by slightly pushing the points within a stationary point process away from each other.
Roughly speaking, the operator applies to a configuration of points as follows: Imagining that the points repel each other through a Coulomb interaction, we perform a \emph{single} step of a numerical scheme to integrate the corresponding differential equation.
When the repulsion operator is applied to a point process, we term the resulting object the \emph{repelled point process}.
Our main theoretical result is that applying the repulsion operator to a PPP of $\bbR^d$ with $d \geq 3$ yields an unbiased Monte Carlo method with \emph{lower} variance than under the initial PPP when the integrand is a $C^2$ function of compact support.
Furthermore, our numerical experiments suggest that the variance reduction also holds when the repulsion operator is applied to point processes that exhibit more regularity than a PPP. 
Examples include the Ginibre Ensemble, known for its hyperuniformity, and the Scrambled Sobol sequence, commonly used in randomized quasi-Monte Carlo.
This finding raises the intriguing possibility that the repulsion operator may be universal, in the sense that it consistently results in a variance reduction, regardless of the (stationary) point process to which it is applied. 
If true, applying the repulsion operator could become a simple and general postprocessing step in any Monte Carlo integration task.

The paper is structured as follows.
Section \ref{sec:Background and Notations} provides background information about point processes.
Section \ref{sec:The repelled Poisson point process} introduces the repulsion operator. %a parametric operator that operates on locally finite sets of points in $\mathbb{R}^d$.
We analyze the properties of this operator and present our main theoretical result, showing variance reduction when the repulsion operator is applied to a PPP, compared to crude Monte Carlo.
Additionally, we describe a sampling procedure and present an experimental illustration of the variance reduction.
In Section \ref{sec:Application to numerical integration}, we put our method in context, by conducting a comparison with standard Monte Carlo methods on synthetic integration tasks.
In Section \ref{sec:ther models and properties}, we explore additional aspects of the repulsion operator, such as iterating it several times, estimating the pair correlation function and structure factor of the RPPP, as well as applying the repulsion operator to already repulsive point processes.
Section \ref{sec:Conclusion} concludes the paper with a few research directions.
Finally, all proofs are gathered in Appendix \ref{sec:Proofs}, while Appendix \ref{app:experiments} gives extra simulation results.
% --------------------------------------
%---------------------------------------
\section{Point processes and their intensity functions} % (fold)
\label{sec:Background and Notations}
In this section, we provide some background on point processes, with key results like the Slyvniak-Mecke theorem.
We refer to \citep{ChiuStoal2013} for details.

\begin{remark}
    Throughout this manuscript, bold lowercase letters, like $\bfx $, indicate vectors in $\mathbb{R}^d$.
    The corresponding non-bold characters, like $x$, are scalars.
    In particular, we write $\bfx  = (x_1, \dots, x_d)$.
    Whenever not confusing, we use the same letter in different fonts for a vector and its Euclidean norm, i.e., $r=\Vert \mathbf{r} \Vert_2$ and $k=\Vert \mathbf{k} \Vert_2$.
    Calligraphic letters like $\calX$ are used for point processes, i.e., random configurations of points.
    Configurations themselves are denoted with sans-serif letters like $\sfX$.
    When the cardinality of a particular point process is almost surely constant, we sometimes write that point process as $\calX_N$, with the value $N$ of the cardinality as the index.
\end{remark}

% subsectionNotations (end)

\subsection{Spatial point processes}
\label{sec:Spatial point processes}

Let a configuration of $\mathbb{R}^d$ be a locally finite set $\sfX$ of $\bbR^d$, that is, for any bounded Borel set $B$ of $\mathbb{R}^d$ the cardinality $\sfX(B)$ of $\sfX\cap B$ is finite.
Endow the family $\frakN$ of such configurations with the $\sigma$-algebra generated by the mappings $\sfX \mapsto \sfX (B)$, for any bounded Borel set $B$.
Formally, a \emph{simple\footnote{The term \emph{simple} indicates that the point process almost surely consists of distinct points.
For us, this is a direct consequence of defining configurations as \textbf{sets}.
However, some authors avoid assuming simplicity; see e.g. \citep[Chapter 6]{LasPen2017}.}} \emph{spatial point process}, hereafter \emph{a point process}, is a random element $\calX$ of $\frakN$.
The distribution of $\calX$ is determined by the system of \emph{void probabilities}
\begin{equation}
    \label{eq:def_void_proba}
    \bbT_{\calX}(K) \triangleq \bbP \left(\calX(K) = 0\right) ,
\end{equation}
as $K$ ranges through the compact sets of $\bbR^d$.
%In words, the void probability $\bbT_{\calX}(B)$ is the probability of observing no point of $\calX$ in a Borel set $B$.

In this paper, we work with \emph{stationary} and \emph{isotropic} point processes, also called
\emph{motion-invariant} point processes.
By stationary, we mean that the law of the point process is invariant by translation: the law of $\calX$ is identical to that of
$\calX + \bfy
\triangleq \{\bfx  + \bfy ; \; \bfx  \in \calX\}$,
for all $\bfy \in \mathbb{R}^{d}$.
Similarly, a point process is isotropic if its law is invariant by rotation.

Point processes are often described in terms of their intensity measures.
The \emph{first intensity measure} $\mu^{(1)}$, for instance, is defined by
\begin{equation*}
    \label{eq:}
    \mu^{(1)}(B) \triangleq \bbE \left[\calX(B)\right] ,
\end{equation*}
for any Borel set $B\subset\mathbb{R}^d$.
When $\mu^{(1)}$ has a density with respect to the Lebesgue measure, $\mu^{(1)}(\d  \bfx ) = \rho^{(1)}
( \bfx )\, \d  \bfx $, we call $\rho^{(1)}$ the \emph{intensity} of $\calX$.
If $\calX$ is stationary, then $\mu^{(1)}$ is proportional to the Lebesgue measure and the intensity $\rho^{(1)}$ is a positive constant $\rho>0$, equal to the mean number of points of $\calX$ per unit volume.

More generally,
%by Campbell's theorem,
the \emph{$n$-th order intensity measure}
$\mu^{(n)}$ of $\calX$
% (a.k.a. $n$-th order factorial moment)
is defined by
\begin{equation}
    \label{eq:campbell_theorem}
    \bbE \left[\sum_{\bfx _1, \dots,\bfx _n \in \calX}^{\neq} f(\bfx _1, \dots,\bfx _n)\right] = \int_{\mathbb{R}^d \times \dots \times \mathbb{R}^d} f(\bfx _1, \dots,\bfx _n) \mu^{(n)}(\d \bfx _1  \dots \d \bfx _n) ,
\end{equation}
where $f$ is any non-negative bounded measurable function, and the summation is over all $n$-tuples of distinct points in $\calX$; see \cite[Chapter 4]{ChiuStoal2013,LasPen2017}.
Again, when $\mu^{(n)} = \rho^{(n)} ( \bfx _1,  \dots, \bfx _n) \d \bfx _1  \dots \d \bfx _n$, $\rho^{(n)}$ is called the \emph{$n$-th order intensity function}.
Intuitively, for any pairwise distinct points $\bfx _1,  \dots, \bfx _n$, $\rho^{(n)} ( \bfx _1,  \dots, \bfx _n)\d \bfx _1  \dots \d \bfx _n$ is the probability that $\calX$ has a point in each of the $n$ infinitesimally small sets around $\bfx _1,  \dots, \bfx _n$, with respective volumes $\d \bfx _1,  \dots, \d \bfx _n$.

\subsection{The homogeneous Poisson point process (PPP)} % (fold)
\label{sub:Poisson Point process}
Consider a compact set $K$ of $\mathbb{R}^d$, and pick $N$ i.i.d. points uniformly distributed in $K$.
The point process $\calX_N$ formed by these $N$ points is called the Binomial point process (BPP) of $N$ points.
Loosely speaking, when $K$ is enlarged to fill out $\mathbb{R}^d$ while maintaining $N = \rho \vert K\vert$, where $|K|$ is the Lebesgue measure of $K$, we obtain a limiting point process $\calP$ that is called the homogeneous Poisson point process (PPP) of intensity $\rho>0$.
We now list a few properties of $\calP$.
%When $\rho=1$, we call $\calX$ the standard Poisson point process.

First, $\calP$ is motion-invariant.
Second, the random number $\mathcal{P}(B)$ of points of $\calP$ in a bounded Borel set $B$ has a Poisson distribution of mean $\rho |B|$.
In particular, the void probability \eqref{eq:def_void_proba} is $\bbT_{\calP}(B) = \exp(-\rho |B|)$.
Third, $\calP(B_1)$ and $\calP(B_2)$ are independent for any disjoint Borel sets $B_1$ and $B_2$.
This second fundamental property is known as \emph{complete randomness}, and translates the intuition that the PPP has as little structure as possible; for more details see \citep[Chapter 2]{ChiuStoal2013}.
Fourth, all the moments of $\calP$ are determined by $\rho$, i.e., for any non-negative measurable function $f$,
\begin{equation}
    \label{eq:def_moment_poisson}
    \bbE \left[\sum_{\bfx _1, \dots,\bfx _n \in \calP}^{\neq} f(\bfx _1, \dots,\bfx _n)\right] = \rho^n \int_{\mathbb{R}^d \times \dots \times \mathbb{R}^d} f(\bfx _1, \dots,\bfx _n) \d \bfx _1 \dots \d \bfx _n .
\end{equation}
With the notation of Section~\ref{sec:Spatial point processes}, this reads $\rho^{(n)}=\rho^n$.
Moreover, by the so-called extended Slivnyak-Mecke theorem, we have
\begin{equation}
    \label{eq:slivnyak-mecke}
    \bbE \left[\sum_{\bfx _1, \dots,\bfx _n \in \calP}^{\neq} h \left(\bfx _1, \dots,\bfx _n,\calP \setminus \{\bfx _1, \dots,\bfx _n\} \right)\right] = \rho^n \int_{\mathbb{R}^d \times \dots \times \mathbb{R}^d} \bbE [h(\bfx _1, \dots,\bfx _n, \calP)] \d \bfx _1 \dots \d \bfx _n ,
\end{equation}
for any non-negative measurable function $h$ on $(\mathbb{R}^{d})^n \times \frakN$; see \citep[Section 5.1]{CoeMolWaa2015}.
Equation \eqref{eq:slivnyak-mecke} provides further evidence of the PPP's lack of dependency structure: informally, once conditioned on a finite number of points belonging to $\calP$, the rest of $\calP$ has the same distribution as $\calP$.
%------------------------------
%---------------------------
\section{Repelled point processes} % (fold)

\label{sec:The repelled Poisson point process}
In this section, given a configuration $\sfX\in   \mathfrak N$ and a parameter $\varepsilon$, we explain how to construct another configuration $\Pi_\varepsilon(\sfX)$, called the \emph{repelled configuration}.
Keeping in mind our motivation for numerical integration, we want $\Pi_\varepsilon$ to be $(i)$ computationally cheap to apply.
Moreover, when applied to a random configuration, $\Pi_\varepsilon$ should $(ii)$ preserve stationarity, isotropy, and intensity, and $(iii)$ reduce the variance of linear statistics.
While the first condition $(i)$ ensures that sampling from $\Pi_\varepsilon(\sfX)$ remains relatively tractable, the latter two $(ii)$ and $(iii)$ guarantee that $\Pi_\varepsilon(\sfX)$ offers a Monte Carlo method with reduced mean squared error compared to the Monte Carlo method based on $\sfX$.
% To gain a clearer understanding of our motivation and requirements see Section \ref{sub:Monte Carlo methods}.

In Section~\ref{s:repulsion_operator}, we define the repulsion operator $\Pi_\varepsilon$.
In Section~\ref{sub:Properties of the repelled Poisson point process}, we detail the properties of $\Pi_\varepsilon(\calP)$, where $\calP$ is a homogeneous Poisson point process (PPP).
Our main theoretical result is in Section~\ref{sub:Main result}, where we show that, for small enough $\varepsilon>0$, the variance of linear statistics under $\Pi_\varepsilon \calP$ is smaller than under $\calP$.
In particular, we give theoretical support for the choice of a particular value of $\varepsilon$, which is both explicit and independent of the considered linear statistic.
Our repulsion operator is based on a stochastic process known as Coulomb force. We discuss basic properties of the latter in Section~\ref{sub:Properties_of_F}.
In Section \ref{sub:Sampling from the repelled point process}, we explain how to \emph{approximately} sample from $\Pi_\varepsilon \calX$, in time quadratic in the number of points of the point process $\calX$ in the observation window.
Finally, Section~\ref{sec:experimental_illustration} gives a first experimental illustration of our variance reduction result.
% Section~\ref{sub:Properties_of_F} contains theoretical foundations required to establish the outcomes of Sections~\ref{sub:Properties of the repelled Poisson point process}, \ref{sub:Main result}, and \ref{sub:Sampling from the repelled point process}.

% In this section, we outline the steps involved in building $\Pi_\varepsilon \calX$ (referred to as \emph{repelled point process}).
% First, we introduce the Repulsion operator $\Pi_\varepsilon$.
% Then we apply $\Pi_\varepsilon$ to a homogeneous Poisson point process $\mathcal{P}$ and present several properties of $\Pi_\varepsilon \mathcal{P}$ and we state the main result. We conclude by presenting our suggested technique for sampling from $\Pi_\varepsilon \mathcal{P}$.

\subsection{The repulsion operator} % (fold)
\label{s:repulsion_operator}
\begin{figure}
    \begin{subfigure}{0.48\textwidth}
        \centering
        \includegraphics[width=0.9\linewidth]{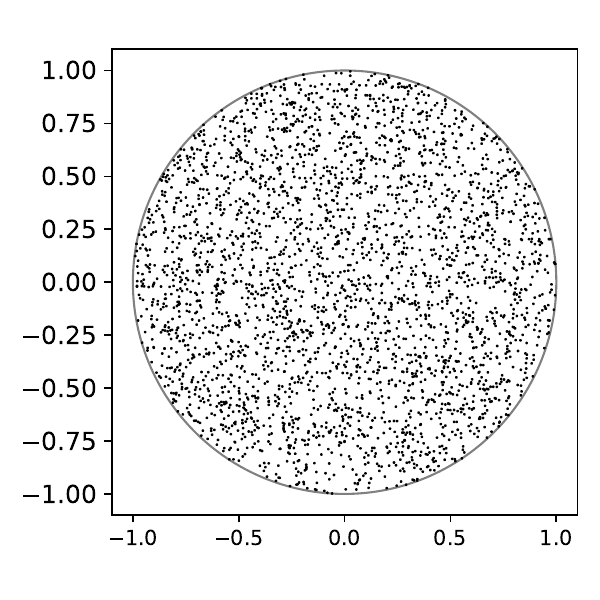}
        \caption{PPP}
    \end{subfigure}
    \begin{subfigure}{0.48\textwidth}
        \centering
        \includegraphics[width=0.9\linewidth]{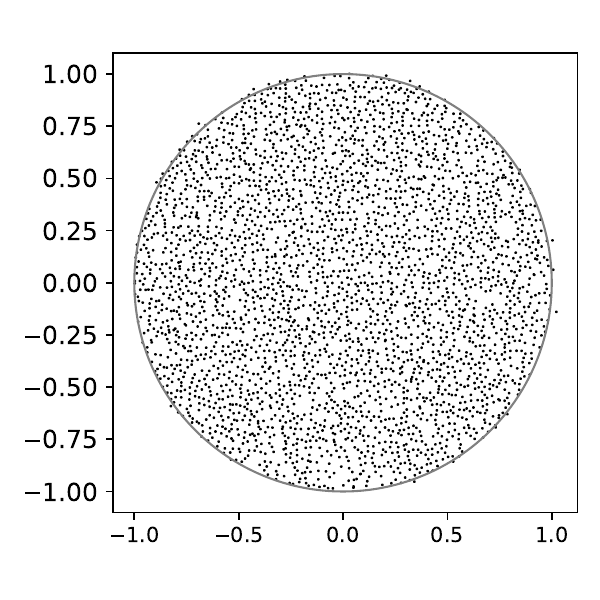}
        \caption{RPPP}
    \end{subfigure}
    \caption{A sample from a homogeneous Poisson point process of intensity $\rho=1000$ and the corresponding repelled sample.}
    \label{fig:poisson_push_force}
\end{figure}

For $\bfx  \in \mathbb{R}^d$ and a configuration $\sfX \in \frakN$, consider the series (when it converges),
\begin{equation}
    \label{eq:def_gravitational_force}
    F_{\sfX}(\bfx ) \triangleq
    \sum_{\substack{\mathbf{z} \in \sfX \setminus \{\bfx \}\\ \lVert \bfx  - \mathbf{z} \lVert_2 \uparrow }}
    \frac{\bfx  - \mathbf{z}}{ \lVert \bfx  - \mathbf{z} \lVert_2^d}=
    \lim_{R\to \infty }
    \sum_{  {\bf z}\in \sfX \setminus \{\bfx \} \cap B({\bf x},R)}
    \frac{\bfx  - \mathbf{z}}{ \lVert \bfx  - \mathbf{z} \lVert_2^d} .
    \tag{$F_1$}
\end{equation}
Several observations are in order.
First, each term in the sum in \eqref{eq:def_gravitational_force} intuitively represents the Coulomb force felt by a charged particle at $\bfx $ and due to a particle of the same charge placed at $\mathbf{z}$.
In a dynamic setting, this force would repel $\bfx $ away from $\mathbf{z}$.
Second, as the series defining $F_{\sfX}(\bfx)$ is not absolutely convergent, the order of the summation is important.
Following \citealp{ChaPelal2010}, we consider the limit in an increasing ball centered at $\bfx$, i.e., the summands in \eqref{eq:def_gravitational_force} are arranged in order of increasing distance from $\bfx $.
We will discuss in Section \ref{sub:Properties_of_F} rearranging the summation by increasing distance from the origin.
Third, a fundamental insight, originally mentioned by \cite{Chan1943}, states that if $d \geq 3$ and $\calP$ is a PPP, then, for every $\bfx $, the series defining $F_{\mathcal{P}}(\bfx )$ converges almost surely. Further information regarding the characteristics of $F_{\calP}$ can be found in Section \ref{sub:Properties_of_F}.

Call $\sfX  \in\frakN$ a \emph{valid} configuration if, for all $\bfx$, the limit defining $F_{\sfX }(\bfx)$ in \eqref{eq:def_gravitational_force} exists.
% We denote by $ \frakN_{F}$ the class of valid configurations.
For $\varepsilon\in\mathbb{R}$, we define the (Coulomb) repulsion\footnote{
    While we generally speak of ``repulsion'', note that when $\varepsilon < 0$ the dynamics become attractive instead of repulsive.
}
operator $\Pi_\varepsilon$, acting on valid configurations, through
\begin{align}
    \label{eq:def_repulsion_operator}
    \Pi_\varepsilon:  \sfX & \mapsto \{\bfx  + \varepsilon F_{\sfX}  (\bfx ): \bfx  \in {\sfX}\} .
\end{align}
There are two formal caveats to our definition \eqref{eq:def_repulsion_operator}.
First, it only applies to valid configurations.
Second, since by definition, $\Pi_\varepsilon \sfX$ is a set, it does not keep track of multiplicities, arising when several points in $\sfX$ are mapped to the same location by $\Pi_\varepsilon$.
Anticipating a bit, Proposition \ref{prop:motion_invariance_for_poisson} below shows that these two caveats are irrelevant when $\Pi_\varepsilon$ is applied to a PPP called $\calP$.
In particular, $\calP$ is almost surely a valid configuration, and for any two distinct points $\bfx, \bfy\in \calP$, almost surely
$$
\bfx  + \varepsilon F_{\calP}(\bfx) \neq \bfy + \varepsilon F_{\calP } (\bfy).
$$
This guarantees that $\Pi_\varepsilon \mathcal{P}$ is a \emph{simple} point process, which we term the \emph{repelled Poisson point process (RPPP)}.
We will occasionally consider the repelled point process $\Pi_\varepsilon \calX$ of a more general point process $\calX$, although its existence needs to be discussed.

The first panel of Figure \ref{fig:poisson_push_force} displays a sample from a PPP of intensity $\rho=1000$ in $d=2$, intersected with a disk-shaped observation window.
Note that we plot the construction in dimension $2$ for graphical convenience, but we are not making any convergence claim for \eqref{eq:def_gravitational_force} in $d=2$.
We illustrate the RPPP construction in the second panel of the figure.
A detailed explanation of the simulation procedure will be provided in Section \ref{sub:Sampling from the repelled point process}.
At this stage, we simply observe that the RPPP sample exhibits a reduced tendency for points to cluster together, compared to the PPP sample.
%---------------------------------------------------
%---------------------------------------------------

\subsection{Properties of the repelled point processes} % (fold)
\label{sub:Properties of the repelled Poisson point process}
In this section, we state some properties of the repelled point processes.
\begin{proposition}[Motion-invariance]
    \label{prop:motion_invariance}
    Let $\calX$ be a point process that is almost surely valid, and $\varepsilon \in \bbR$.
    If $\calX$ is motion-invariant, then $\Pi_\varepsilon \calX$ is also motion-invariant.
\end{proposition}
The proof of this proposition is deferred to Appendix \ref{sub:Proof_of_Proposition_1}.
Of particular importance to us is the following corollary.
\begin{proposition}
    \label{prop:motion_invariance_for_poisson}
    Let $d\geq 3$, and $\mathcal{P} \subset \mathbb{R}^d$ be a Poisson point process of intensity $\rho>0$.
    Then, for any $\varepsilon \in \mathbb{R}$,  and any two distinct points $\bfx, \, \bfy \in \calP$, we have, almost surely,
    $$
    \bfx  + \varepsilon F_{\calP}(\bfx) \neq \bfy + \varepsilon F_{\calP } (\bfy) .
    $$
    Moreover, $\Pi_\varepsilon \mathcal{P}$ is a stationary and isotropic point process of intensity $\rho$.
\end{proposition}
The proof of this Proposition is also deferred to Appendix \ref{sub:Proof_of_Proposition_1}.
According to Proposition \ref{prop:motion_invariance_for_poisson}, $\Pi_\varepsilon \mathcal{P}$ is of intensity $\rho$, the same intensity as $\mathcal{P}$.
Consequently, for any integrable function $f$ of compact support $K$, Equation~\eqref{eq:def_moment_poisson} yields
\begin{equation*}
    \bbE \left[
    \sum_{\bfx  \in \Pi_\varepsilon \mathcal{P}} f(\bfx )\right]
    =
    \rho \int_{K} f(\bfx ) \d \bfx .
\end{equation*}
In particular,
\begin{equation}
    \label{eq:mc_push_estimator}
    \widehat{I}_{ \Pi_\varepsilon \mathcal{P}}(f) \triangleq
    \frac{1}{\rho} \sum_{ \bfx  \in \Pi_\varepsilon \mathcal{P} } f(\bfx )
\end{equation}
is an unbiased estimator of
\begin{equation}
    \label{eq:target_integral}
    I_K(f) \triangleq \int_{K} f(\bfx ) \d \bfx .
\end{equation}
We shall also consider the so-called \emph{self-normalized} estimator
% version of $\widehat{I}_{ \Pi_\varepsilon \mathcal{P} }$ by $\widehat{I}_{s, \, \Pi_\varepsilon \mathcal{P} \cap K }$
\begin{equation}
    \label{eq:mc_push_estimator_self_normalized}
    \widehat{I}_{s, \, \Pi_\varepsilon \mathcal{P}\cap K  }(f) \triangleq
    \frac{|K|}{\Pi_\varepsilon \mathcal{P} (K)} \mathds{1}_{\{\Pi_\varepsilon \mathcal{P} (K)>0\}}  \sum_{ \bfx  \in \Pi_\varepsilon \mathcal{P} \cap K} f(\bfx ) ,
\end{equation}
where $\Pi_\varepsilon \mathcal{P} (K)$ is the number of points of $\Pi_\varepsilon \mathcal{P}$ in $K$.
Compared to \eqref{eq:mc_push_estimator}, \eqref{eq:mc_push_estimator_self_normalized} replaces $\rho$ by an unbiased estimator.
Self-normalized estimators are frequent in spatial statistics, and one can expect a (small) variance reduction in \eqref{eq:mc_push_estimator_self_normalized} at the cost  of a small bias.
%; see also Remark \ref{rmk:binomial_vs_poisson_mc}.
\begin{remark}
    \label{rmk:self_normalized_estimator}
    For a stationary point process $\calX$ that is almost surely valid, the self-normalized estimator $\widehat{I}_{s, \, \Pi_\varepsilon \calX \cap K  }(f)$ of $I_K(f)$ is biased.
    Indeed
    \begin{align*}
        \bbE \left[\widehat{I}_{s, \, \Pi_\varepsilon \calX}(f)\right]
        % & =
        % \bbE \left[\frac{|K|}{\Pi_\varepsilon \calX (K)} \sum_{ \bfx  \in \Pi_\varepsilon \calX \cap K} f(\bfx ) \mathds{1}_{\{\Pi_\varepsilon \calX (K)>0\}}\right] \\
        &=
        \bbE \left[\frac{|K|}{\Pi_\varepsilon \calX (K)} \mathds{1}_{\{\Pi_\varepsilon \calX (K)>0\}} \bbE \left[ \sum_{ \bfx  \in \Pi_\varepsilon \calX \cap K} f(\bfx ) \giventhat \Pi_\varepsilon \calX (K)\right] \right] .
    \end{align*}
    As $\Pi_\varepsilon \calX$ is a stationary point process by Proposition \ref{prop:motion_invariance}, once conditioning on $\Pi_\varepsilon \calX(K)$ each point of $\Pi_\varepsilon \calX \cap K$ is uniformly distributed over $K$.
    Let $(Y_i)_{i \geq 1}$ be random variables that follow the uniform distribution over $K$, we have
    \begin{align*}
        \bbE \left[\widehat{I}_{s, \, \Pi_\varepsilon \calX}(f)\right]
        &=
        \bbE \left[\frac{|K|}{\Pi_\varepsilon \calX (K)} \mathds{1}_{\{\Pi_\varepsilon \calX (K)>0\}} \sum_{ i=1}^{\Pi_\varepsilon \calX (K)} \bbE \left[  f(Y_i ) \right] \right]
        \\
        % &=
        % \bbE \left[\frac{|K| \Pi_\varepsilon \calX (K)}{\Pi_\varepsilon \calX (K)} \mathds{1}_{\{\Pi_\varepsilon \calX (K)>0\}} \bbE \left[ f(Y_1) \right] \right]
        % \\
        &=
        |K|\bbE \left[\mathds{1}_{\{\Pi_\varepsilon \calX (K)>0\}} \bbE \left[ f(Y_1) \right] \right]\\
        % &=|K|\bbE \left[ f(Y_1) \right] \bbP(\Pi_\varepsilon \calX (K)>0)
        % \\
        % &= \bbP(\Pi_\varepsilon \calX (K)>0) \int_{K} f(\bfx ) \d \bfx
        & = (1- \bbT_{\Pi_\varepsilon \calX}(K))\int_{K} f(\bfx ) \d \bfx,
    \end{align*}
    where the void probability $\bbT_{\Pi_\varepsilon \calX}$ is defined in \eqref{eq:def_void_proba}.
    As $K$ grows, the bias thus decreases.
    It is actually reasonable to expect that it vanishes exponentially fast with the size of $K$.
\end{remark}

Before we investigate the variance of linear statistics under $\Pi_\varepsilon \calP$, we need to ensure that the variance exists.
Actually, the following result ensures that the RPPP has moments of any order.
\begin{proposition}[Existence of the moments]
    \label{prop:extistance_of_the_moments}
    Let $d \geq 3$ and $\mathcal{P}$ be a homogeneous Poisson point process of intensity $\rho>0$ in $\mathbb{R}^d$.
    Let $\varepsilon\in (-1,1)$ and $R >0$.
    For any positive integer $m$
    \begin{equation*}
        \bbE
        \left[\left(\sum_{\bfx  \in \Pi_\varepsilon \mathcal{P}}
        \mathds{1}_{B(\mathbf{0},R)}(\bfx )\right)^m\right]
        < \infty.
    \end{equation*}
\end{proposition}
The proof is deferred to Appendix ~\ref{proof:exitance_moment}, and we note that a quantitative upper bound of the expectation can be deduced from the proof.
For now, a direct consequence of Proposition \ref{prop:extistance_of_the_moments} is that for any continuous function $f$ of compact support $K$, we have
\begin{align*}
    \Var
    \Big[\widehat{I}_{ \Pi_\varepsilon \mathcal{P}}(f)
    \Big]
    \leq
    \frac{\|f\|_{\infty}^2}{ \rho^2 }
    \bbE  \left[ \left(\sum_{\bfx  \in \Pi_\varepsilon \mathcal{P}} \mathds{1}_{K}(\bfx )\right)^2 \right] - I_K(f)^2
    < \infty .
\end{align*}
In the next section, we provide a more explicit expansion of the variance for small $\varepsilon$.

\subsection{Main result} % (fold)
\label{sub:Main result}

The following variance reduction result is the main theoretical finding of the present paper. Its proof is deferred to Appendix ~\ref{subs:proof_variance_reduction}.
\begin{theorem}[Variance reduction]
    \label{thm:Variance_reduction}
    Let $d \geq 3$, $\mathcal{P} \subset \bbR^d$ be a homogeneous Poisson point process of intensity $\rho > 1$, and let $f \in C^2(\mathbb{R}^d)$ have compact support $K$.
    For any $\delta \in (0, 1+\frac{1}{d-1})$, there exists a constant $C_{\delta} < \infty$ such that for all $\varepsilon\in (-1,1)$, we have
    \begin{equation}
        \label{eq:variance_linear_stat_push}
        \Var \left[ \widehat{I}_{ \Pi_\varepsilon \mathcal{P}  }(f)\right]
        \leq \Var\left[\widehat{I}_{ \calP}(f)\right] \left(1 - 2d \kappa_d \rho \varepsilon \right) + C_{\delta} \lvert\varepsilon\rvert^{\delta},
    \end{equation}
    where $C_{\delta}$ does not depend on the intensity $\rho$, $\widehat{I}_{ \Pi_\varepsilon \mathcal{P}  }(f)$ is defined in \eqref{eq:mc_push_estimator},
    \begin{equation}
        \label{eq:mc_poisson_estimator}
        \widehat{I}_{\mathcal{P}}(f) \triangleq
        \frac{1}{\rho} \sum_{ \bfx  \in \mathcal{P} } f(\bfx ),
    \end{equation}
    and $\kappa_d$ is the volume of the unit ball of $\mathbb{R}^d$. 
\end{theorem}

Several remarks are in order.

\begin{remark}
    \label{rk:variance_reduction_rmk}
    Upon noting that $\Pi_0 \calP= \calP$, Equation \eqref{eq:variance_linear_stat_push} implies a negative derivative of the variance of $ \widehat{I}_{ \Pi_\varepsilon \calP }(f) $ at $ \varepsilon =0$.
    Actually,
    \begin{equation}
        \Var \left[ \widehat{I}_{ \Pi_\varepsilon \calP }(f)\right]
        <
        \Var \left[\widehat{I}_{ \calP }(f)\right] = \rho^{-1} I_K(f^2),
    \end{equation}
    for a small enough stepsize $ \varepsilon >0 $.
    Computing the second-order derivative of the variance is more challenging because the moments of $F_{\calP}$ of order greater than $1+1/(d-1)$ are not well-defined.
\end{remark}

\begin{remark}
    \label{remk:epsilon_0}
    Taking $\varepsilon$ equal to
    \begin{equation}
        \label{eq:epsilon_0}
        \varepsilon_{0} \triangleq  \frac{1}{2d \kappa_d \rho}
    \end{equation}
    makes the term of order $\varepsilon$ in \eqref{eq:variance_linear_stat_push} vanish.
    Note that $\varepsilon_0$ does not depend on the integrand $f$.
    We shall later confirm in numerical experiments that this choice of $\varepsilon$ is a robust default. 
\end{remark}
\begin{remark}
    To gain insight into the rate of convergence behind Theorem~\ref{thm:Variance_reduction}, we set $\varepsilon = \varepsilon_0$ as in \eqref{eq:epsilon_0} and let $\delta = 1+1/(d-1) - \mu$ for some arbitrarily small $\mu \in (0, 1)$.
        This reduces \eqref{eq:variance_linear_stat_push} to  
        \begin{equation}\label{eq:var_with_epsilon_0}
            \Var \left[ \widehat{I}_{ \Pi_{\varepsilon_0} \mathcal{P}  }(f)\right] \leq \frac{\widetilde{C}_{\mu}}{\rho^{1+1/(d-1) - \mu}},
        \end{equation}
        where $\widetilde{C}_{\mu} < \infty$ is independent of $\rho$. Accordingly, repelling a Poisson point process with $\varepsilon = \varepsilon_0$ reduces the variance by a factor close to $\rho^{1/(d-1)}$. 
        Unlike recent approaches establishing optimal convergence rates for regular functions on compact domains, such as \citet{chopin2024higher} for $C^r$ regularity (with $r \in \mathbb{N}$) or \citet{leluc2025speeding} for Lipschitz functions, we do not recover here the optimal convergence rate, first derived by \citet{Bakh1959}; see also \citealp{Nov06,Nov2015}: for our $C^2$ regularity setting for $n$ evaluations of such an integrand, the worst-case squared error of the best algorithm scales as $n^{-1-4/d}$.
        Our rate is closer to the suboptimal rates $n^{-1-1/d}$ obtained recently for Monte Carlo methods using determinantal point processes \citep{BarHar2020,CoeMazAmb2021}.
        Actually, in the experiments of Section \ref{sec:experimental_illustration}, we even observe a slower rate of convergence than $1+1/(d-1)$. 
        This can be explained by two factors. First, border effects may lead to numerical degradation. Second, the constant $\widetilde{C}_{\mu}$ in \eqref{eq:var_with_epsilon_0} involves moments of order $\gamma$, with $d/(d-1) - \mu \leq \gamma < d/(d-1)$, of the $d/(d-1)$-stable force $F_{\calP} (\mathbf{0})$ (see, for instance, equation \eqref{eq:in_proofs_div_Cmu}). 
        Accordingly, $\widetilde{C}_{\mu}$ may diverge as $\mu \to 0$ at a rate $\mu^{-1}$, suggesting that polylogarithmic terms appear in the rate of convergence.
\end{remark}

\begin{remark}
    When $\varepsilon<0$, we obtain a positive first-order derivative of the variance at $\varepsilon=0$, so that, for $\vert\varepsilon\vert$ small enough,
    \begin{equation*}
        \Var \left[ \widehat{I}_{ \Pi_\varepsilon \mathcal{P} }(f)\right]
        >
        \Var \left[ \widehat{I}_{ \mathcal{P} }(f)\right].
    \end{equation*}
    This result is expected as the behavior of $\Pi_\varepsilon$ shifts from repulsive to attractive.
\end{remark}

% \begin{remark}
    %     Technical challenges arise as one tries to extend Theorem~\ref{thm:Variance_reduction} to discontinuous integrands, such as indicator functions. Furthermore, as demonstrated in the numerical experiment presented in Figure \ref{fig:variance_reduction} and analyzed in Section \ref{sec:experimental_illustration}, the variance behavior of the indicator function $f_2$ differs from that of the smooth $C^2$ functions $f_1$ and $f_3$ \eqref{eq:f_1_f_2_f_3}.
    % \end{remark}

\begin{remark}
    \label{rmk:harmonicity of the potential}
    A key element of the proof of Theorem \ref{thm:Variance_reduction} is the super-harmonicity
    of the Coulomb potential $U_\calP$, which defines the force function $F_\calP$.
    In other words, defining $U_\calP$ such that $\nabla U_\calP(\bfx) = F_{\calP} (\bfx)$,
    we have
    $$ \Delta  U_{\calP} (\bfx ) = \text{div}\left(F_{\calP \cap K} (\bfx )\right) = d \kappa_d \sum_{\bfz \in \calP\setminus \{\bfx \}} \delta_{\{\bfz\}}(\bfx ) - \kappa_d \rho,$$
    which is negative on $\bbR^d \setminus \calP$.
    This property, combined with a tailored integration by parts, forms the main ingredient of the proof; see Section \ref{sub:iterative Repulsion} and Remark \ref{rmk:remark_proof_singularity_importance}.
\end{remark}

\begin{remark}
    Without further assumptions on the integrand, other types of interaction than Coulomb do not necessarily yield such a variance reduction if plugged into our repulsion operator.
    Relatedly, there are many links between Coulomb interaction and numerical integration beside our result.
    For instance, the so-called Fekete points, defined as maximizers of the Coulomb energy
    \begin{align*}
        \bfx_{1}, \dots, \bfx_{N} \mapsto \sum_{1\leqslant i, j\leqslant N}^{\neq}\frac{1 }{\|\bfx_{i}-\bfx_{j}\|_2^{d-2}}
    \end{align*}
    on a compact, have been studied as a quadrature scheme, see e.g. \cite{Ser2019} and references therein.
\end{remark}
%------------------------------------
%-----------------------------------
% subsectionMain result (end)
\subsection{Properties of the force} % (fold)
\label{sub:Properties_of_F}
In this section, we discuss key characteristics of the random function $F_{\mathcal{P}}$, when $\calP$ is a PPP. Additional properties can be found in Appendix \ref{sec:Proofs}.

First, \citet[Proposition 1]{ChaPelal2010} proved that when $d\geq 3$, almost surely, the series defining $F_{\mathcal{P}}(\bfx)$ converges simultaneously for all $\bfx$ and defines a translation-invariant (in distribution) vector-valued random function, which is also almost surely continuously differentiable.
The subsequent proposition provides further insights into the distribution of $F_\calP$.

\begin{proposition}%[Distribution of the force]
    \label{prop:distribution_of_F}
    Let $\calP$ be a homogeneous Poisson point process of intensity $\rho$ of $\bbR^d$, with $d \geq 3$.
    Then, for any $\bfx \in \mathbb{R}^d$, $F_{\mathcal{P}}(\bfx )$ has a symmetric $\alpha-$stable distribution of index $\alpha= \frac{d}{d-1}$.
\end{proposition}

This observation was mentioned by \cite{ChaPelal2010} and can be easily checked by observing that the union of $n$ i.i.d.\ copies of $\mathcal{P} $, which is a PPP of intensity $n \rho$, is also a PPP of intensity $\rho$ scaled by $n^{-1/d}$.
The individual terms in $F_{\mathcal{P}}(\bfx)$ scale as a $(d-1)$-th power of the distance, so the sum of $n$ i.i.d.\ copies of $F_{\mathcal{P} }(\bfx)$ has the same distribution as $n^{(d-1)/d} F_{\mathcal{P} }(\bfx)$.
Symmetry is obvious, as $-\calP=\calP$ in distribution.
Proposition~\ref{prop:distribution_of_F} implies that $\bbE [F_{\mathcal{P} }(\bfx )]=0$ and $\bbE [ \lVert F_{\mathcal{P} }(\bfx ) \lVert_2^\nu ] < \infty$ iff $\nu<\alpha$.
For more details about stable distributions, we refer to
% \footnote{While the multivariate stable distribution is a significant theoretical concept as it extends the multivariate normal distribution, it is not widely employed in practical applications.
    % This is primarily because of its heavy-tailed behavior and the absence of a closed-form expression for its density function.}
\citep[Section 1.5]{Nol2020} and \citep{AbhNol1998}.

Second, we have the following result regarding the distribution of the difference of forces.
\begin{proposition}
    \label{prop:joint_density_of_force}
    Let $\calP$ be a homogeneous Poisson point process of $\bbR^d$.
    Then, for any two distinct points $\bfx, \bfy$ of $\bbR^{d}$, the random vector $F_{\calP} (\bfx)- F_{\calP} (\bfy)$ is continuous, i.e.,
    for any $\bfc \in \bbR^d$,
    $$\bbP
    \left(
    F_{\calP}(\bfx) - F_{\calP}(\bfy) = \bfc
    \right) =0 .$$
\end{proposition}
The proof of this proposition is deferred to Appendix \ref{sub:Existence of the joint distribution of vector of force}.
We note that an additional result regarding the joint density of the vector $(F_{\calP}(\bfx), F_{\calP}(\bfy))$ can be found in \citep[Theorem 10]{ChaPelal2010}, where the authors demonstrate the existence of the joint density of $(F_{\calP}(\bfx), F_{\calP}(\bfy))$ for $\bfx$ and $\bfy$ sufficiently far apart, further conditioning on having at least one point of $\calP$ within balls centered at $\bfx$ and $\bfy$.
They also derive an upper bound for the density.

Third, it is possible to derive an alternative expression for $F_{\mathcal{P}}(\bfx)$ that avoids the requirement of a different order of summation at each point $\bfx$. More precisely, ordering terms by their distance to the origin yields
\begin{equation}
    \label{eq:alternative_def_F}
    F_{\mathcal{P}}(\bfx )
    =
    \sum_{\substack{
            \mathbf{z} \in \mathcal{P} \setminus \{\bfx \}
            \\ \lVert \mathbf{z} \lVert_2 \uparrow }}
    \frac{\bfx  - \mathbf{z}}{ \lVert \bfx  - \mathbf{z} \lVert_2^d}
    - \kappa_d \rho \bfx ,
    \tag{$F_2$}
\end{equation}
where $\kappa_d$ is the volume of the unit ball of $\mathbb{R}^d$.
Note the additional term in \eqref{eq:alternative_def_F}, which compensates for fixing the order of summation.
\citet[Proposition 5]{ChaPelal2010} proved that the expressions \eqref{eq:def_gravitational_force} and \eqref{eq:alternative_def_F} are equivalent when $\mathcal{P}$ is a PPP of unit intensity.
A similar proof with slight modifications holds in the general case when $\rho \neq 1$.

Finally, for a stationary point process $\calX \subset \bbR^d$, and $0 \leq q < p $, we define the truncated force
\begin{equation}
    \label{eq:def_truncated_gravitational_force}
    F^{(q, p)}_{\calX }(\bfx )
    \triangleq
    \sum_{\substack{
            \mathbf{z} \in \calX \setminus \{\bfx \} \cap A^{(q,p)}(\bfx )
            \\ \lVert \bfx  - \mathbf{z} \lVert_2 \uparrow }}
    \frac{\bfx  - \mathbf{z}}{ \lVert \bfx  - \mathbf{z} \lVert_2^d} ,
\end{equation}
where $A^{(q,p)}(\bfx )= B(\bfx , p) \setminus B(\bfx , q)$ is the annulus centered at $\bfx $ with small radius $q$ and big radius $p$. We will denote $A^{(q,p)}(\mathbf{0})$ simply by $A^{(q,p)}$.
Intuitively, $F^{(q, p)}_{\mathcal{X}}(\bfx )$ represents the total Coulomb force experienced by a charged particle at $\bfx $ due to the influence of other particles of the same charge located in $\mathcal{X}\cap A^{(q,p)}(\bfx )$.
Note that the law of $F^{(q, p)}_{\mathcal{P}}(\bfx )$ is invariant under translation of $\bfx$, as was the case for its non-truncated counterpart.
% Further properties of $F^{(q, p)}_{\mathcal{P}}(\bfx )$ can be found in the Appendices.
The truncated force is a useful tool for practical implementation, just like the truncated repelled point process
\begin{equation}
    \label{eq:truncated_operator}
    \Pi_{\varepsilon}^{(q,p)}\calX \triangleq \left \{\bfx  + \varepsilon F_{\calX}^{(q,p)} (\bfx ): \bfx  \in {\calX} \right \}.
\end{equation}

\begin{remark}
    \label{rmk:truncated_var_reduc}
    The proof of Theorem \ref{thm:Variance_reduction} holds even when we replace $F_{\mathcal{P}}$ with its truncated version $F^{(0,p)}_{\mathcal{P}}$, as long as $p$ is larger than the diameter of the support $K$ of the integrand.
\end{remark}
% Moreover, if we denote by $\Pi^{(q,p)}_{\varepsilon}$ the repulsion operator resulting from plugging $F^{(q, p)}_{\mathcal{P}}$ instead of $F_{\mathcal{P}}$ into \eqref{eq:def_repulsion_operator},

%-----------------------------------------
%-----------------------------------------
\subsection{Sampling from the repelled Poisson point process} % (fold)
\label{sub:Sampling from the repelled point process}
\begin{figure}[!h]
    \centering
    \begin{subfigure}{0.24\linewidth}
        \centering
        \includegraphics[width=0.9\linewidth]{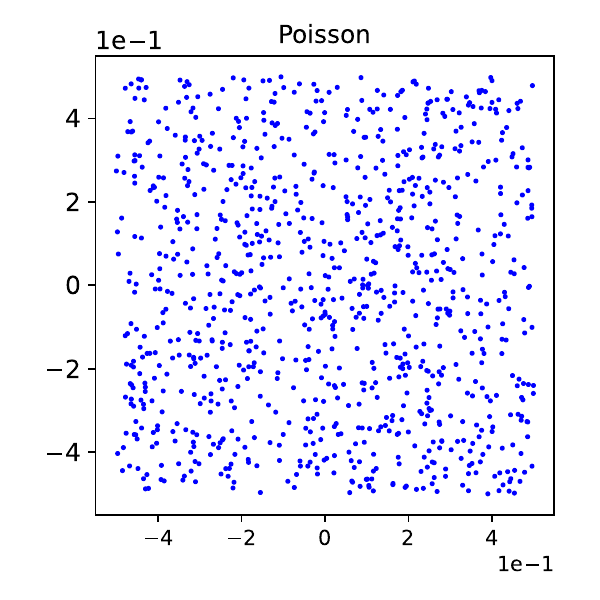}
    \end{subfigure}
    \begin{subfigure}{0.24\linewidth}
        \centering
        \includegraphics[width=0.9\linewidth]{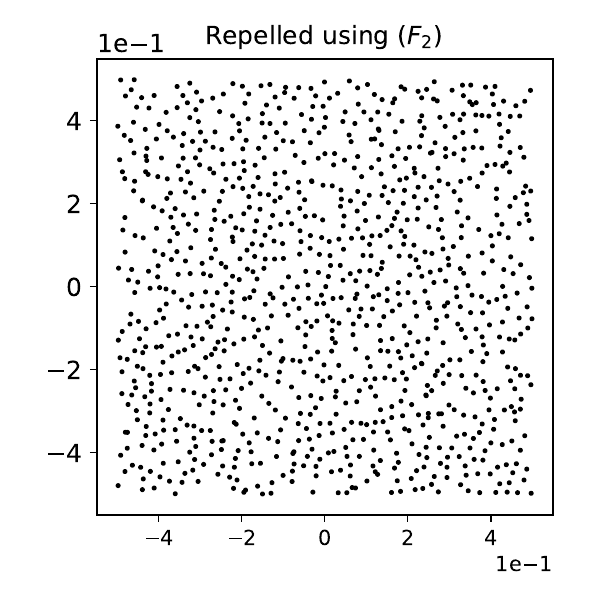}
    \end{subfigure}
    \begin{subfigure}{0.24\linewidth}
        \centering
        \includegraphics[width=0.9\linewidth]{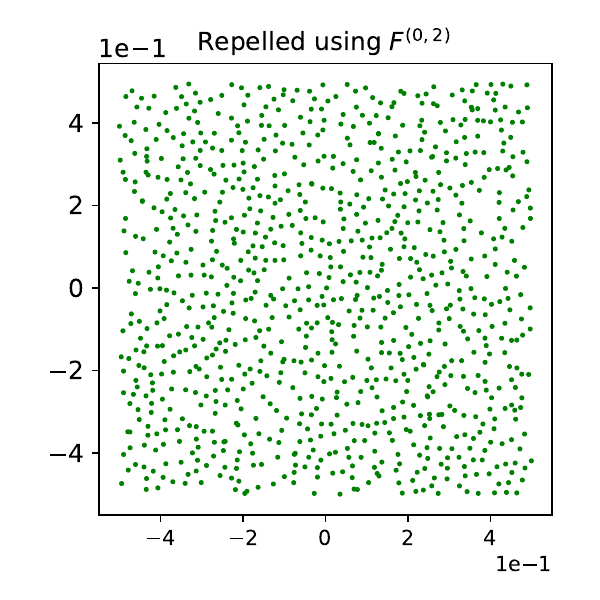}
    \end{subfigure}
    \begin{subfigure}{0.24\linewidth}
        \centering
        \includegraphics[width=0.9\linewidth]{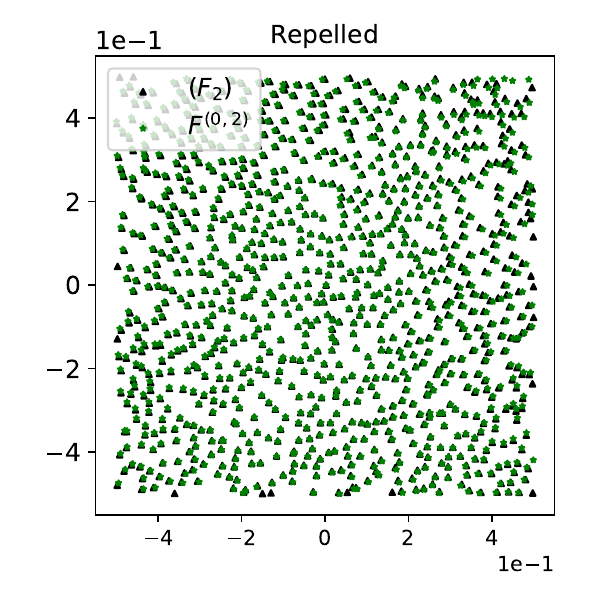}
    \end{subfigure}
    ~
    \begin{subfigure}{0.24\linewidth}
        \centering
        \includegraphics[width=0.9\linewidth]{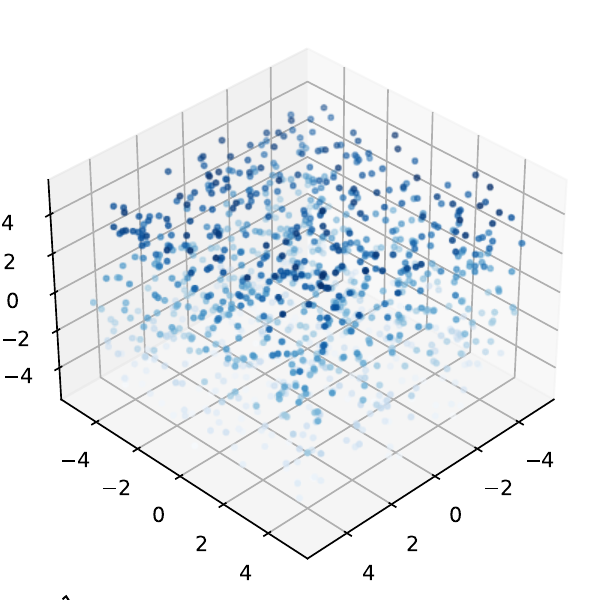}
    \end{subfigure}
    \begin{subfigure}{0.24\linewidth}
        \centering
        \includegraphics[width=0.9\linewidth]{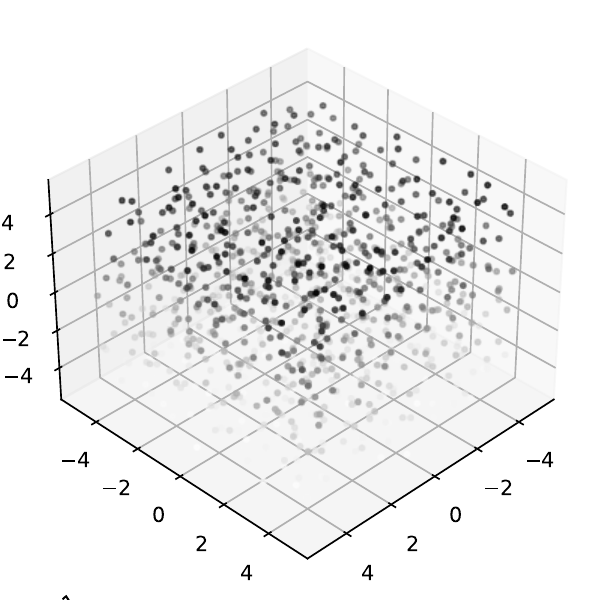}
    \end{subfigure}
    \begin{subfigure}{0.24\linewidth}
        \centering
        \includegraphics[width=0.9\linewidth]{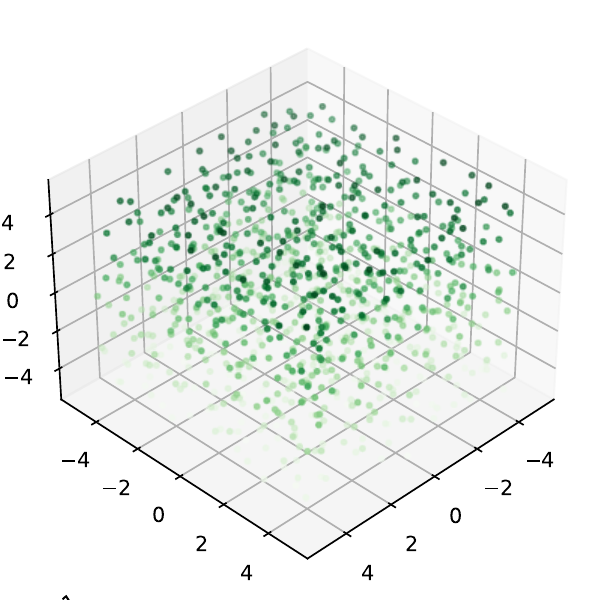}
    \end{subfigure}
    \begin{subfigure}{0.24\linewidth}
        \centering
        \includegraphics[width=0.9\linewidth]{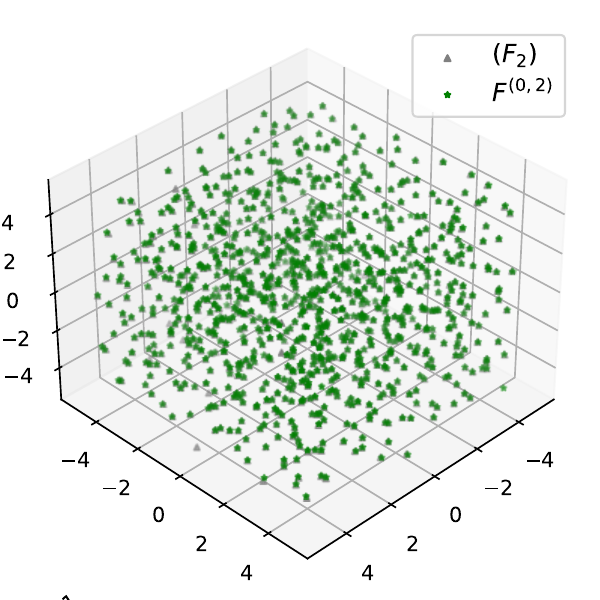}
    \end{subfigure}
    \caption{PPP sample (first column) and the corresponding RPPP samples, obtained using \eqref{eq:alternative_def_F} (second column) and using $F^{(0,2)}$ (third column).
        The first row corresponds to $d=2$ and the second row to $d=3$, with $\varepsilon$ set in each row to the value $\varepsilon_0 = \varepsilon_0(d)$ in \eqref{eq:epsilon_0}.
        The last column shows the two RPPP samples superimposed.}
    \label{fig:poisson_push_2d_3d}
\end{figure}
Let $\calP$ be a PPP of intensity $\rho>0$ and $ \Pi_\varepsilon \calP$ be the associated RPPP.
Let $K\subset\mathbb{R}^d$ be compact, with diameter $\text{diam}(K)$.
In this section, we propose two approaches to approximately sampling from $\Pi_\varepsilon \calP\cap K$.
By Proposition \ref{prop:motion_invariance_for_poisson}, $\Pi_\varepsilon \calP$ is stationary, and we henceforth assume that $K\subseteq B(\mathbf{0},p)$, where $p=\text{diam}(K)/2$.

\begin{algorithm}[h]
    \caption{: Sampling from the repelled point process corresponding to $\calP$ in a centered observation window $K$ using \eqref{eq:def_gravitational_force}}
    \label{algo:sampling_rppp_A1}
    \begin{algorithmic}[1]
        \State \textbf{fix} $\varepsilon$
        \State \textbf{set} $p = \text{diam}(K)/2$, and $k=1$
        \State \textbf{sample} a point pattern $\{\bfx_i\}_{i=1}^N$ from the point process $\calP$ in the centered ball window $B(\mathbf{0}, 2p)$
        \State \textbf{sort} $\{\bfx_i\}_{i=1}^N$ in a KD-tree
        \ParFor{$i: 1 \rightarrow N$}
        \If{$\bfx_i \in K$ }
        \State {\textbf{search} in the KD-Tree the points $\{\bfx_j\}_{j}$ located in $B(\mathbf{x_i}, p) \setminus \{\bfx_i\}$}
        \State{\textbf{use} $\{\bfx_j\}_{j}$ to compute the truncated force $F_{\calP}^{(0, p)} (\bfx_i) $ at $\bfx_i$ via Equation \eqref{eq:def_truncated_gravitational_force}}
        \State{\textbf{set} $\bfy_k = \bfx_i + \varepsilon F_{\calP}^{(0, p)} (\bfx_i)$, and $k=k+1$ }
        \EndIf
        \EndParFor
        \State \textbf{return}  $\{\bfy_k\}_k \cap K$
    \end{algorithmic}
\end{algorithm}

Our first approach is simply to sample $\Pi_{\varepsilon}^{(0,p)} \calP \cap K$.
This seems reasonable in the context of application to numerical integration since the variance reduction result also holds for $\Pi_{\varepsilon}^{(0,p)} \calP$; see Remark \ref{rmk:truncated_var_reduc}.
The corresponding pseudo-code is provided by Algorithm \ref{algo:sampling_rppp_A1}.
In words, we use the points of $\calP$ that fall in the larger ball $B(\mathbf{0},  2p)$ to displace the points of $\calP \cap K$.
Informally, for large $K$, we expect the resulting distribution to be close to that of $\Pi_\varepsilon \calP\cap K$ because, for each $\bfx\in\calP\cap K$, we only neglected contributions to the force \eqref{eq:def_gravitational_force} from points at distance further than $p$ from $\bfx$, and the magnitude of these contributions decreases fast.
One downside of this approach is that it requires, for each $\bfx$, to find the points of $\calP$ located in $B(\mathbf{x}, p)$.
While storing the initial sample of $\calP\cap B(\mathbf{0}, 2p)$ in an ad-hoc data structure like a KD-tree may help \citep{Ben1975}, we empirically found it more computationally tractable to rely on the alternative expression \eqref{eq:alternative_def_F} of the force.

Indeed, our second approach stems from the fact that the partial sums of \eqref{eq:alternative_def_F} use the same points of $\calP$, independently of $\bfx$.
The correction term in \eqref{eq:alternative_def_F} seems partially taking into account the effects of using a fixed order in the sum, so we propose to sample $ \calP\cap B(\mathbf{0}, p) $ and use the points in the latter sample to displace the points of $\calP \cap K$ using \eqref{eq:alternative_def_F} as described by Algorithm \ref{algo:sampling_rppp_A2}.
We still need to compute for each $\bfx \in \calP \cap K $ the distances between $\bfx$ and the points of $ \calP\cap B(\mathbf{0}, p) $, which in total cost $O(NM)$ operations, where $M = \calP(K)$ and $N= \calP(B(\mathbf{0}, p))$.

Some comments are in order.
First, the expected number of points of $\Pi_\varepsilon \calP \cap K$ is equal to $\rho |K|$.
Second, in both sampling approaches, the points of $\Pi_\varepsilon \calP \cap K$ can be sampled concurrently, resulting in a reduction in computational time roughly proportional to the number of available processors. This parallelization appears in Algorithm \ref{algo:sampling_rppp_A1}, and \ref{algo:sampling_rppp_A2} through a \emph{parfor} loop.
Third, currently, we do not have a strategy in place to mitigate border effects without increasing the computational cost. We recommend, if possible, sampling $\Pi_\varepsilon \calP \cap W$, where $W$ is a window slightly larger than $K$, and then restricting the obtained sample to the target window $K$.
In our upcoming experiments, we use $W=B(\mathbf{0}, \text{diam}(K)/2)$.
Finally, we provide a Python package, called \toolboxmcrppy{}, available on GitHub\footnote{\url{https://github.com/dhawat/MCRPPy}}, which implements Algorithm \ref{algo:sampling_rppp_A1} and \ref{algo:sampling_rppp_A2}.

Figure \ref{fig:poisson_push_2d_3d} shows samples of $\Pi_{\varepsilon_0} \calP$ in $[-1/2, 1/2]^d$ obtained with the two aforementioned approaches, for $d=2$ (first row) and $d=3$ (second row).
The corresponding PPP is of intensity $1000$, and the initial samples are given in the first column.
In the second column, \eqref{eq:alternative_def_F} was used (Algorithm \ref{algo:sampling_rppp_A2}), while $F^{(0,2)}_{\calP}$ was used in the third column (Algorithm \ref{algo:sampling_rppp_A1}).
The last column is a superposition of the samples obtained in columns 2 and 3, displaying very close agreement.
Finally, note that Figures \ref{fig:poisson_push_force} and \ref{fig:variance_reduction} were obtained using Algorithm \ref{algo:sampling_rppp_A2}, and we will keep using this simulation method in the next sections for sampling from the repelled point process of a stationary point process which may not necessarily be the PPP.
\begin{algorithm}[h]
    \caption{: Sampling from the repelled point process corresponding to $\calP$ in a centered observation window $K$ using \eqref{eq:alternative_def_F}}
    \label{algo:sampling_rppp_A2}
    \begin{algorithmic}[1]
        \State \textbf{fix} $\varepsilon$
        \State \textbf{set} $p = \text{diam}(K)/2$, and $k=1$
        \State \textbf{sample} a point pattern $\{\bfx_i\}_{i=1}^N$ from the point process $\calP$ in the centered ball window $B(\mathbf{0}, p)$
        \ParFor{$i: 1 \rightarrow N$}
        \If{$\bfx_i \in K$ }
        \State{\textbf{use} $\{\bfx_i\}_{i=1}^N$ to compute the force $F_{\calP} (\bfx_i) $ at $\bfx_i$ via Equation \eqref{eq:alternative_def_F}}
        \State{\textbf{set} $\bfy_k = \bfx_i + \varepsilon F_{\calP} (\bfx_i)$, and $k=k+1$ }
        \EndIf
        \EndParFor
        \State \textbf{return}  $\{\bfy_k\}_k \cap K$
    \end{algorithmic}
\end{algorithm}

\subsection{An experimental illustration of the variance reduction}
\label{sec:experimental_illustration}

\begin{figure}[!h]
    \centering
    \begin{subfigure}{\linewidth}
        \centering
        \includegraphics[width=0.99\linewidth]{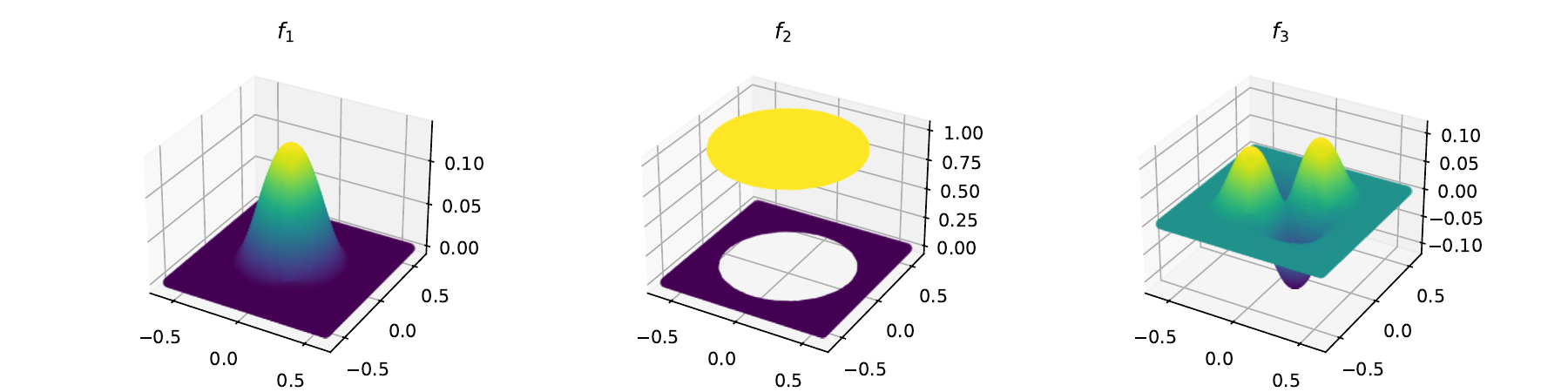}
    \end{subfigure}
    \begin{subfigure}{\linewidth}
        \centering
        \includegraphics[width=0.99\linewidth]{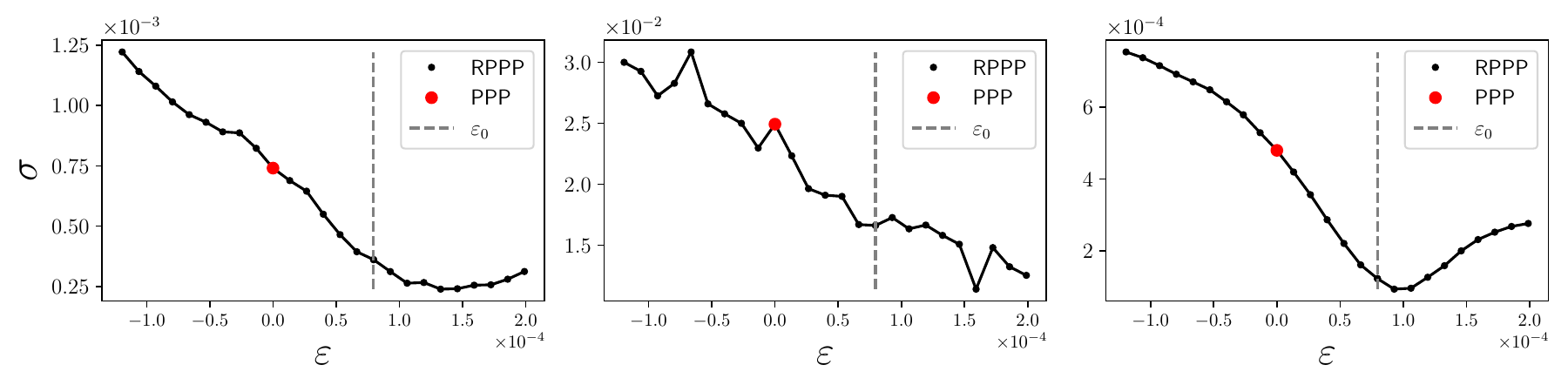}
    \end{subfigure}
    \caption{Estimated standard deviations of $\widehat{I}_{s, \Pi_\varepsilon \mathcal{P} \cap K }$ with respect to $\varepsilon$, for $f_1, \, f_2$, and $f_3$, in $d=3$.}
    \label{fig:variance_reduction}
\end{figure}

In this section, we present a numerical experiment to confirm the variance reduction found in Theorem \ref{thm:Variance_reduction}.
Additional experiments can be found in Section \ref{sec:Application to numerical integration}.

Let $K=[-1/2, 1/2]^d$.
Consider the three following integrands, all supported in $K$,
\begin{align}
    \label{eq:f_1_f_2_f_3}
    \nonumber
    f_1(\bfx ) &\triangleq  \left (1 - 4\|\bfx \|_2^2\right )^2 \exp\left(\frac{-2}{1 - 4\|\bfx \|_2^2}\right) \mathds{1}_{B(\mathbf{0}, 1/2)}(\bfx ),\\
    f_2(\bfx ) &\triangleq  \mathds{1}_{B(\mathbf{0}, 1/2)}(\bfx ),
    \quad
    \text{and } \quad
    f_3(\bfx ) \triangleq  \prod_{i=1}^d \cos^3(\pi x_i)\sin(\pi x_i) \mathds{1}_{K}(\bfx ).
\end{align}
Both $f_1$ and $f_3$ satisfy the requirements of Theorem \ref{thm:Variance_reduction}, while the indicator $f_2$ is discontinuous on $\partial B(\mathbf{0}, 1/2)$.
For each of these functions, Figure \ref{fig:variance_reduction} shows the estimated standard deviations $\widehat{\sigma}(\widehat{I}_{ s, \, \Pi_\varepsilon \mathcal{P} \cap K }(.))$ of the self-normalized estimator $\widehat{I}_{s, \, \Pi_\varepsilon \mathcal{P} \cap K }$ \eqref{eq:mc_push_estimator_self_normalized} for varying values of $\varepsilon$  in $d=3$ and for $\calP$ a PPP of intensity $\rho=500$ .
We conducted the analysis using $50$ independent samples of $\calP$.

% For a comprehensive explanation of the sampling technique employed, please refer to Section \ref{sub:Sampling from the repelled point process}.
The estimated standard deviations of $\widehat{I}_{s, \mathcal{P} \cap K }$, corresponding to $\varepsilon=0$, are represented by the large red dots. The black dots indicates the values of $\widehat{\sigma}(\widehat{I}_{s, \Pi_\varepsilon \mathcal{P} \cap K }(.))$.
The dashed lines indicate $\varepsilon_0$ \eqref{eq:epsilon_0}.
Note that within the range of $\varepsilon$ employed, the number of points of $(\Pi_\varepsilon \mathcal{P})\cap K$ remains relatively stable, with an average ranging between 493 and 501.

First, we observe that for negative values of $\varepsilon$, $\widehat{\sigma}(\widehat{I}_{s,\, \Pi_\varepsilon \mathcal{P} \cap K }(.))$ are greater than $\widehat{\sigma}(\widehat{I}_{s, \, \mathcal{P} \cap K }(.))$, for the three functions.
This behavior is expected because the operator $\Pi_\varepsilon$ is attractive for negative values of $\varepsilon$.
Second, for positive values of $\varepsilon$, up to $\varepsilon_0$, we observe that $\widehat{\sigma}(\widehat{I}_{s, \, \Pi_\varepsilon \mathcal{P} \cap K }(.))$ are lower than $\widehat{\sigma}(\widehat{I}_{s,\, \mathcal{P}\cap K }(.))$.
This result aligns with our theoretical expectations and provides evidence for the variance reduction in Theorem~\ref{thm:Variance_reduction}.
Third, for $f_1$ and $f_3$ we observe an interesting trend when $\varepsilon$ exceeds $\varepsilon_0$. The standard deviations decrease until reaching a minimum value. This minimum value, particularly for $f_3$, is relatively close to $\varepsilon_0$.
However, after this minimum point, the standard deviations start to increase again.
The behavior of $f_2$ in this scenario appears to be more intricate and less predictable.
Overall, it appears that $\varepsilon_0$ is a reasonable choice for $\varepsilon$, regardless of the integrand, although not necessarily the optimal threshold for a specific integrand.
%-----------------------------------
%-----------------------------------

\section{Application to numerical integration} % (fold)
\label{sec:Application to numerical integration}
In this section, we benchmark the RPPP among a few key Monte Carlo methods, to provide context.
The Python code for replicating this study can be found in \toolboxmcrppy{}\footnote{\url{https://github.com/dhawat/MCRPPy}}.
%--------------------------------------
\subsection{A few Monte Carlo methods} % (fold)
\label{sub:Monte Carlo methods}

Let $f$ be a continuous function supported in $K=[-1/2, 1/2]^d$.
Our goal is to estimate the integral $I_K(f)$ in \eqref{eq:target_integral}.
The simple (or crude) Monte Carlo method employs a Binomial point process (BPP) $\mathcal{B}_N$ supported on $K$, see Section \ref{sub:Poisson Point process}, to estimate $I_K(f)$ as follows
\begin{equation}
    \label{eq:MC}
    \widehat{I}_{\mathrm{MC}}(f) = \frac{1}{N} \sum_{\bfx  \in \mathcal{B}_N} f(\bfx ).
\end{equation}
The number of points used is fixed to $N$.
$\widehat{I}_{\mathrm{MC}}(f)$ is an unbiased estimator of $I_K(f)$ with a variance equal to $N^{-1} \Var (f(\bfu))$ where $\bfu$ is a uniformly drawn point of $K$.
\begin{remark}
    \label{rmk:binomial_vs_poisson_mc}
    As mentioned in Section \ref{sub:Poisson Point process}, when the number of points of a BPP increases appropriately with the size of the observation window, the BPP converges to a homogeneous Poisson point process (PPP).
    Consequently, one can think of the estimator $\widehat{I}_{\mathrm{MC}}(f)$ as a self-normalized version of the estimator $\widehat{I}_{\calP}(f)$ in \eqref{eq:mc_poisson_estimator}, where $\calP$ represents a PPP with suitable intensity.
    However, by setting the intensity of $\calP$ equal to $N$, we obtain the following variance for $\widehat{I}_{\calP}(f)$
    \begin{equation}
        \Var \left[\widehat{I}_{\calP}(f)\right] = N^{-1} \int_{K} f^2(\bfx) \d \bfx,
        \label{e:variance_poisson}
    \end{equation}
    which is larger than the variance of $\widehat{I}_{\mathrm{MC}}(f)$
    $$\Var \left[\widehat{I}_{\mathrm{MC}}(f)\right] =  N^{-1} \left(\int_{K} f^2(\bfx) \d \bfx - \left(\int_{K} f(\bfx) \d \bfx\right)^2\right).$$
    These variances are only equal if the integral of $f$ is zero.
    Otherwise, fixing the number of points is preferable to using a random number of points.
\end{remark}

Much research has gone into reducing the variance of $\widehat{I}_{\mathrm{MC}}$; see e.g. \citep{Owen2013} and references therein.
Control variate methods, for instance, rely on subtracting from the integrand $f$ a function that is computationally cheap to evaluate and possesses a known integral.
In our setting where the integrand is supported on the hypercube $K$ and the integral is w.r.t. the Lebesgue measure, of which we know all moments, we can easily derive the integral $m_\alpha$ of any multivariate monomial $\mathbf{x}^{\mathbf \alpha} = x_1^{\alpha_1}\dots x_d^{\alpha_d}$, and thus of any multivariate polynomial.
Then, for any $c\in\mathbb{R}$ and $b_{\alpha}\in\mathbb{R}, \alpha\in\mathbb{N}^d$, consider the estimator
\begin{equation}
    \label{eq:estimator_MCCV}
    \widehat{I}_{\mathrm{MCCV}}(f) = \widehat{I}_{\mathrm{MC}}\left(f - c \sum_{\vert\alpha\vert\leq 2} b_{\alpha} \mathbf{x}^\alpha\right) + c \sum_{\vert\alpha\vert\leq 2} b_{\alpha} m_\alpha,
    % = \frac1N \sum_{\bfx \in \mathcal{B}_N} (f(\bfx) - \tilde{h}(\bfx)) + I_K(\tilde{h}).
\end{equation}  
where we have arbitrarily limited ourselves to a multivariate polynomial of degree $2$ as control variate.   	
There are several ways to determine $c$ and $(b_\alpha)$ in \eqref{eq:estimator_MCCV}, we chose one that yields an unbiased estimator and relies on two additional independent copies of $\mathcal{B}_N$; see Appendix~\ref{a:comparison_of_CV_estimators} for details.
While using three independent copies of $\mathcal{B}_N$ intuitively gives an advantage to the MCCV estimator over estimators that use only $N$ samples from the uniform distribution, our point is to illustrate how our own (repelled, unbiased) estimators compare to an ideal control variates approach.
That being said, we provide a numerical comparison of the estimated MSEs of three natural variants of the CV-based estimator on our integration benchmark in Appendix~\ref{a:comparison_of_CV_estimators}.
The resulting mean-squared errors are essentially statistically indistinguishable in our experiments.

While a fixed number of control variates can help gain a constant in the mean-square error of a Monte Carlo estimator, some Monte Carlo methods focus on improving the \emph{rate} of convergence of the variance have been proposed, starting with grid-based stratification \citep[Chapter 10]{Owen2013}.
For a recent example, replacing the BPP of crude Monte Carlo with one of a particular family of determinantal point processes (DPPs) has been shown to enhance the convergence rate of the variance beyond $O(N^{-1})$; see \citep{BarHar2020, CoeMazAmb2021}.
For our comparison, we use the unbiased estimator given in \citep[Section 2.1]{BarHar2020} and implemented in \href{https://github.com/guilgautier/DPPy}{DPPy} \citep{DPPY},
\begin{equation}
    \label{eq:estimator_MCDPP}
    \widehat{I}_{\mathrm{MCDPP}}(f) = \sum_{\bfx  \in \mathcal{D}_N} \frac{f (\bfx /2 )}{2^d \mathtt{K}_N(\bfx ,\bfx )},
\end{equation}
where $\mathcal{D}_N$ is the multivariate Legendre ensemble and $\mathtt{K}_N$ its kernel.
According to \cite{BarHar2020}, if $f$ is $C^1$ and supported on an open set that is bounded away from the boundary of the hypercube $K$, $\widehat{I}_{\mathrm{MCDPP}}(f)$ is an unbiased estimator of $I_K(f)$ and its variance is $O(N^{-1-1/d})$, which is faster than the usual $O(N^{-1})$.
One of the limitations of DPP-based methods lies in the computational complexity associated with sampling from DPPs, which is at least cubic in the number $N$ of integrand evaluations.

We conclude with an instance of the Randomized Quasi-Monte Carlo method, which is an attempt at getting both the convenience of variance statements and the error reduction of stratified deterministic quadrature.
The estimator used is
\begin{equation}
    \label{eq:estimator_RQMC}
    \widehat{I}_{\mathrm{RQMC}}(f) = \frac{1}{N} \sum_{\bfx  \in \mathcal{S}_N} f(\bfx ),
\end{equation}
where $\calS_N$ is obtained by applying a suitable random perturbation to a low-discrepancy (deterministic) sequence \citep[Chapter 17]{Owen2013}.
In particular, each point of $\mathcal{S}_N$ is uniformly distributed in $K$, so that \eqref{eq:estimator_RQMC} is an unbiased estimator of $I_K(f)$.
Under strong  regularity assumptions on $f$ (at least all mixed partial derivatives of $f$ of order less than $d$ should be continuous on $K$), the variance of \eqref{eq:estimator_RQMC} is $O(\log(N)^{d-1}N^{-3})$ \cite[Theorem 17.5]{Owen2013}.
Despite the $d$-dependence of the rate and the strong smoothness assumptions, $\widehat{I}_{\mathrm{RQMC}}$ is remarkably efficient in small-to-moderate dimensions, and sampling is computationally cheap compared to, e.g., DPPs.
In this paper, we take $\mathcal{S}_N$ to be a scrambled Sobol sequence\footnote{We used the method \href{https://docs.scipy.org/doc/scipy/reference/generated/scipy.stats.qmc.Sobol.html}{Sobol} from the Python package \href{https://docs.scipy.org/}{scipy} to sample from $\mathcal{S}_N$ \citep{SciPy}.} \citep{Sob1967}.

Finally, note that we have classically ranked estimators here by the decay rate of their variance.
Yet the dependence of the variance to the dimension also includes the constants hidden in our asymptotic comparisons. 
For instance, if for each dimension, one considers an integrand that is lower bounded by a positive, dimension-independent threshold, then the variance \eqref{e:variance_poisson} scales as the volume of $K$ as the dimension grows and $N$ is kept fixed.
% subsectionSampling from  (end)
\subsection{Experimental comparison} % (fold)
\label{sub:Experminents}
This section focuses on examining and comparing the performance of $\widehat{I}_{\Pi_{\varepsilon} \calP}(f)$, for $\varepsilon = \varepsilon_0$ \eqref{eq:epsilon_0}, with the Monte Carlo estimators outlined in Section \ref{sub:Monte Carlo methods}.

While $\widehat{I}_{\mathrm{MC} }$, $\widehat{I}_{\mathrm{MCCV} }$, $\widehat{I}_{\mathrm{MCDPP} }$, and $\widehat{I}_{\mathrm{RQMC} }$ use a constant number of points across trials, $\widehat{I}_{\Pi_{\varepsilon_0} \calP}$ is the only method where the number of points is not fixed, neither in the computational budget nor in the (smaller) number of integrand evaluations.
In an effort to conduct a fair comparison, we replace the PPP in the estimator $\widehat{I}_{\Pi_{\varepsilon_0} \calP}$ with a BPP, which has a fixed number of points. The resulting estimator is referred to as $\widehat{I}_{\mathrm{MCRB} }$.
Next, we sample $M$ realizations from the repelled point processes using Algorithm \ref{algo:sampling_rppp_A2} as described in Section~\ref{sub:Sampling from the repelled point process}, and find the average number of points obtained within the $M$ trials.
We set the number $N$ of points used in the other methods to this average.
%Note that we give an unfair advantage to the estimator $\widehat{I}_{\mathrm{MCCV} }$ by not accounting for the evaluations of the integrand necessary to estimate the coefficient \eqref{eq:mccv_coefficient}.

We use $M=100$ samples from each of the point processes and we make sure $N$ ranges roughly from 50 to 1000.
We examine the functions $f_1$, $f_2$, and $f_3$ defined in Section \ref{sec:The repelled Poisson point process} by Equation~\eqref{eq:f_1_f_2_f_3}.
Their integrals are
$$
I_{K}(f_2) = 2^{-d} \kappa_d ,
\quad\quad
I_{K}(f_3) =  0,
$$
while the precise value of $I_{K}(f_1)$ is unknown.
\begin{figure}[h!]
    \centering
    \begin{subfigure}{\linewidth}
        \centering
        \includegraphics[width=0.8\linewidth]{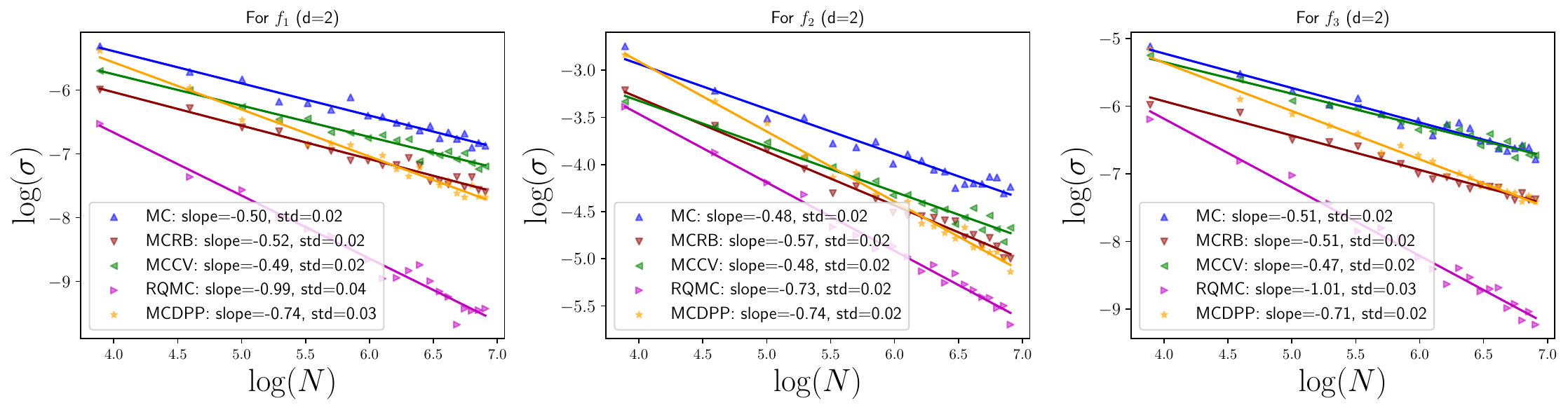}
    \end{subfigure}
    \begin{subfigure}{\linewidth}
        \centering
        \includegraphics[width=0.8\linewidth]{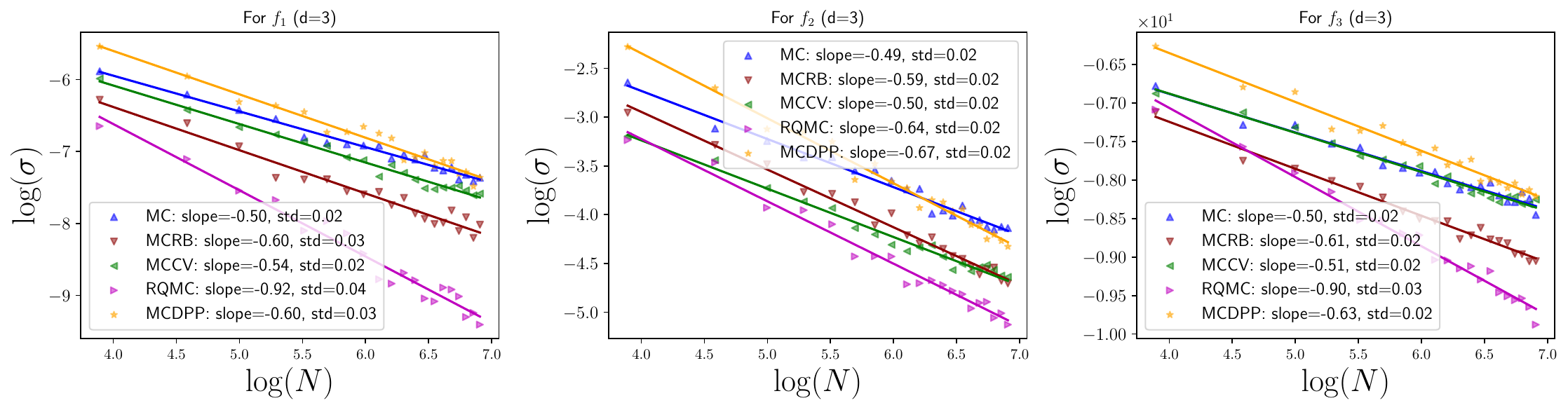}
    \end{subfigure}
    \begin{subfigure}{\linewidth}
        \centering
        \includegraphics[width=0.8\linewidth]{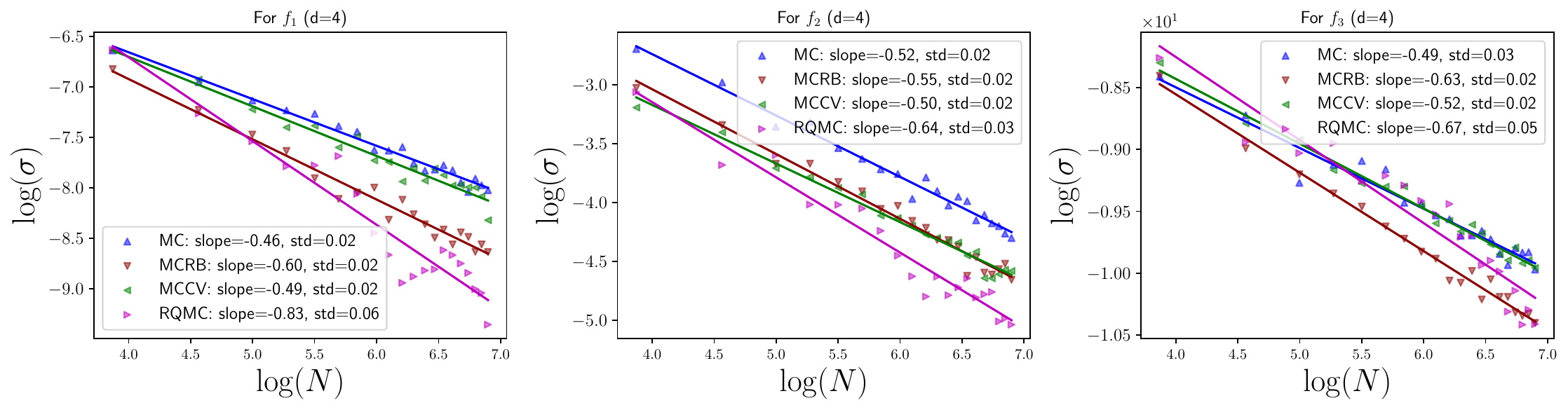}
    \end{subfigure}
    \begin{subfigure}{\linewidth}
        \centering
        \includegraphics[width=0.8\linewidth]{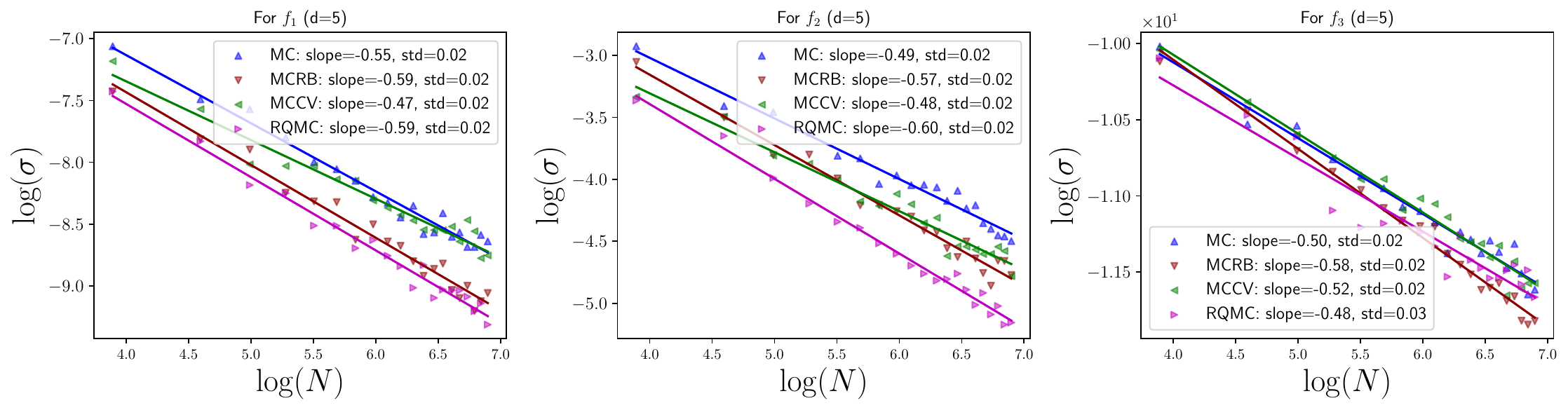}
    \end{subfigure}
    \begin{subfigure}{\linewidth}
        \centering
        \includegraphics[width=0.8\linewidth]{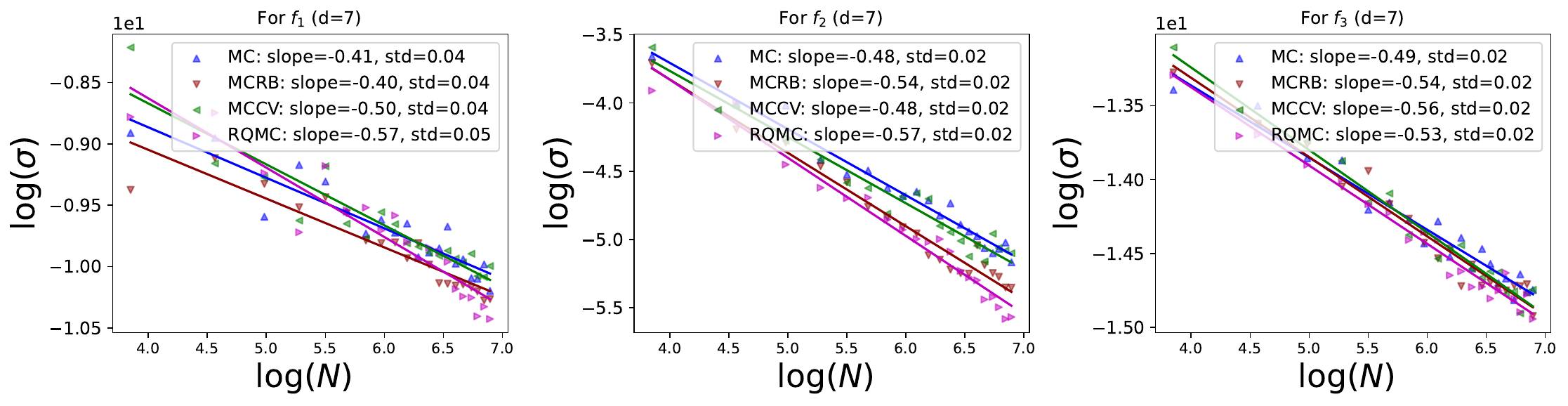}
    \end{subfigure}
    \caption{Estimated standard deviations of various Monte Carlo methods for $f_1$, $f_2$, and $f_3$ across different dimensions $d \in \{2,3, 4, 5,7\}$.}
    \label{fig:std_mc}
\end{figure}

Figure \ref{fig:std_mc} displays the log of the estimated standard deviation of the estimators $\widehat{I}_{\mathrm{MC} }$, $\widehat{I}_{\mathrm{MCRB} }$, $\widehat{I}_{\mathrm{MCCV} }$, and $\widehat{I}_{\mathrm{RQMC} }$, plotted against the log number of points $\log(N)$, for the functions $f_1$, $f_2$, and $f_3$ and $d$ in $\{2, 3, 4, 5, 7\}.$
$\widehat{I}_{\mathrm{MCDPP} }$ is only examined for $d\in \{2, 3\}$ due to its high cost compared to other methods.
Lines correspond to ordinary least square linear regressions (OLS).
The slopes and standard deviations of the slopes are indicated in the legend.

First, as expected, the estimated variances of $\widehat{I}_{\text{MCRB}}$ are generally lower than those of $\widehat{I}_{\text{MC}}$.
$\widehat{I}_{\mathrm{MCRB} }$ outperforms $\widehat{I}_{\mathrm{MC} }$ in most scenarios, except for $f_3$ with $d=7$ where the estimated variances of $\widehat{I}_{\text{MC}}$ and $\widehat{I}_{\text{MCRB}}$ are comparable. Interestingly, in this case, the estimated variances of all methods are in the same ballpark.
Second, $\widehat{I}_{\mathrm{MCRB} }$ outperforms $\widehat{I}_{\mathrm{MCDPP} }$ in dimension $3$.
Although the convergence rate for the variance of $\widehat{I}_{\mathrm{MCDPP} }$ is faster than $N^{-1}$, it seems that the hidden factor in the variance convergence rate increases with the dimension.
Moreover, $\widehat{I}_{\mathrm{MCDPP} }$ is computationally demanding.
Fourth, it appears that $\widehat{I}_{\mathrm{MCRB}}$ outperforms $\widehat{I}_{\mathrm{MCCV}}$ in most cases.
$\widehat{I}_{\mathrm{MCRB}}$ and $\widehat{I}_{\mathrm{RQMC}}$ seem the main competitors.

$\widehat{I}_{\mathrm{RQMC}}$ consistently performs well.
For $d\leq 3$, the estimated variances of $\widehat{I}_{\text{RQMC}}$ are lower than those of $\widehat{I}_{\text{MCRB}}$ for $N$ large enough, and the slope of the variance of $\widehat{I}_{\text{RQMC}}$ is steeper than that of $\widehat{I}_{\text{MCRP}}$.
%However, $\widehat{I}_{\mathrm{RQMC}}$'s performance decreases significantly as $d$ increases.
Interestingly, for $f_3$ and $d =4$, the estimated variances of $\widehat{I}_{\text{RQMC}}$ are larger than those of $\widehat{I}_{\text{MCRB}}$, but the slope for $\widehat{I}_{\text{RQMC}}$ remains steeper than for $\widehat{I}_{\text{MCRB}}$, letting $\widehat{I}_{\text{RQMC}}$ catch up as $N$ grows large.
The same trend is observed for $f_1$ in $d=7$ up to a large value of $N$, while the opposite trend is observed for $f_3$ in $d=5$.
When comparing the estimated variances of $\widehat{I}_{\text{RQMC}}$ and $\widehat{I}_{\text{MCRB}}$ for $f_1$ and $f_2$ across different dimensions, it appears that $\widehat{I}_{\mathrm{RQMC}}$'s performance declines more rapidly than $\widehat{I}_{\mathrm{MCRB}}$'s as the dimension $d$ increases.

Finally, we conducted additional analyses in Appendix \ref{app:analysis_of_the_experiment_4.2}
and computed confidence intervals for the slopes corresponding to $\widehat{I}_{\text{MCRB}}$ illustrated in Figure \ref{fig:std_mc}.
We also compared the error of the Monte Carlo estimators under investigation.
In short, data suggests that the slope is the usual Monte Carlo rate or slightly faster. Additionally, there is no substantial evidence of bias in $\widehat{I}_{\text{MCRB}}$.

%---------------------------------------
\section{Other models and properties} % (fold)
\label{sec:ther models and properties}

This section contains numerical investigations of several intriguing questions that arise from the repulsion operator.
In Section~\ref{sub:iterative Repulsion}, we analyze the behavior of a PPP when the repulsion operator is applied several times.
Then, in Section~\ref{sub:Second-order properties of the repelled Poisson point process}, we estimate the second-order properties of the RPPP.
Finally, in Section~\ref{sub:Other models}, we explore the behavior of the repulsion operator when applied to two point processes that are more regular than the PPP.

\subsection{Iterating the repulsion} % (fold)
\label{sub:iterative Repulsion}
\begin{figure}[h!]
    \centering
    \includegraphics[width=0.4\linewidth]{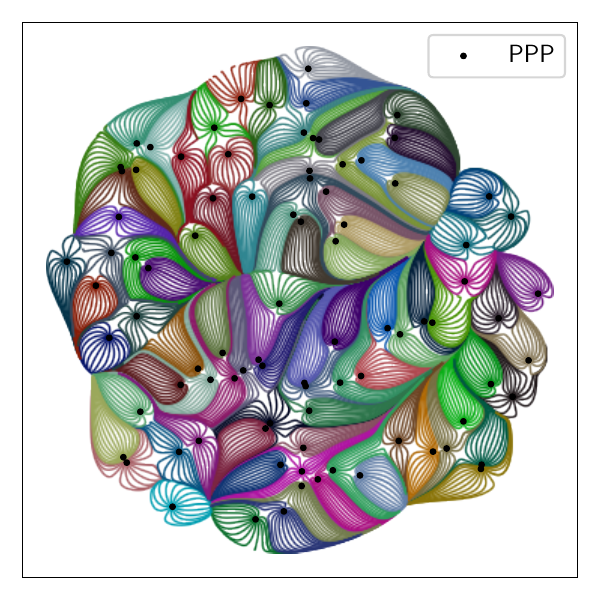}
    \caption{Illustration of the gravitational allocation from Lebesgue to a realization (black points) of a PPP in a disk.
        Each set of curves sharing the same color illustrates a gravitational cell, which indicates the points of the space allocated to the point of the PPP that belongs to that particular colored region.
        The code to generate this picture can be found in \toolboxmcrppy{}.
        Note that the image is purely for illustrative purposes, and no claims are being made regarding the existence of a gravitational allocation from the Lebesgue measure to a PPP in dimension $d=2$.
    }
    \label{fig:gravitational_alloc_to_poisson}
\end{figure}
Our repulsion operator $\Pi_{\varepsilon}$ in \eqref{eq:def_repulsion_operator} can be seen as performing one step of a discretization scheme for a system of differential equations describing gravitational allocation.
In this section, we carry the analogy with gravitational allocation further by iterating the application of $\Pi_{\varepsilon}$.

To provide more context, consider a PPP $\calP$ of unit intensity, let $\bfx \in \bbR^d \setminus \calP$, and consider the differential equation
\begin{equation}
    \label{eq:diff_eq_grav_alloc}
    \frac{\partial}{\partial t} Y_{\bfx }(t) = - F_{\calP}(Y_{\bfx }(t)) \quad \text{with} \quad Y_{\bfx }(0) = \bfx .
\end{equation}
The solution $t\mapsto Y_{\bfx}(t)$ of the differential equation \eqref{eq:diff_eq_grav_alloc}, defined up to some positive time $\tau_{\bfx} \in (0, \infty]$, is called a flow curve of the gravitational allocation of $\calP$ to the Lebesgue measure \citep{ChaPelal2010}; see Figure~\ref{fig:gravitational_alloc_to_poisson} for an illustration of the gravitational allocation of a PPP of unit intensity.
Remark that for $\varepsilon <0$, $\bfx + \varepsilon F_{\calP}(\bfx)$ is the first step in a naive numerical scheme discretizing the differential equation \eqref{eq:diff_eq_grav_alloc}, with a stepsize equal to $-\varepsilon$.
Similarly, each point of $\Pi_{\varepsilon} \calP$ can be viewed as the initial discretization step of a differential equation akin to \eqref{eq:diff_eq_grav_alloc}, with stepsize $-\varepsilon$.
The catch is that $\Pi_{\varepsilon} \calP$ is applied to the points of $\calP$ itself, not to points in $\bbR^d \setminus \calP$ as in \eqref{eq:diff_eq_grav_alloc}.
Loosely speaking, the image of $\bfx \in \calP$ in $\Pi_{\varepsilon} \calP$ can be seen as a first step of the numerical discretization of a flow curve, in a gravitational allocation from the reduced Palm measure of $\calP$ to the Lebesgue measure. However, to fully formalize and understand this allocation, further technical details are required and are out of the scope of this paper.

Consider now performing $M$ steps of the same numerical scheme, i.e., for each $\bfx_0\in \calP$, define
\begin{equation}
    \label{e:iterates}
    \bfx_{t} = \bfx_{t-1} + \varepsilon F_{\calP}(\bfx_{t-1}), \quad t=1, \dots, M.
\end{equation}
Call $\Pi_{\varepsilon,t} \calP$ the set of $t-$th iterates \eqref{e:iterates} of the points of $\calP$.
\begin{figure}[h!]
    \centering
    \begin{subfigure}{0.45\linewidth}
        \centering
        \includegraphics[width=1.0\linewidth]{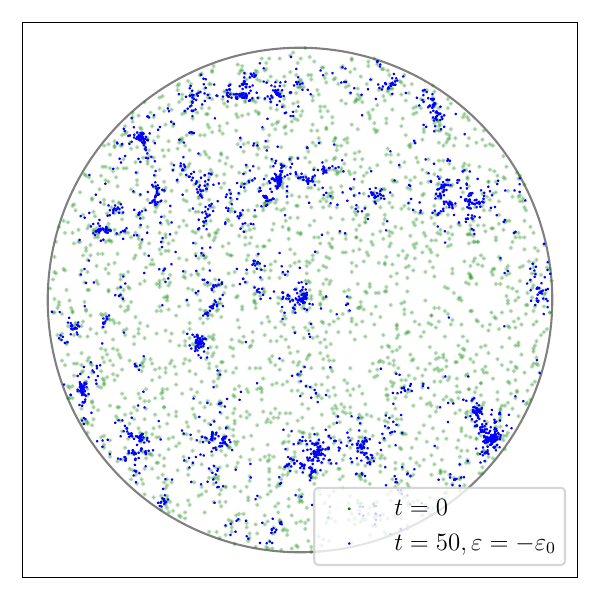}
    \end{subfigure}
    \begin{subfigure}{0.45\linewidth}
        \centering
        \includegraphics[width=1.0\linewidth]{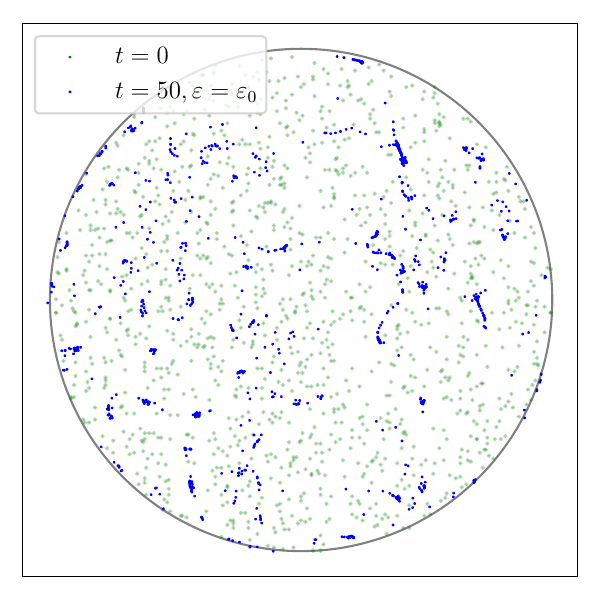}
    \end{subfigure}
    \caption{The green points represent a sample from a PPP $\calP$ of unit intensity, the blue points correspond to $\Pi_{-\varepsilon_0, 50}\calP$ (left), and $\Pi_{+\varepsilon_0, 50} \calP$ (right).}
    \label{fig:exteme_repelled}
\end{figure}
We expect that for $\varepsilon <0$ and large $t$, the distribution of the points of $\Pi_{\varepsilon,t}\calP$ will not be more regular than $\calP$.
This assertion is supported by the clustered arrangement of the points of $\Pi_{-\varepsilon_0, 50}\calP$ observed in the left panel of Figure~\ref{fig:exteme_repelled}; see also Figure ~\ref{fig:extreme_repelled_additional}.
However, it is important to note that some points of $\Pi_{-\varepsilon_0, 50}$ are situated outside the observation window and are thus not visible in Figure~\ref{fig:exteme_repelled}.

Following the arguments of \cite{ChaPelal2010}, one can prove that the differential equation \eqref{eq:diff_eq_grav_alloc} defines an allocation rule.
So, in particular, for almost any $\bfx \in \bbR^d\setminus \calP$, each curve $Y_{\bfx}(t)$ will eventually end at a point of $\calP$, almost surely, as $t $ goes to $\tau_{\bfx}$.
Similarly, in an ``ideal'' discretization scenario, we would expect that $\bfx_{t}$ ends at a point of $\calP$ as $t \rightarrow \tau_{\bfx_0}$.
However, during our experiments, we observed that certain points have moved away from the observation window and the remaining points clustered together within it.
The movement of certain points away from the observation window can be attributed to the naive discretization scheme \eqref{e:iterates}, where we employed a fixed stepsize $\varepsilon$, while the points of $\calP$ are singular points of $F_{\calP}$.
As a result, when an $\bfx_{t}$ is in close proximity to a point in $\calP$, the force acting on it becomes considerable, which might project  $\bfx_{t+1}$ outside the observation window, even if the initial gradient direction is towards a pont of $ \calP$.
Applying a truncated version of the force might help prevent this phenomenon.

The scenario where $\varepsilon > 0$ and $ M$ is large is less straightforward.
This case can be associated with the reverse dynamics of Equation \eqref{eq:diff_eq_grav_alloc}.
Indeed, for $\bfx \in \bbR^d \setminus \calP$, the trajectory of $Y_{\bfx}(t)$ halts at some point of the boundary of an allocation cell of the gravitational allocation from Lebesgue to $\calP$; see Figure \ref{fig:gravitational_alloc_to_poisson}.
By the same analogy as before, we can expect that for $\bfx_0 \in \calP$ and $\varepsilon$ small enough, as $t \rightarrow \infty $, each point $\bfx_{t}$ will eventually approach the boundary of an allocation cell of the gravitational allocation from Lebesgue to the reduced Palm measure of $\calP$.
It is hard to dig deeper without further investigation.
Interestingly, in the right panel of Figure~\ref{fig:exteme_repelled}, we observe a peculiar clustering behavior of the points in $\Pi_{ \varepsilon_0, 50}\calP$ characterized by points appearing to overlap or superimpose with each other; see also Figure ~\ref{fig:extreme_repelled_additional}.
%----------------------------------
%----------------------------------
\subsection{Second-order properties of the repelled Poisson point process} % (fold)
\label{sub:Second-order properties of the repelled Poisson point process}
The second-order characteristics of a point process, such as the pair correlation function and the structure factor, offer valuable insights into the regularity of the point process \citep{BadRubTur2013}, variance reduction \citep{Pilleboue+al:2015}, and hyperuniformity \citep{Cos2021,Tor2018}.
In this section, we estimate the pair correlation function and the structure factor of the RPPP.

\begin{figure}[h!]
    \centering
    \includegraphics[width=1.0\linewidth]{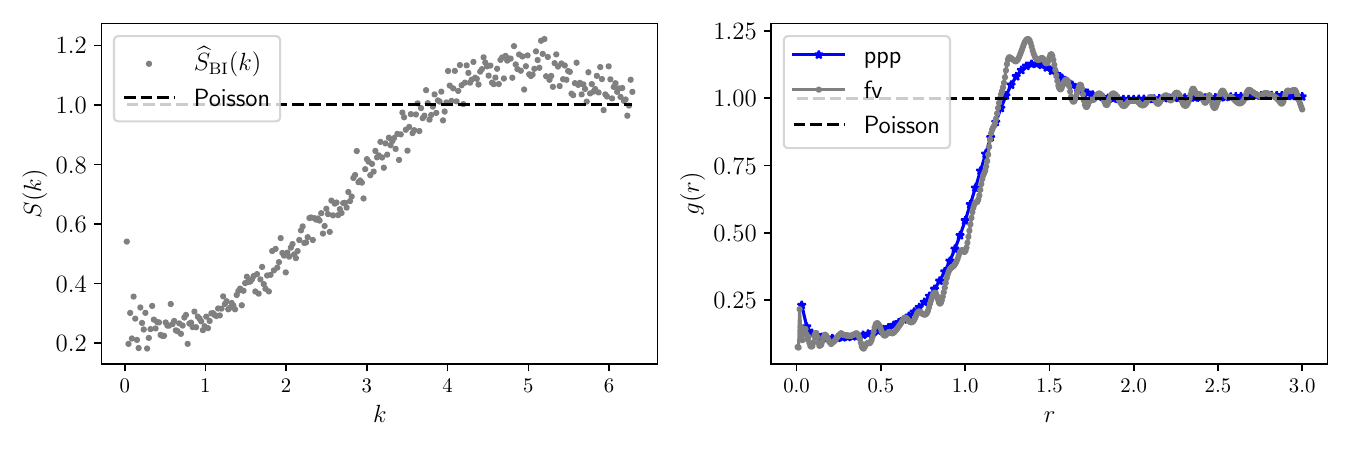}
    \caption{ Estimated structure factor (left) and pair correlation function (right) of an RPPP of $\bbR^2$ of intensity $1/\pi$.
        The labels `ppp' and `fv' correspond to two different estimators of the pair correlation function; see e.g. \citep{HawAl2023}.}
    \label{fig:pcf_and_sf}
\end{figure}

The pair correlation function $g$ of a stationary point process $\calX$ of intensity $\rho>0$ is a function characterizing the probability of finding a pair of points of the point process at a certain distance. It is important to note that not all point processes have a pair correlation function in the strict sense, as it can sometimes be a measure instead. This is the case for a perturbed lattice, for example. However, if $g$ exists, it is given by
\begin{equation*}
    \bbE
    \left[
    \sum_{\bfx  , \bfy \in \calX }^{\neq}
    h(\bfx , \bfy)
    \right] \triangleq
    \int_{\mathbb{R}^d \times \mathbb{R}^d} h(\bfx  + \bfy, \bfy)\rho^{2} g(\bfx ) \d\bfx  \d\bfy,
\end{equation*}
for any nonnegative measurable bounded function $h$ with compact support.
When $g-1$ is integrable we can further define the structure factor $S$ of $\calX$ by
\begin{equation}
    \label{eq:stucture_factor_function}
    S(\mathbf{k}) = 1 + \rho \mathcal{F} (g-1)(\mathbf{k}),
\end{equation}
where $\mathcal{F}$ denotes the Fourier transform.
Assuming that $\calX$ is also isotropic\footnote{
    The assumptions of stationarity and isotropy, in this case, can be straightforwardly weakened to assuming that the intensity measure is invariant to translations or that the pair correlation function only depends on the inter-point distance.
},
both $g$ and $S$ become radial functions, meaning that $g(\mathbf{r})$ and $S(\mathbf{k})$ depend only on the norm of $\mathbf{r}$ and $\mathbf{k}$ respectively. In that case, we abusively write $g(\mathbf{r}) = g(r)$, and $S(\mathbf{k}) = S(k)$ where $x\triangleq \| \bfx \|_2$.
For a PPP, the structure factor and the pair correlation function are both constant and equal to 1.
A pair correlation function $g(r)$ below 1 indicates repulsion between the points of the corresponding point process at scale $r$, whereas a pair correlation function greater than 1 is a sign of attraction. The structure factor, on the other hand, is used as an indicator of hyperuniformity: $S(\mathbf{0})=0$ implies that the point process is hyperuniform, that is, the variance of the number of points that fall in a ball grows slower than the volume of that ball.
For the RPPP, with $\varepsilon= \varepsilon_0$ we expect to observe $g(r)$ less than 1 for a certain range of $r$, as the repulsion operator introduces repulsion between the points.
However, having intuition on hyperuniformity is more difficult, and a statistical diagnostic using an estimator of $S$ is required.

Figure \ref{fig:pcf_and_sf} shows the estimated pair correlation function and structure factor of an RPPP of intensity $1/\pi$ in dimension $d=2$.
The estimations were obtained using a sample observed within a centered ball of radius $r=150$.
Bartlett's isotropic estimator $\widehat{S}_{\mathrm{BI}}$ was used to estimate the structure factor on the corresponding set of allowed wavenumbers $k$; see \citep[Sections 3.2.1 and 5.3]{HawAl2023}.
The estimators \href{https://www.rdocumentation.org/packages/spatstat.core/versions/2.1-2/topics/pcf.ppp}{pcf.ppp}, and \href{https://www.rdocumentation.org/packages/spatstat.core/versions/2.1-2/topics/pcf.fv}{pcf.fv} \citep[Sections 7.4.4 and 7.4.5]{BadRubTur2013} of the R library \href{https://spatstat.org/}{spatstat} were used to estimate the pair correlation function through the Python package \href{https://github.com/For-a-few-DPPs-more/structure-factor}{structure-factor}; see \citep[Section 5.3]{HawAl2023}.
Dashed black lines represent the structure factor and pair correlation function of the PPP.
As expected, the estimated values of $g$ are smaller than $1$ up to a certain value of $r$, indicating repulsion between the points at small scales.
However, $\widehat{S}_{\mathrm{BI}}(k)$ is greater than $0.2$ at small $k$, suggesting that the RPPP may not be hyperuniform.

%---------------------------------
%---------------------------------
\subsection{Other repelled point processes} % (fold)
\label{sub:Other models}
\begin{figure}[h!]
    \centering
    \begin{subfigure}{0.45\linewidth}
        \centering
        \includegraphics[width=1.0\linewidth]{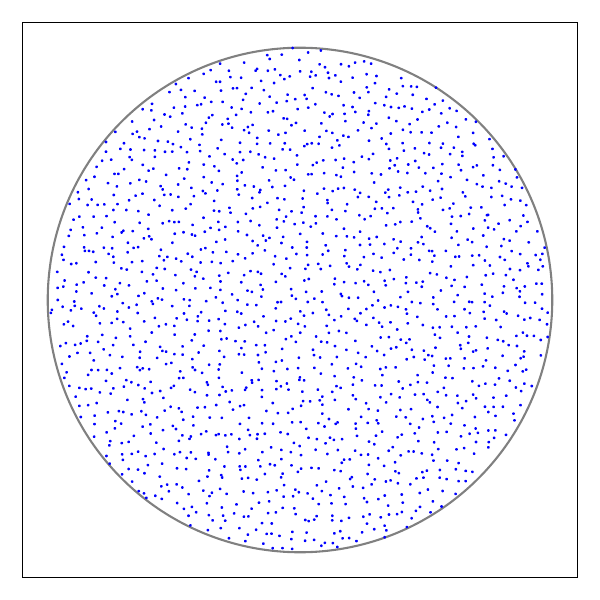}
        \caption{GPP}
    \end{subfigure}
    \begin{subfigure}{0.45\linewidth}
        \centering
        \includegraphics[width=1.0\linewidth]{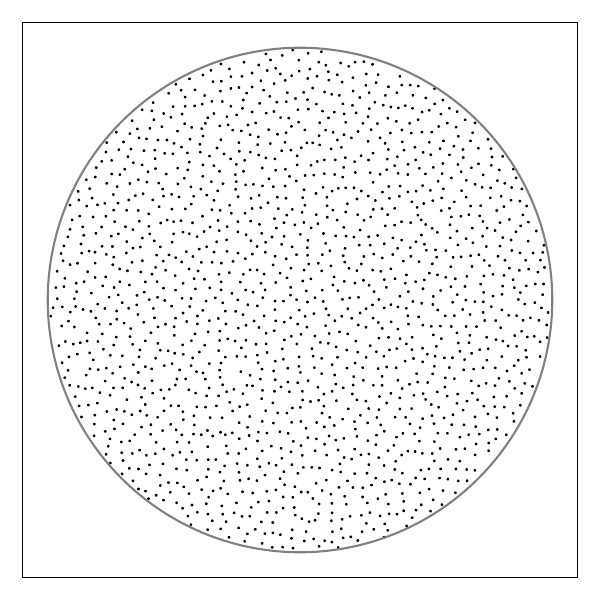}
        \caption{RGPP}
    \end{subfigure}
    \caption{A sample from the Ginibre ensemble (left) and the obtained repelled sample (right) with $\varepsilon=\varepsilon_0$.}
    \label{fig:repelled_ginibre}
\end{figure}
In this section, we perform numerical investigations to determine whether the variance reduction identified in Theorem \ref{thm:Variance_reduction} remains valid when the initial points are more evenly distributed than a PPP, such as in the case of the Ginibre ensemble (GPP) or a scrambled Sobol sequence (SSS).

The GPP is a motion-invariant point process of $\mathbb{R}^2$ of intensity $1/\pi$.
It can be defined (and approximately sampled) as the limit of the set of eigenvalues of matrices filled with i.i.d.\ standard complex Gaussian entries, as the size of the matrix goes to infinity \citep[Theorem 4.3.10]{Houal2013}. Interestingly, the GPP is a hyperuniform point process, and the variance of the number of points in $B(\mathbf{0}, R)$ scales like $O(R^{d-1})$ as $R$ goes to infinity, instead of scaling like the volume of the ball, as for the PPP.
By \cite{Beck1987}, this is the smallest possible growth rate.
Figure \ref{fig:repelled_ginibre} displays a sample of a GPP and the corresponding repelled sample called RGPP observed within $B(\mathbf{0}, 50)$, with $\varepsilon = \varepsilon_0$.
The sampling methodology follows Algorithm \ref{algo:sampling_rppp_A2} provided in Section \ref{sub:Sampling from the repelled point process}.
\begin{figure}[h!]
    \centering
    \begin{subfigure}{0.45\linewidth}
        \centering
        \includegraphics[width=1.0\linewidth]{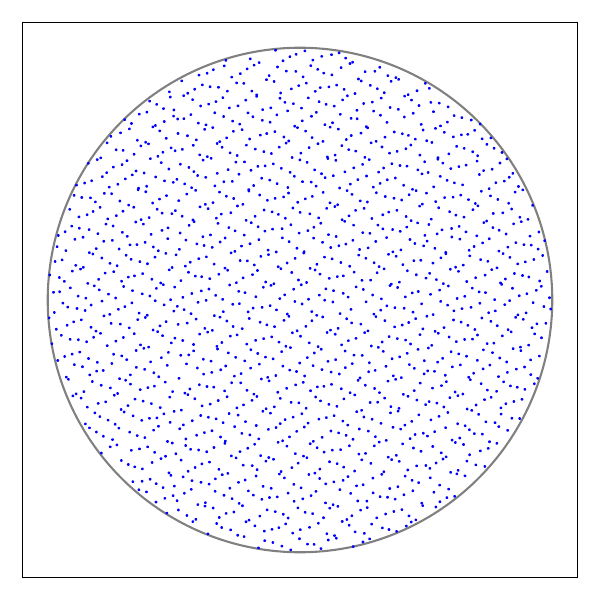}
        \caption{SSS}
    \end{subfigure}
    \begin{subfigure}{0.45\linewidth}
        \centering
        \includegraphics[width=1.0\linewidth]{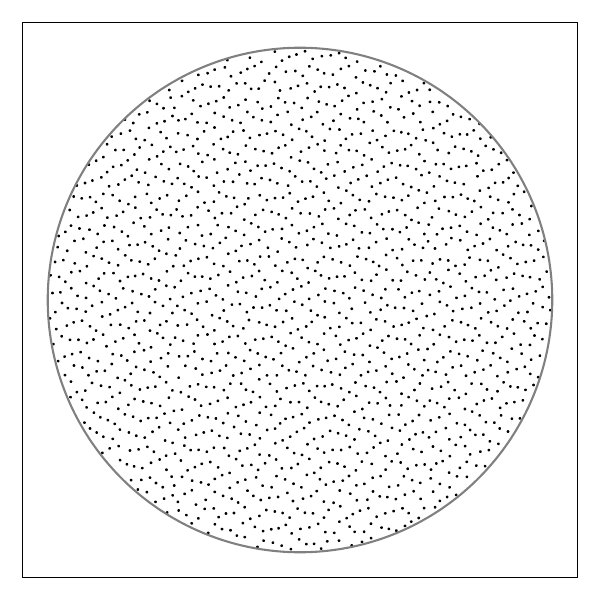}
        \caption{RSSS}
    \end{subfigure}
    \caption{A sample from the scrambled Sobol sequence (left) and the obtained repelled sample (right) with $\varepsilon=\varepsilon_0$.}
    \label{fig:reppeled_sobol}
\end{figure}
Another model of interest is the SSS, which is a typical example of a randomized low-discrepancy sequence, as mentioned in Section \ref{sub:Monte Carlo methods}.
Figure \ref{fig:reppeled_sobol} displays a sample of the SSS and the corresponding repelled sample called RSSS observed within $B(\mathbf{0}, 50)$, with $\varepsilon = \varepsilon_0$.
Like the RGPP, the RSSS displays a high level of homogeneity.

We now examine the behavior of the variance of $\widehat{I}_{s, \Pi_\varepsilon \mathcal{G} \cap K }$ and $\widehat{I}_{s, \Pi_\varepsilon \mathcal{S} \cap K }$ w.r.t.\ $\varepsilon$, where $\mathcal{G}$ represents a GPP and $\mathcal{S}$ denotes a SSS.
We use the functions $f_1$, $f_2$ and $f_3$ defined in Equation \eqref{eq:f_1_f_2_f_3}.
Figure \ref{fig:estimated_variance_repelled_ginibre_sobol} illustrates the estimated standard deviations of $\widehat{I}_{s, \, \Pi_\varepsilon \mathcal{G} \cap K }$ (first row) for $d=2$ and $\widehat{I}_{s, \, \Pi_\varepsilon \mathcal{S} \cap K }$ for $d=2$ (second row), as well as for $d=3$ (last row), for various values of $\varepsilon$.
We conducted the experiments using $50$ independent samples of $\mathcal{G}$ and $\mathcal{S}$ of intensity $\rho=500$.
GPP's samples were rescaled to achieve $\rho=500$.
The estimated standard deviations of $\widehat{I}_{s, \,\mathcal{G} \cap K }$ and  $\widehat{I}_{s, \,\mathcal{S} \cap K }$ are indicated by the large red dots in Figure \ref{fig:estimated_variance_repelled_ginibre_sobol}, while the black dots correspond to the estimated standard deviations of $\widehat{I}_{s, \, \Pi_\varepsilon \mathcal{G} \cap K }$ and $\widehat{I}_{s, \, \Pi_\varepsilon \mathcal{S} \cap K }$.
The dashed lines indicate the value of $\varepsilon_0$ defined in Equation \eqref{eq:epsilon_0}.
For the $C^2$ functions $f_1$ and $f_3$, we observe a behavior similar to the RPPP results depicted in Figure \ref{fig:variance_reduction}, indicating a variance reduction within a certain range of positive values of $\varepsilon$. However, for $f_2$, the variance decreases as $\varepsilon$ increases for $\Pi_\varepsilon \mathcal{G}$ in a manner similar to Figure \ref{fig:variance_reduction}, while a more intricate behavior is observed for $\Pi_\varepsilon \mathcal{S}$.
These observations allow one to conjecture that the repulsion operator may produce variance reduction for smooth functions for a wide range of point processes.

\begin{figure}[h!]
    \centering
    \begin{subfigure}{\textwidth}
        \centering
        \includegraphics[width=1.0\linewidth]{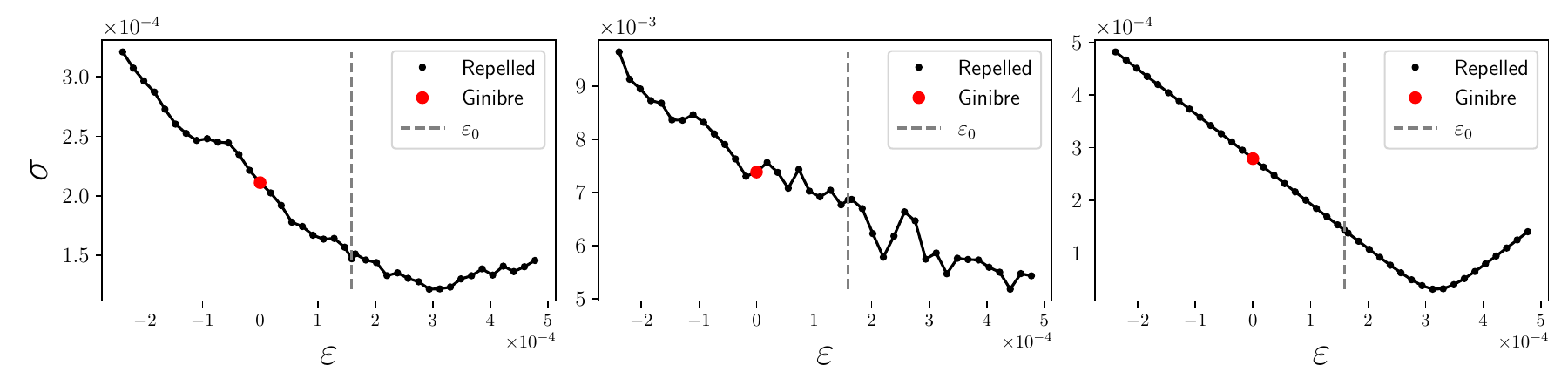}
    \end{subfigure}
    \begin{subfigure}{\textwidth}
        \centering
        \includegraphics[width=1.0\linewidth]{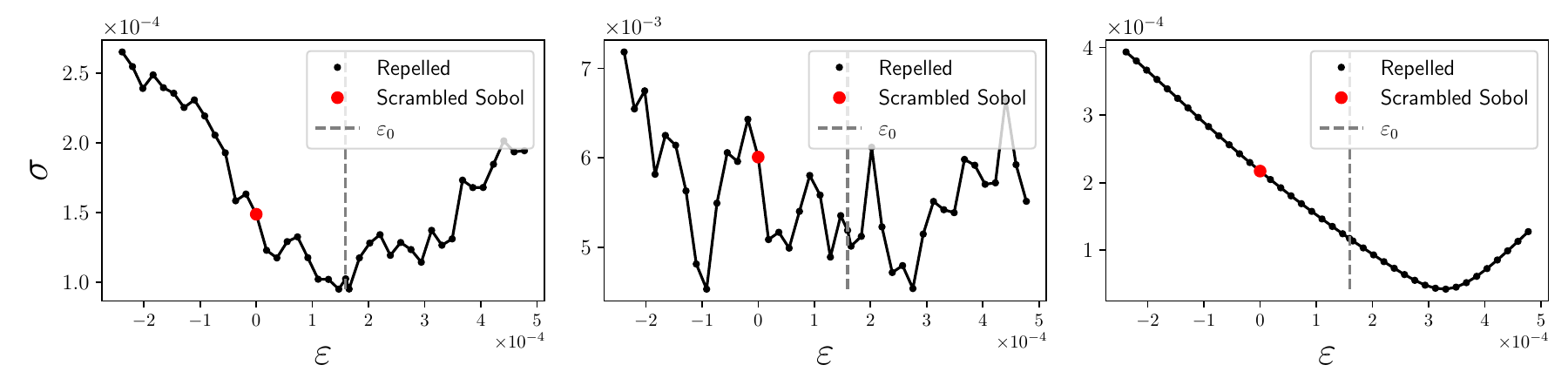}
    \end{subfigure}
    \begin{subfigure}{\textwidth}
        \centering
        \includegraphics[width=1.0\linewidth]{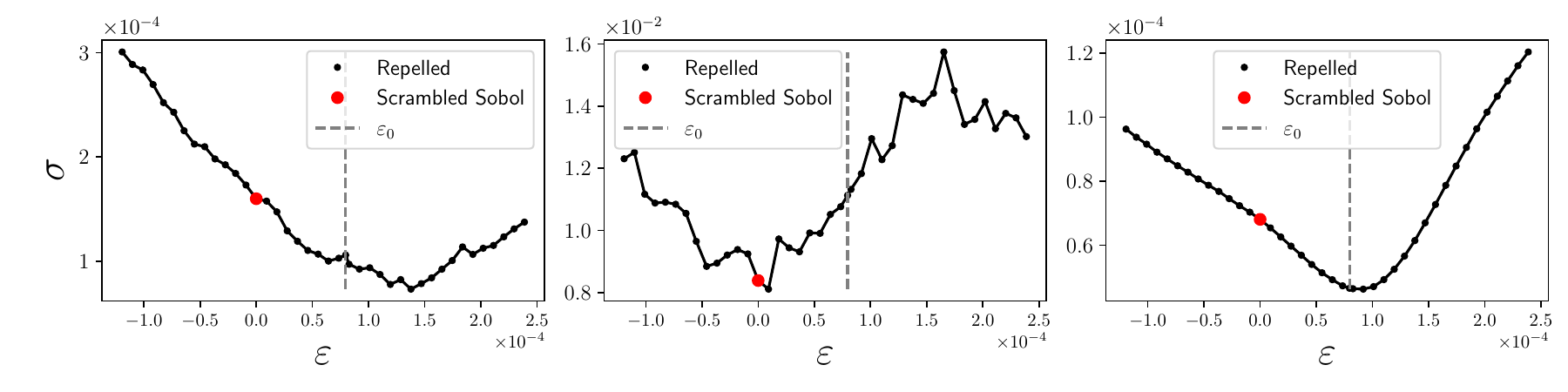}
    \end{subfigure}
    \caption{Estimated standard deviations of $\widehat{I}_{s, \Pi_\varepsilon \mathcal{G} \cap K }$ (Ginibre) and $\widehat{I}_{s, \Pi_\varepsilon \mathcal{S} \cap K }$ (scrambled Sobol) with respect to $\varepsilon$ for $f_1$ (first column), $f_2$, (second column) and $f_3$ (last column).
        The first row shows the obtained results for the GPP in $d=2$, the second row for the SSS in $d=2$ and the last one for the SSS in $d=3$.}
    \label{fig:estimated_variance_repelled_ginibre_sobol}
\end{figure}
%-----------------------------
%-----------------------------
\section{Conclusion} % (fold)
\label{sec:Conclusion}
Motivated by variance reduction for Monte Carlo methods and inspired by gravitational allocation, we have introduced the repulsion operator $\Pi_{\varepsilon}$.
For small $\varepsilon>0$, this operator intuitively makes point processes more regular by slightly pushing their points apart from each other by an amount controlled by $\varepsilon$.
We have provided a detailed theoretical study of $\Pi_\varepsilon\calP$, where the repulsion operator is applied to a homogeneous Poisson point process $\calP$.
In particular, we have proved variance reduction for smooth linear statistics and suggested a practical value for the parameter $\varepsilon$.
Numerical experiments support the variance reduction and make the repelled Poisson point process a promising alternative to crude Monte Carlo if one can afford the quadratic cost of computing pairwise distances.
As a side observation, we have found no numerical evidence of hyperuniformity for the repelled Poisson point process through the estimation of its structure factor, while the variance decay in experiments on smooth functions is compatible with a slightly faster rate than crude Monte Carlo.

Exploratory experiments suggest that variance reduction is also achieved when applying the repulsion operator to other point processes, such as the (hyperuniform) Ginibre point process and the (low-discrepancy) scrambled Sobol sequence.
Proving this is a natural next step for future work, as well as proving variance reduction when applying the repulsion operator to non-homogeneous point processes.
Once such a result is obtained under sufficiently weak assumptions, applying the repulsion operator could become a default post-processing step in many Monte Carlo integration tasks.
In another direction, it would be interesting to explore the attractive case of the operator $\Pi_\varepsilon$, achieved by setting the parameter $\varepsilon$ to a negative value.
This intuitively leads to natural clustered point process models, where points tend to aggregate near the points of the original point process.

\bmsection*{Acknowledgments}
We thank David Dereudre for valuable discussions along this project and Manjunath Krishnapur for helpful insights about gravitational allocations.

\bmsection*{Financial disclosure}
This work is supported by ERC-2019-STG-851866 and ANR-20-CHIA-0002.

\bmsection*{Data availability}
All the present illustrations and experiments were obtained using the Python toolbox \href{https://github.com/dhawat/MCRPPy}{MCRPPy}.
\href{https://github.com/dhawat/MCRPPy}{MCRPPy} is available on \href{https://github.com/dhawat/MCRPPy}{GitHub} and is accompanied by tutorial \href{https://github.com/dhawat/MCRPPy/tree/main/notebooks}{Jupyter Notebooks}.

% refs
\bibliographystyle{apalike}
\bibliography{ref}
% appendices
\appendix

\bmsection{Proofs}
\label{sec:Proofs}

In this section, we present the proofs of the results mentioned in Sections \ref{sub:Properties of the repelled Poisson point process}, \ref{sub:Main result} and \ref{sub:Properties_of_F}.
The proof of Proposition \ref{prop:joint_density_of_force} can be found in Section \ref{sub:Existence of the joint distribution of vector of force}. The motion invariance results stated by Proposition \ref{prop:motion_invariance} and Proposition \ref{prop:motion_invariance_for_poisson} are proven in Section \ref{sub:Proof_of_Proposition_1}.
The most intricate proof, that of Proposition \ref{prop:extistance_of_the_moments}, is presented in Section \ref{proof:exitance_moment}. Finally, Section \ref{subs:proof_variance_reduction} contains the proof of the variance reduction result formulated in Theorem \ref{thm:Variance_reduction}, which is the main outcome of the paper.

\subsection{Proofs of Proposition \ref{prop:joint_density_of_force} } % (fold)
\label{sub:Existence of the joint distribution of vector of force}
In this section, our objective is to prove Proposition \ref{prop:joint_density_of_force}.
This proposition will play a crucial role in establishing the first part of Proposition \ref{prop:motion_invariance_for_poisson}.
\begin{proof}[Proof of Proposition \ref{prop:joint_density_of_force}]
    \label{proof:corollary_joint_density_of_force}
    Let $\{B_n\}_{n \geq 1}$ be a sequence of disjoint balls, each with the same volume $|B_1|= 1/ \rho$, and consider the collection of events $\{\Omega_n\}_{n \geq 1}$ defined by
    \begin{equation*}
        \Omega_n  = \left \{\calP(B_n) = 1 \right \}.
    \end{equation*}
    As $\calP$ is a PPP, the events $\Omega_1, \Omega_2, \dots$ are independent, and we have
    \begin{align*}
        \sum_{n \geq 1} \bbP(\Omega_n) =  \sum_{n \geq 1}  \exp(-1) = \infty.
    \end{align*}
    By Borel-Cantelli, $\bbP(\lim\sup \Omega_n) = 1$.
    Thus almost surely, infinitely many $\Omega_n$ occur, and for $\mathbf{c} \in \mathbb{R}^d$,
    \begin{align*}
        \left\{F_{\calP}(\bfx) - F_{\calP}(\bfy) = \bfc \right\} &\subseteq \bigcup_{n \geq 1}\limits \left\{F_{\calP}(\bfx) - F_{\calP}(\bfy) = \bfc, \Omega_n \right\} \\
        &\subseteq \bigcup_{n \geq 1}\limits \left\{\bfX_n = \bfc - \bfY_{n}, \Omega_n \right\} ,
    \end{align*}
    where
    $\bfX_n = F_{\calP\cap B_{n}}(\bfx) - F_{\calP\cap B_{n}}(\bfy)$,
    and
    $
    \bfY_{n} \triangleq F_{\calP \cap B_{n}^c}(\bfx) - F_{\calP \cap B_{n}^c}(\bfy).
    $
    In particular,
    \begin{align}
        \bbP
        \left(
        F_{\calP}(\bfx) - F_{\calP}(\bfy) = \bfc
        \right)
        % &\leq
        % \bbP
        % \left(
        % \bigcup_{n \geq 1}\limits \left\{\bfX_n = \bfc - \bfY_{n}, \Omega_n \right\}
        % \right)
        % \\
        & \leq
        \sum_{n \geq 1} \bbP
        \left(
        \bfX_n = \bfc - \bfY_{n}, \Omega_n
        \right) .
        \label{e:union_argument}
    \end{align}
    Now, for all $n$, conditionally on $\Omega_n$, $\bfX_{n}$ and $\bfY_n$ are independent random vectors and we further claim that $\bfX_{n}$ is continuous, we thus get
    $
    \bbP
    \left(
    \bfX_n = \bfc - \bfY_{n}, \Omega_n
    \right) = 0.
    $
    By \eqref{e:union_argument}, we conclude that
    $$\bbP
    \left(
    F_{\calP}(\bfx) - F_{\calP}(\bfy) = \bfc
    \right)=0.
    $$
    
    Finally, the claim that $\bfX_{n}$ is continuous conditionally on the event $\Omega_n$ can be easily verified using harmonic function theory.
    For instance, for $i\in {1, \dots, d}$, let
    \begin{align*}
        g_i \colon B_n \setminus \{\bfx, \bfy\} &\to \mathbb{R} \\
        \bfz &\mapsto \frac{x_i - z_i}{\|\bfx - \bfz\|_2^d} - \frac{y_i - z_i}{\|\bfy - \bfz\|_2^d} .
    \end{align*}
    Actually, for $\bfz \triangleq \calP \cap B_n$, $g_i(\bfz)$ is the $i$-th component of the random vector $\bfX_n$.
    As $g_i$ is a non-constant real harmonic function on $B_n \setminus \{\bfx, \bfy\}$, by Theorem 1.28 of \cite{AxlBourRam2001}, $g_i$ is a non-constant real analytic function on $B_n \setminus \{\bfx, \bfy\}$ (which is connected).
    By Proposition 1 of \cite{Mit2020}, the zero-set $g_i^{-1}(0)$ of $g_i$ has Lebesgue measure zero.
    By translation, we can deduce that any level set of $g_i$ is negligible.
    Finally, conditioning on $\Omega_n$, $\bfz$ is uniformly distributed on $B_n$, so
    $$\bbP\left(\bfX_n = \bfc \giventhat \Omega_n\right) \leq |B_n \cap \{g_i^{-1}(c_i)\}| = 0.$$
\end{proof}
% subsection Proofs of Section (end)
\subsection{Proof of Propositions \ref{prop:motion_invariance} and \ref{prop:motion_invariance_for_poisson} (motion-invariance)} % (fold)
\label{sub:Proof_of_Proposition_1}
Throughout this section, we fix $\varepsilon \in \mathbb{R}$ and we prove Propositions \ref{prop:motion_invariance} and \ref{prop:motion_invariance_for_poisson}.

First, to prove Proposition \ref{prop:motion_invariance}, we show that for a point process $\calX$, satisfying the conditions of Proposition \ref{prop:motion_invariance}, any void probability \eqref{eq:def_void_proba} of a translation, respectively rotation, of $\Pi_\varepsilon \calX$ is equal to the void probability for $\Pi_\varepsilon \calX$ of the same Borel set.
The key argument is that, for any $\bfx  \in \mathbb{R}^d$, the force $F_{\calX} (\bfx )$ defined in \eqref{eq:def_gravitational_force} is invariant under both translations and rotations.

\begin{proof}[Proof of Proposition \ref{prop:motion_invariance}]
    \label{proof:propo_motion_invariant}
    Assume that $\calX$ is motion-invariant.
    As the law of $\Pi_\varepsilon \calX$ is defined by the void probabilities \eqref{eq:def_void_proba}, to show that $\Pi_\varepsilon \calX$ is motion-invariant, it is enough to prove that for any Borel set $B$, any $\mathbf{a} \in \mathbb{R}^d$, and any rotation $\mathfrak{r}$,
    \begin{equation*}
        \bbT_{\mathbf{a}+\Pi_\varepsilon \calX}(B)
        =
        \bbT_{\Pi_\varepsilon \calX}(B)
        \quad
        \text{and}
        \quad
        \bbT_{\mathfrak{r}(\Pi_\varepsilon \calX)}(B)
        =
        \bbT_{\Pi_\varepsilon \calX}(B) .
    \end{equation*}
    
    First, observe that by construction $\Pi_{\varepsilon}(\mathbf{a}+ \calX) = \mathbf{a} + \Pi_{\varepsilon} \calX$.
    This yields
    \begin{align*}
        \bbT_{\bfa+ \Pi_{\varepsilon} \calX} (B)
        & = 1 - \bbP\left((\bfa +\Pi_{\varepsilon} \calX ) \cap B = \emptyset     \right)\\
        &= 1 - \bbP\left(\Pi_{\varepsilon}(\bfa+ \calX ) \cap B = \emptyset \right) ,
    \end{align*}
    which, by stationarity of $\calX$ is equal to $\bbT_{\Pi_{\varepsilon} \calX} (B)$. Thus, $\Pi_{\varepsilon} \calX$ is stationary.
    
    Second, since $\lVert \mathfrak{r}(\bfx ) \lVert_2 = \lVert \bfx  \lVert_2$ and $\mathfrak{r}$ is linear, we can write
    \begin{align*}
        \Pi_{\varepsilon}(\mathfrak{r}(\calX))
        &=
        \left \{
        \mathfrak{r}(\bfx )
        +
        \varepsilon
        \sum_{\substack{\bfz \in \calX \setminus \{\bfx \} \\ \lVert \mathfrak{r}(\bfx ) - \mathfrak{r}(\bfz)\lVert_2 \uparrow  }}
        \frac{\mathfrak{r}(\bfx )- \mathfrak{r}(\bfz)}{\lVert \mathfrak{r}(\bfx )- \mathfrak{r}(\bfz)\lVert_2^{d}}
        \right \}_{\bfx  \in \calX}
        =
        \left \{ \mathfrak{r}(\bfx )
        +
        \varepsilon
        \sum_{\substack{\bfz \in \calX \setminus \{\bfx \} \\ \lVert \mathfrak{r}(\bfx - \bfz)\lVert_2 \uparrow  }}
        \frac{\mathfrak{r}(\bfx - \bfz)}
        {\lVert \mathfrak{r}(\bfx - \bfz)\lVert _2^{d}} \right \}_{\bfx  \in \calX} \\
        &=
        \left \{
        \mathfrak{r}(\bfx )
        +
        \varepsilon
        \sum_{\substack{\bfz \in \calX \setminus \{\bfx \} \\ \lVert \bfx - \bfz \lVert_2 \uparrow  }}
        \frac{\mathfrak{r}(\bfx - \bfz)}
        {\lVert \bfx - \bfz \lVert_2^{d}} \right \}_{\bfx  \in \calX}
        =
        \left \{
        \mathfrak{r} ( \bfx  + \varepsilon F_{ \calX} (\bfx ) ) \right \}_{\bfx  \in \calX} .
    \end{align*}
    Thus, we have $\Pi_{\varepsilon}(\mathfrak{r}(\calX)) = \mathfrak{r}(\Pi_{\varepsilon}(\calX))$.
    This implies that
    \begin{align*}
        \bbT_{\mathfrak{r}(\Pi_{\varepsilon}(\calX))} (B)
        & = 1 - \bbP\left(\mathfrak{r}(\Pi_{\varepsilon}(\calX) ) \cap B = \emptyset \right)                                             \\
        & = 1 - \bbP\left(\Pi_{\varepsilon}(\mathfrak{r}(\calX) ) \cap B = \emptyset \right) ,
    \end{align*}
    which is equal to $\bbT_{\Pi_{\varepsilon}(\calX)} (B)$ by the isotropy of $\calX$.
    Thus, $\Pi_{\varepsilon}(\calX)$ is isotropic.
\end{proof}
We highlight that the proof remains valid when substituting the repulsion operator $\Pi_{\varepsilon}$ with its truncated version $\Pi_{\varepsilon}^{(q,p)}$ \eqref{eq:truncated_operator}, with $0 \leq q < p$.

Now, to prove Proposition~\ref{prop:motion_invariance_for_poisson}, we first note that a homogeneous Poisson point process $\calP$ is almost surely a valid configuration, as mentioned in Section~\ref{sub:Properties_of_F}.
In view of Proposition~\ref{prop:motion_invariance}, it is enough to prove that, almost surely, the images of two distinct points from $\calP$ under $\Pi_{\varepsilon}$ remain distinct and that $\Pi_{\varepsilon} \calP $ possesses the same intensity as $\calP$.
\begin{proof}[Proof of Proposition \ref{prop:motion_invariance_for_poisson}]
    \label{proof:corol_motion_invariant_poisson}
    To show that almost surely the images of two distinct points from $\calP$ under $\Pi_{\varepsilon}$ remain distinct,\footnote{
        In the broader framework of point processes, our objective is to establish that $\Pi_{\varepsilon}(\calP)$ is a \emph{simple} point process, assuming that we define $\Pi_{\varepsilon}(\calP)$ as a multiset rather than a set.
        To accomplish this, we employ \citep[Proposition 6.7]{LasPen2017}.
    } we need to show that
    \begin{equation}
        \label{e:objective}
        \bbE  \left[\sum_{\bfx , \bfy \in \mathcal{P}}^{\neq} \mathds{1}_{\{\bfx + \varepsilon F_{\calP}(\bfx) = \bfy + \varepsilon F_{\calP}(\bfy)\}}(\bfx, \bfy)\right] = 0 .
    \end{equation}
    Using the extended Slivnyak-Mecke theorem \eqref{eq:slivnyak-mecke} we get
    \begin{align}
        \bbE
        \left[\sum_{\bfx , \bfy \in \mathcal{P}}^{\neq}  \mathds{1}_{\{\bfx + \varepsilon F_{\calP}(\bfx) = \bfy + \varepsilon F_{\calP}(\bfy)\}}(\bfx, \bfy)\right]
        &= \bbE
        \left[
        \sum_{\bfx , \bfy \in \mathcal{P}}^{\neq}
        \mathds{1}_{\{\bfx + \varepsilon F_{\calP \setminus \{\bfx, \bfy\}}(\bfx) + \varepsilon \frac{\bfx - \bfy}{\|\bfx - \bfy\|_2^d} = \bfy + \varepsilon F_{\calP \setminus \{\bfx, \bfy\}}(\bfy) + \varepsilon \frac{\bfy - \bfx}{\|\bfy - \bfx\|_2^d}\}}(\bfx, \bfy)\right]
        \nonumber \\
        &= \int_{\bbR^d \times \bbR^d}
        \bbE
        \left[
        \mathds{1}_{\{\bfx + \varepsilon F_{\calP}(\bfx) + \varepsilon \frac{\bfx - \bfy}{\|\bfx - \bfy\|_2^d} = \bfy + \varepsilon F_{\calP}(\bfy) + \varepsilon \frac{\bfy - \bfx}{\|\bfy - \bfx\|_2^d}\}}(\bfx, \bfy) \right]
        \rho^2
        \d \bfx \d \bfy
        \nonumber \\
        &= \int_{\bbR^d \times \bbR^d}
        \bbP
        \left(
        F_{\calP}(\bfx) - F_{\calP}(\bfy) = (\bfy - \bfx)\left(\varepsilon^{-1} + 2 \|\bfy - \bfx\|_2^{-d}\right)
        \right)
        \rho^2
        \d \bfx \d \bfy .
        \label{e:tool}
    \end{align}
    By Proposition \ref{prop:joint_density_of_force}, for $\bfx\neq \bfy$, the random vector $F_{\calP}(\bfx)- F_{\calP}(\bfy)$ is continuous. Thus
    \begin{equation*}
        \bbP
        \left(
        F_{\calP}(\bfx) - F_{\calP}(\bfy) = (\bfy - \bfx)\left(\varepsilon^{-1} + 2 \|\bfy - \bfx\|_2^{-d}\right)
        \right) = 0 .
    \end{equation*}
    Plugging back in \eqref{e:tool} yields \eqref{e:objective}.
    
    It remains to show that the intensity of $\Pi_{\varepsilon} \calP $ is equal to $\rho$.
    Consider a compact $K$ of $\bbR^d$.
    Based on the previous reasoning, almost surely, when the repulsion operator $\Pi_\varepsilon$ is applied to the points of $\mathcal{P}$, no two points will end up at the same location. Thus we have
    \begin{align*}
        \Pi_{\varepsilon} \calP (K)= \sum_{\bfx \in \calP} \mathds{1}_K (\bfx + \varepsilon F_{\calP}(\bfx)) .
    \end{align*}
    Applying the extended Slivnyak-Mecke theorem \eqref{eq:slivnyak-mecke} we get
    \begin{align*}
        \bbE \left[\Pi_{\varepsilon} \calP (K)\right] &= \bbE \left[\sum_{\bfx \in \calP} \mathds{1}_K \left(\bfx + \varepsilon F_{\calP \setminus \{\bfx\}}(\bfx)\right) \right]
        = \int_{\bbR^d} \bbE \left[ \mathds{1}_K (\bfx + \varepsilon F_{\calP}(\bfx)) \right] \rho \d \bfx .
    \end{align*}
    As the distribution of $F_{\calP}(\bfx)$ is translation-invariant (in $\bfx$) we get
    \begin{align*}
        \bbE[\Pi_{\varepsilon} \calP (K)] &=  \int_{\bbR^d} \bbE \left[ \mathds{1}_K (\bfx + \varepsilon F_{\calP}(\mathbf{0})) \right] \rho \d \bfx .
    \end{align*}
    Exchanging the integral and the expectation using Tonelli's theorem gives
    \begin{align*}
        \bbE[\Pi_{\varepsilon} \calP (K)] &=  \rho \bbE \left[ \int_{\bbR^d} \mathds{1}_K (\bfx + \varepsilon F_{\calP}(\mathbf{0})) \d \bfx \right]
        \\
        &=\rho \bbE \left[ \int_{\bbR^d} \mathds{1}_{K- \varepsilon F_{\calP}(\mathbf{0})} (\bfx )  \d \bfx\right]
        \\
        &= \rho \bbE \left[|K- \varepsilon F_{\calP}(\mathbf{0}) |\right]= \rho  | K | .
    \end{align*}
    Thus the intensity of $\Pi_{\varepsilon} \calP $ is equal to $\rho$, which completes the proof.
\end{proof}

% subsection (end)
\subsection{Proof of Proposition \ref{prop:extistance_of_the_moments} (existence of the moments)} % (fold)
\label{proof:exitance_moment}
In this section, we consider a homogeneous Poisson point process $\mathcal{P} \subset \mathbb{R}^d$ of intensity $\rho$, with $d \geq 3$, and let $\varepsilon \in (-1,1)$.
Our objective is to demonstrate the existence of the moments of the repelled Poisson point process $\Pi_{\varepsilon}\mathcal{P}$.
To wit, let $R>0$ and $m$ be a positive integer, we need to show that
\begin{equation*}
    \bbE \left[
    \left(\sum_{\bfx  \in \Pi_\varepsilon \mathcal{P}}
    \mathds{1}_{B(\mathbf{0},R)}(\bfx )\right)^m\right]
    < \infty.
    \label{e:repetition_of_moment_condition}
\end{equation*}
To accomplish this, the key idea is to decompose the Coulomb force $F_\calP(\bfx)$ acting on $\bfx\in\mathbb{R}^d$ and defined in \eqref{eq:alternative_def_F} into two truncated sums, one that collects the influence of points close to $\bfx$ and the other one of those points in $\calP$ far from $\bfx$; the two terms shall be controlled by different means.
Formally, denote $B(\mathbf{0}, R)^c \triangleq \bbR^d \setminus B(\mathbf{0}, R)$.
For $\bfx  \in \mathbb{R}^d$, we write
\begin{equation}
    \label{eq:decomposition_internal_external}
    F_{\mathcal{P}  }(\bfx )= F^{(0,1)}_{\mathcal{P}  }(\bfx ) + F^{(1,\infty)}_{\mathcal{P}  }(\bfx ) ,
\end{equation}
where the truncated forces are defined in \eqref{eq:def_truncated_gravitational_force}.
We refer to the first term in the right-hand side of \eqref{eq:decomposition_internal_external} as the ``internal'' force, and to the second term as the ``external'' force.
In words, our proof works by showing that, for $\bfx \in \calP\cap B(\mathbf{0}, R)^c$ to be pushed inside $B(\mathbf{0}, R)$, i.e. for $\bfx  + \varepsilon F_{\mathcal{P}  }(\bfx ) \in B(\mathbf{0}, R)$, one of two low-probability events must occur.
One of these events involves the internal force, and the other one the external force.

Let $0<\gamma<1/(d-1)$, $0 < \beta < \gamma/(d-1)$, and $r:\bfx \mapsto \| \bfx \|_2^{\beta}$.
Let also $R^\prime = (R+m-1)^{1/\beta}$.
Now, for $\bfx  \in \mathcal{P} \cap B(\mathbf{0}, R^\prime)^c$, it holds
\begin{equation}
    \label{eq:inclusion_event_a_point_enter}
    \left\{\bfx  + \varepsilon F_{\mathcal{P} }(\bfx ) \in B(\mathbf{0}, R) \right\}
    \subset
    \left\{\bfx  + \varepsilon F^{(0,1)}_{\mathcal{P} } (\bfx ) \in B(\mathbf{0}, r(\bfx )) \right\}
    \cup
    \left\{\|F^{(1,\infty)}_{\mathcal{P}}(\bfx )\|_2 \geq \frac{r(\bfx )- R}{|\varepsilon|} \right\} .
\end{equation}
To see the validity of this inclusion, note that if $\bfx$ is not in the right-hand side of \eqref{eq:inclusion_event_a_point_enter}, then
\begin{align*}
    \|\bfx  + \varepsilon F_{\mathcal{P} }(\bfx )\|_2
    \geq
    \|\bfx  + \varepsilon F^{(0,1)}_{\mathcal{P} }(\bfx )\|_2
    -
    \|\varepsilon F^{(1,\infty)}_{\mathcal{P} }(\bfx )\|_2
    > r(\bfx ) - (r(\bfx ) - R) = R .
\end{align*}
Now, using \eqref{eq:inclusion_event_a_point_enter}, we get
\begin{align*}
    \bbE \left[\left(\sum_{\bfx  \in \Pi_{\varepsilon}\mathcal{P}} \mathds{1}_{B(\mathbf{0},R)}(\bfx )\right)^m \right]
    &\leq
    \bbE \Bigg [\Bigg(\sum_{\bfx  \in \mathcal{P}} \mathds{1}_{B(\mathbf{0},R^\prime)}(\bfx )
    +
    \sum_{\bfx  \in \mathcal{P} \cap B(\mathbf{0},R^\prime)^c} \mathds{1}_{B(\mathbf{0},r(\bfx ))}(\bfx  + \varepsilon F^{(0,1)}_{\mathcal{P}}(\bfx ))
    \\
    &\qquad
    +
    \sum_{\bfx  \in \mathcal{P} \cap B(\mathbf{0},R^\prime)^c} \mathds{1}_{\{\|F^{(1,\infty)}_{\mathcal{P}}(\bfx )\|_2 \geq (r(\bfx )- R)/|\varepsilon|\}}(\bfx )\Bigg)^m \Bigg] .
\end{align*}
By convexity of $h: x \mapsto x^m$ on $\mathbb{R}^{+}\setminus \{0\}$, it further comes
\begin{align}
    \bbE \left[\left(\sum_{\bfx  \in \Pi_{\varepsilon}\mathcal{P}} \mathds{1}_{B(\mathbf{0},R)}(\bfx )\right)^m \right]
    &\leq
    3^{m-1}\Biggl(
    \bbE \left[\left(\sum_{\bfx  \in \mathcal{P}} \mathds{1}_{B(\mathbf{0},R^\prime)}(\bfx )\right)^m \right] \nonumber
    \\
    &
    \hspace{1cm}
    +
    \bbE \left[\left(\sum_{\bfx  \in \mathcal{P} \cap B(\mathbf{0},R^\prime)^c} \mathds{1}_{B(\mathbf{0},r(\bfx ))}(\bfx  + \varepsilon F^{(0,1)}_{\mathcal{P} }(\bfx ))\right)^m \right] \nonumber
    \\
    &
    \hspace{1cm}
    +
    \bbE \left[\left(\sum_{\bfx  \in \mathcal{P}\cap B(\mathbf{0},R^\prime)^c} \mathds{1}_{\{\|F^{(1,\infty)}_{\mathcal{P}}(\bfx )\|_2 \geq (r(\bfx )- R)/|\varepsilon|\}}(\bfx )\right)^m \right] \label{e:three_term_decomposition}
    \Bigg).
\end{align}
The first term on the right-hand side of \eqref{e:three_term_decomposition} is finite since $\mathcal{P}$ is a PPP.
The rest of the proof consists in proving that the remaining two terms are finite, which will be a consequence of Corollaries \ref{coro:existence_external_moments} and \ref{coro:existence_internal_moments}.

We first focus on the term in \eqref{e:three_term_decomposition} involving the external force.
\begin{lemma}
    \label{lem:bound_external_force}
    Consider a homogeneous Poisson point process $\mathcal{P} \subset \mathbb{R}^d$ of intensity $\rho$, with $d \geq 3$.
    There exist $c_1, \, c_2, \, c_3 >0$ such that for all $p>q>0$ and $t>0$ we have
    \begin{equation}
        \label{eq:tail_bound_external_force}
        \bbP \left(\lVert F^{(q,p)}_{\mathcal{P}}( \mathbf{0} ) \lVert_2 >t\right)
        \leq
        c_1 \exp \left( -c_2 q^{d-1} t \log\left(\frac{c_3t}{q} \right) \right) .
    \end{equation}
\end{lemma}
Note that by translation-invariance, the choice of $\mathbf{0}$ in \eqref{eq:tail_bound_external_force} is arbitrary.
In addition, when $\rho=1$, Equation \eqref{eq:tail_bound_external_force} corresponds to Equation (32) in Theorem 16 of \cite{ChaPelal2010}.
While the paper does not offer an exhaustive proof of this equation, it does provide a similar and detailed proof for another equation within the same theorem.
For completeness, we provide here a simple proof of Lemma~\ref{lem:bound_external_force}, valid for any value of $\rho>0$.
The proof involves bounding the exponential moment of each component in the random vector $F^{(q,p)}_{\calP}(\mathbf{0})$ and using Markov's inequality to obtain a tail bound.
The Poisson assumption then helps simplify the bound.

\begin{proof}[Proof of Lemma \ref{lem:bound_external_force}]
    In order to streamline the notations used in this proof, we set $G^{(q,p)}=-F^{(q,p)}_{\mathcal{P}}$, with $G^{(q,p)}_{i}$ its $i$-th component.
    
    If we prove that for any $i \in \{1, \dots, d\}$, there exist $C_1, \, C_2, \, C_3 >0$ such that
    \begin{equation}
        \label{eq:eq_1_to_proof}
        \bbP \left(G^{(q,p)}_{i}(\mathbf{0}) >t\right) \leq C_1 \exp \left(-C_2 q^{d-1}t \log\left(\frac{C_3 t}{q}\right) \right) .
    \end{equation}
    Using that $-\mathcal{P}$ is also a PPP of intensity $\rho$, it follows that
    \begin{equation*}
        \label{eq:eq_2_to_proof}
        \bbP \left(G^{(q,p)}_{i}(\mathbf{0}) <-t\right) \leq C_1 \exp \left(-C_2 q^{d-1}t \log\left(\frac{C_3 t}{q}\right) \right) .
    \end{equation*}
    Combining this with \eqref{eq:eq_1_to_proof} we obtain \eqref{eq:tail_bound_external_force}, with $c_1 = 2dC_1, \, c_2 = C_2/d$, and $c_3 = C_3/d$.
    Thus, we only need to verify Equation \eqref{eq:eq_1_to_proof}.
    
    Let $\theta \geq 0$, using Markov's inequality we have
    \begin{align}
        \label{eq:intermediate_step}
        \bbP \left(G^{(q,p)}_{i}(\mathbf{0}) >t\right) & =
        \bbP \left(\exp \left(\theta \sum_{\substack{\mathbf{z} \in A^{(q,p)}\cap \mathcal{P} \\ \|\mathbf{z}\|_2 \uparrow}} \frac{z_i}{\|\mathbf{z}\|_2^d}\right) > \exp(\theta t)\right)
        \nonumber
        \\
        & \leq \bbE  \left[\exp  \left(\theta \sum_{\substack{\mathbf{z} \in A^{(q,p)}\cap \mathcal{P} \\ \|\mathbf{z}\|_2 \uparrow}} \frac{z_i}{\|\mathbf{z}\|_2^d}\right) \right] \exp(-\theta t) .
    \end{align}
    Conditioning on $\mathcal{P}(A^{(q,p)})$, the points of $\mathcal{P} \cap A^{(q,p)}$ are independent and uniformly distributed in $A^{(q,p)}$.
    Let $\mathbf{u}$ be a uniform r.v. in $A^{(q,p)}$, and $N>0$. By symmetry, $\bbE [\frac{\mathbf{u}}{\|\mathbf{u}\|_2^d}]=0$ and we have
    \begin{align*}
        \bbE  \left[
        \exp \left(\theta \sum_{\substack{\mathbf{z} \in A^{(q,p)}\cap \mathcal{P} \\ \|\mathbf{z}\|_2 \uparrow}} \frac{z_i}{\|\mathbf{z}\|_2}\right) \giventhat \mathcal{P}(A^{(q,p)}) = N  \right]
        &=
        \bbE  \left[
        \prod_{\substack{\mathbf{z} \in A^{(q,p)}\cap \mathcal{P} \\ \|\mathbf{z}\|_2 \uparrow}}
        \exp \left(\theta \frac{z_i}{\|\mathbf{z}\|_2^d}\right) \giventhat \mathcal{P}(A^{(q,p)}) = N  \right]
        \\
        &= \bbE  \left[\exp \left(\theta \frac{u_i}{\|\mathbf{u}\|_2^d}\right) \right]^N
        \\
        &=  \left( 1 +  \bbE  \left[\sum_{k \geq 2} \frac{1}{k!} \theta^k \frac{u_i^k}{\|\mathbf{u}\|_2^{kd}} \right] \right)^N.
    \end{align*}
    In particular,
    \begin{align}
        \bbE  \left[
        \exp \left(\theta \sum_{\substack{\mathbf{z} \in A^{(q,p)}\cap \mathcal{P} \\ \|\mathbf{z}\|_2 \uparrow}} \frac{z_i}{\|\mathbf{z}\|_2}\right) \giventhat \mathcal{P}(A^{(q,p)}) = N  \right]
        &\leq
        \left( 1 +  \sum_{k \geq 2} \frac{1}{k!} \theta^k \bbE  \left[\|\mathbf{u}\|_2^{-k(d-1)} \right] \right)^N.
        \label{e:intermediate_step}
    \end{align}
    Recall that the surface area of the unit ball of $\bbR^d$ is equal to $d\kappa_d$. As $\bfu$ has uniform distribution on $A^{(q,p)}$ we get
    \begin{align*}
        \bbE  \left[\frac{1}{\|\mathbf{u}\|_2^{k(d-1)}}  \right]
        &=
        \frac{1}{\leb{A^{(q,p)}}}\int_{A^{(q,p)}} \frac{1}{\|\mathbf{u}\|_2^{k(d-1)}} \d \mathbf{u}
        \\
        &=
        \frac{d\kappa_d}{\leb{A^{(q,p)}}}\int_{q}^{p} r^{d-1 - k(d-1)} \d r .
    \end{align*}
    As $k \geq 2$ and $d \geq 3$, we have $k > d/(d-1)$. Thus
    \begin{align*}
        \bbE  \left[\frac{1}{\|\mathbf{u}\|_2^{k(d-1)}}  \right]
        &=
        \frac{d\kappa_d}{\leb{A^{(q,p)}} (k(d-1) -d)}\left(\frac{1}{q^{k(d-1) -d}} -  \frac{1}{p^{k(d-1) -d}}\right)
        \\
        & \leq
        \frac{d\kappa_d}{\leb{A^{(q,p)}} q^{k(d-1) -d}} .
    \end{align*}
    Plugging back into \eqref{eq:intermediate_step}, we obtain
    \begin{align*}
        \bbP \left(G^{(q,p)}_{i}( \mathbf{0} ) >t\right) &
        \leq
        \bbE  \left[
        \left(
        1 +
        \frac{ d\kappa_d q^d}{ \leb{A^{(q,p)}}} \sum_{k \geq 2} \frac{1}{k!} \left(\frac{\theta}{q^{d-1}}\right)^k \right)^{\mathcal{P}(A^{(q,p)})} \right] \exp(-\theta t)
        \\
        & \leq
        \bbE  \left[
        \left(
        1 +
        \frac{ d \kappa_d q^d}{ \leb{A^{(q,p)}}} \exp \left(\frac{\theta}{q^{d-1}}\right)\right)^{\mathcal{P}(A^{(q,p)})} \right] \exp(-\theta t) .
    \end{align*}
    Now, remembering that $\mathcal{P}(A^{(q,p)})$ is a Poisson random variable of parameter $\rho \leb{A^{(q,p)}}$, we obtain\footnote{If $X$ is a Poisson random variable of parameter $\lambda$, then for any $\gamma$ the mean of the random variable $(1+\gamma)^X$ is equal to $\exp(\lambda \gamma)$.}
    \begin{align*}
        \bbP (G^{(q,p)}_{i}( \mathbf{0} ) >t) &
        \leq
        \exp \left(\rho \leb{A^{(q,p)}} \frac{ d \kappa_d q^d}{ \leb{A^{(q,p)}}} \exp \left(\frac{\theta}{q^{d-1}}\right)\right) \exp(-\theta t)
        \\
        &=
        \exp \left( d \kappa_d \rho  q^d \exp \left(\frac{\theta}{q^{d-1}}\right) -\theta t\right) .
    \end{align*}
    Taking $\theta = q^{d-1}\log(\frac{t}{d \kappa_d \rho q})$, we get
    \begin{align*}
        \bbP (G^{(q,p)}_{i}( \mathbf{0} ) >t) &
        \leq
        \exp \left( - tq^{d-1}\log \left(\frac{t}{d \kappa_d \rho q}\right) +  t  q^{d-1} \right)
        \\
        &=
        \exp \left( - tq^{d-1}\log\left(\frac{t}{ e  d \kappa_d \rho q}\right) \right),
    \end{align*}
    which ends the proof.
\end{proof}

Lemma \ref{lem:bound_external_force} has the following corollary.

\begin{corollary}
    \label{coro:existence_external_moments}
    Consider a homogeneous Poisson point process $\mathcal{P} \subset \mathbb{R}^d$ of intensity $\rho$, with $d \geq 3$.
    Let $ R > 0$, $ \varepsilon \in (-1,1)$, $\beta \in (0, 1)$ and $r(\bfx ) = \|\bfx \|_2^{\beta}$.
    Then, for any positive integer $m$, there exist positive constants $(a_k)_{k=1}^m$, $(b_k)_{k=1}^m$, and $(c_k)_{k=1}^m$  such that
    \begin{align}
        \label{eq:bound_moment_external_force}
        &\bbE  \left[
        \left(\sum_{\bfx  \in \mathcal{P} \cap B(\mathbf{0}, R^\prime)^c} \mathds{1}_{\left\{\lVert F_{\mathcal{P} }^{(1,\infty )}(\bfx )\lVert_2 >\frac{r(\bfx ) - R}{|\varepsilon|}\right\}}(\bfx )\right)^m
        \right] \leq
        \sum_{k=1}^{m}a_k
        \left(\int_{B(\mathbf{0}, R^\prime)^c}
        \exp \left( -b_k g_k(\bfx ) \log( c_k g_k(\bfx )) \right) \rho \d \bfx \right)^k,
    \end{align}
    where $R^\prime =(R+ m-1)^{1/\beta}$ and $g_k(\bfx ) = \frac{r(\bfx ) - (R + k-1)}{ |\varepsilon| k} $.
\end{corollary}
This corollary helps us control the external term in \eqref{eq:decomposition_internal_external}.
Indeed, as $\|\bfx\|_2 \rightarrow \infty$ we have
$$
\exp \left( -b_k g_k(\bfx ) \log( c_k g_k(\bfx )) \right) =  o\left(\exp \left(- \|\bfx\|_2^{\beta}\right)\right).
$$
Thus, for any positive integer $m$
\begin{equation}
    \label{eq:existence_moment_external_force}
    \bbE  \left[
    \left(\sum_{\bfx  \in \mathcal{P} \cap B(\mathbf{0}, R^\prime)^c} \mathds{1}_{\{\lVert F_{\mathcal{P} }^{(1,\infty )}(\bfx )\lVert_2 >\frac{r(\bfx ) - R}{|\varepsilon|}\}}(\bfx )\right)^m
    \right] < \infty.
\end{equation}
\begin{proof}[Proof of Corollary \ref{coro:existence_external_moments}]
    Fix a positive integer $m$.
    There exists $m$ constants $(d_i)_{i=1}^m$ such that
    \begin{align}
        &\bbE \left[\left(\sum_{\bfx  \in \mathcal{P}\cap B(\mathbf{0}, R^\prime)^c}
        \mathds{1}_{\{\| F^{(1, \infty)}_{\mathcal{P} }(\bfx )\|_2 > \frac{r(\bfx ) - R}{|\varepsilon|}\}}(\bfx )\right)^{m} \right]
        =
        \\
        &
        \sum_{k=1}^{m}d_k
        \underbrace{\bbE  \left[\sum_{\bfx _1, \dots, \bfx _k \in \mathcal{P}\cap B(\mathbf{0},R^\prime)^c}^{\neq}
            \mathds{1}_{\{\| F^{(1, \infty)}_{\mathcal{P} }(\bfx _1)\|_2 > \frac{r(\bfx _1) - R}{|\varepsilon|}\}}(\bfx _1)
            \dots
            \mathds{1}_{\{\| F^{(1, \infty)}_{\mathcal{P} }(\bfx _k)\|_2 > \frac{r(\bfx _k) - R}{|\varepsilon|}\}}(\bfx _k)\right]}_{\triangleq E_{k}}.
        \label{eq:eq_0_proof_coro_3}
    \end{align}
    Using Lemma \ref{lem:bound_external_force} we will show that for any $k \geq 1$ there exists positive constants $a_k, \, b_k$ and $c_k$ such that
    $$E_{k} \leq \left(\int_{B(\mathbf{0}, R^\prime)^c}
    a_k \exp \left[ -b_k g_{k}(\bfx ) \log( c_k g_{k}(\bfx )) \right] \rho \d \bfx \right)^{k},$$
    with $g_k(\bfx ) = \frac{r(\bfx ) - (R + k-1)}{ |\varepsilon| k} $ for $k \in \{1, \dots, m\}$.
    
    To simplify the notations, we denote $\mathcal{P} \setminus \{\bfx _1, \dots, \bfx _{k}\}$ by $\hat{\mathcal{P}}^{(k)}$ for any $k \in \{1, \dots, m\}$ and sometimes omit to remind that $\bfx _1, \dots, \bfx _{m} \in \mathcal{P}\cap B(\mathbf{0},  R^\prime )^c$ when it is clear from the context.
    
    First, remark that for two distinct points $\bfx $ and $\bfy$ of $\mathbb{R}^d$, we have
    \[
    F^{(1,\infty)}_{\mathcal{P}}(\bfx )=
    \begin{cases*}
        F^{(1,\infty)}_{\mathcal{P}\setminus\{\bfx , \bfy\}}(\bfx )& if  $\|\bfx  - \bfy\|_2<1$
        \\
        F^{(1,\infty)}_{\mathcal{P}\setminus\{\bfx , \bfy\}}(\bfx ) + \frac{\bfx -\bfy}{\|\bfx  - \bfy\|_2^d} & if $\|\bfx  - \bfy\|_2 \geq 1$
    \end{cases*}.
    \]%
    Thus
    $$\|F^{(1,\infty)}_{\mathcal{P}}(\bfx )\|_2 \leq \|F^{(1,\infty)}_{\mathcal{P}\setminus\{\bfx , \bfy\}}(\bfx )\|_2 + 1.$$
    In particular
    \begin{equation}
        \label{eq:bound_event_F_1_infty}
        \mathds{1}_{\left\{\|F^{(1,\infty)}_{\mathcal{P}}(\bfx )\|_2 \geq (r(\bfx )- R)/|\varepsilon|\right\}}(\bfx ) \leq \mathds{1}_{\left\{\|F^{(1,\infty)}_{\mathcal{P}\setminus\{\bfx , \bfy\}}(\bfx )\|_2 \geq 2 g_2(\bfx) \right\}}(\bfx ).
    \end{equation}
    Generalizing Equation \eqref{eq:bound_event_F_1_infty} to $k$ points gives
    \begin{equation}
        \label{eq:eq_1_proof_coro_3}
        \mathds{1}_{\left\{\|F^{(1,\infty)}_{\mathcal{P}}(\bfx )\|_2 \geq (r(\bfx )- R)/|\varepsilon|\right\}}(\bfx ) \leq \mathds{1}_{\left\{\|F^{(1,\infty)}_{\hat{\mathcal{P}}^{(k)}}(\bfx )\|_2 \geq k g_k(\bfx) \right\}}(\bfx ).
    \end{equation}
    Using Equation \eqref{eq:eq_1_proof_coro_3} and the extended Slivnyak-Mecke theorem \eqref{eq:slivnyak-mecke} we get
    \begin{align*}
        E_{k}
        &\leq
        \bbE \left[\sum_{\bfx _1, \dots, \bfx _{k}}^{\neq}
        \mathds{1}_{\{\| F^{(1, \infty)}_{\hat{\mathcal{P}}^{(k)}}(\bfx _1)\|_2 > kg_k(\bfx_1), \dots ,
            \| F^{(1, \infty)}_{\hat{\mathcal{P}}^{(k)}}(\bfx _{k})\|_2
            > kg_k(\bfx_k) \}}(\bfx _1, \dots, \bfx _{k})\right]
        \\
        &=
        \int_{(B(\mathbf{0}, R^\prime)^c)^{k}}
        \bbP\left(
        \| F^{(1, \infty)}_{\mathcal{P} }(\bfx _1)\|_2 > kg_k(\bfx_1) , \dots ,
        \| F^{(1, \infty)}_{\mathcal{P}}(\bfx _{k})\|_2 > k g_k(\bfx_k)
        \right)
        \rho^{k} \d \bfx _1 \dots \d \bfx _{k}
        \\
        &\leq
        \int_{(B(\mathbf{0},R^\prime)^c)^{k}}
        \min_{j \in \{1, \dots, k\}}
        \bbP \left(
        \|F^{(1,\infty)}_{\mathcal{P}}(\mathbf{x}_j)\|_2 \geq k g_k(\bfx_j) \right)
        \rho^{k} \d \bfx _1 \dots \d \bfx _{k}.
    \end{align*}
    As the distribution of $F^{(1,\infty)}_{\mathcal{P}}(\mathbf{x})$ is translation invariant (w.r.t. $\bfx$), we get
    \begin{align}
        E_{k}
        &\leq
        \int_{(B(\mathbf{0},R^\prime)^c)^{k}}
        \min_{j \in \{1, \dots, k\}}
        \bbP \left(
        \|F^{(1,\infty)}_{\mathcal{P}}(\mathbf{0})\|_2 \geq k g_k(\bfx_j) \right)
        \rho^{k} \d \bfx _1 \dots \d \bfx _{k}
        \nonumber
        \\
        &=
        \int_{(B(\mathbf{0},R^\prime)^c)^{k}}
        \bbP \left(
        \|F^{(1,\infty)}_{\mathcal{P}}(\mathbf{0})\|_2 \geq \max_{j \in \{1, \dots, k\}}k g_k(\bfx_j) \right)
        \rho^{k} \d \bfx _1 \dots \d \bfx _{k}
        \nonumber
        \\
        &\leq
        \int_{(B(\mathbf{0}, R^\prime)^c)^{k}}
        \bbP \left(
        \|F^{(1,\infty)}_{\mathcal{P}}(\mathbf{0})\|_2
        \geq \sum_{j=1}^{k} g_k(\bfx_j) \right)
        \rho^{k} \d \bfx _1 \dots \d \bfx _{k}.
        \label{eq:eq_2_proof_coro_3}
    \end{align}
    For $\mathbf{x_1}, \dots , \bfx _k \in B(\mathbf{0},R^\prime)^c$, Lemma~\ref{lem:bound_external_force} with $q=1$ and $p=\infty$ guarantees the existence of $C_1, \, C_2, \, C_3 >0$ such that
    \begin{align*}
        \bbP \left(
        \|F^{(1,\infty)}_{\mathcal{P}}( \mathbf{0} )\|_2
        \geq
        \sum_{j=1}^{k} g_k(\bfx_j)
        \right)
        &\leq
        C_1
        \exp \left(-C_2 \sum_{j=1}^{k} g_k(\bfx _j) \log \left(C_3 \sum_{j=1}^{k} g_k(\bfx _j)\right)\right)
        \\
        & =
        C_1
        \prod_{j=1}^{k}
        \exp \left(
        - C_2 g_k(\bfx _j)\log \left(
        C_3 \sum_{j=1}^{k} g_k(\bfx _j)
        \right)
        \right)
        \\
        & \leq
        C_1
        \prod_{j=1}^{k}
        \exp \left(
        - C_2 g_k(\bfx _j)\log \left(
        C_3 g_k(\bfx _j)
        \right)
        \right).
    \end{align*}
    Plugging back into \eqref{eq:eq_2_proof_coro_3}, and then in \eqref{eq:eq_0_proof_coro_3}, we obtain the existence of positive constants $(a_k)_{k=1}^m$, $(b_k)_{k=1}^m$, and $(c_k)_{k=1}^m$ such that
    \begin{align*}
        \bbE \left[ \left(\sum_{\bfx  \in \mathcal{P}\cap B(\mathbf{0},R^\prime)^c}
        \mathds{1}_{\{\| F^{(1, \infty)}_{\mathcal{P} }(\bfx )\|_2 > \frac{r(\bfx ) - R}{|\varepsilon|}\}}(\bfx )\right)^{m} \right]
        &\leq
        \sum_{k=1}^m d_k a_k
        \int_{(B(\mathbf{0},R^\prime )^c)^{k}}
        \prod_{j=1}^{k}
        \exp \left(
        - b_k g_k(\bfx _j)\log \left(
        c_k g_k(\bfx _j)
        \right)
        \right)
        \rho^{k} \d \bfx _1 \dots \d \bfx _{k}
        \\
        &=
        \sum_{k=1}^m d_k a_k
        \left(
        \int_{B(\mathbf{0}, R^\prime)^c}
        \exp \left(
        - b_k g_k(\bfx)\log \left(
        c_k g_k(\bfx)
        \right)
        \right)
        \rho ~\d \bfx  \right)^k,
    \end{align*}
    which concludes the proof.
\end{proof}

Now we switch focus to bounding the contribution from the internal force to \eqref{eq:decomposition_internal_external}.
Again, we work with a lemma and a corollary.
\begin{lemma}
    \label{lem:bound_internal_force}
    Consider a homogeneous Poisson point process $\mathcal{P} \subset \mathbb{R}^d$ of intensity $\rho$, with $d \geq 3$.
    Let $R > 0$ and $\varepsilon \in (-1,1)$.
    Consider a function $r: \, \mathbb{R}^d \rightarrow \mathbb{R}^{+}\setminus \{0\}$ such that $r(\bfx ) < \|\bfx \|_2$ for any $\bfx  \in B(\mathbf{0}, R)^c$.
    Then, for any $\bfx  \in B(\mathbf{0}, R)^c$, and $0<\gamma< 1/(d-1)$, there exists $c_1,\dots, c_5>0$ such that
    \begin{align}
        \label{eq:bound_proba_F_0_1}
        \bbP
        \left(
        \bfx  + \varepsilon F^{(0,1)}_{\mathcal{P} }(\mathbf{0}) \in B(\mathbf{0}, r(\bfx ))
        \right)
        &\leq
        \frac{c_1}
        {h(\bfx) g(\bfx )^{1 + \gamma}}
        +
        c_3\exp\left(-c_4g(\bfx )  \log( c_5 g(\bfx ))\right),
    \end{align}
    where $h(\bfx) = \max \left(1,c_2\left(\|\bfx \|_2/r(\bfx ) -1 \right)^{d-1} -1 \right)$, and $g(\bfx ) = (\|\bfx \|_2 -  r(\bfx ))/|\varepsilon|$.
\end{lemma}

First, by translation-invariance of the distribution of $F^{(q,p)}_{\mathcal{P} }( \bfx  )$, the choice of $\mathbf{0}$ in \eqref{eq:bound_proba_F_0_1} is arbitrary.
Second, by choosing $r(\bfx )= \|\bfx \|_2^{\beta}$, with $0<\beta<1/(2(d-1))$, the upper bound in Equation \eqref{eq:bound_proba_F_0_1} is $o(\|\bfx \|_2^{-d})$ as $\|\bfx\|_2$ goes to infinity. Thus, $\bbP
\left(
\bfx  + \varepsilon F^{(0,1)}_{\mathcal{P} }(\mathbf{0}) \in B\left(\mathbf{0}, r(\bfx )\right)
\right)$ is integrable over $B(\mathbf{0}, R)^c$.
\begin{proof}[Proof of Lemma \ref{lem:bound_internal_force}]
    \begin{figure}[!h]
        \centering
        \includegraphics[width=0.4\linewidth]{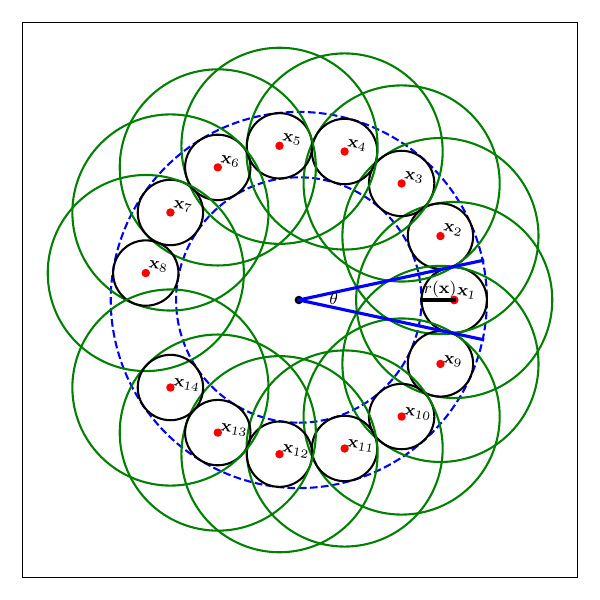}
        \caption{Illustration of the proof idea.}
        \label{fig:proof_idea_illustration}
    \end{figure}
    Fix $\bfx  \in B(\mathbf{0}, R)^c$, and let $q=\|\bfx \|_2 - r(\bfx )$ and $p=\|\bfx \|_2 + r(\bfx )$.
    Pick points $\{\bfx_i\}_{i=1}^m$ iteratively on the sphere $S(\mathbf{0}, \|\bfx\|_2)$ with $\bfx_1 = -\bfx$, such that the balls $\{B(\bfx_i , r(\bfx ))\}_{i=1}^m$ are disjoint, and it is not possible to add an additional similar ball not overlapping the previous ones.
    Next, let $B_i = B(\bfx_i , 3r(\bfx ))$ and $C_1\triangleq \kappa_d (3r(\bfx))^d$ be the volume of $B_1$. Then, $\bigcup_{i=1}^m B_i$ is a covering  of the annulus $A^{(q, p)}$.
    To see why this holds, suppose that there exists a point $\bfy \in A^{(q, p)} \setminus \bigcup_{i=1}^m B_i$. In particular, $\min_{i} \|\bfy - \bfx_i\|_2> 3r(\bfx)$. Take $\bfz = \bfy \frac{\|\bfx\|_2}{\|\bfy\|_2}$, $\bfz$ is in fact the orthogonal projection of $\bfy$ onto $S(\mathbf{0}, \|\bfx\|_2)$.
    For any $i \in \{1, \dots, m\}$
    \begin{align*}
        \|\bfz - \bfx_i\|_2 &\geq \|\bfy - \bfx_i\|_2 - \| \bfz - \bfy\|\\
        & > 3 r(\bfx) - \left|\|\bfy\|_2 - \|\bfx\|_2\right|\\
        & > 2 r(\bfx).
    \end{align*}
    Thus, we have $\bfz \in S(\mathbf{0}, \|\bfx\|_2)$ and $B(\bfz, r(\bfx)) \subset A^{(q, p)}\setminus \bigcup_{i=1}^m B(\bfx_i , r(\bfx ))$ giving a contradiction.
    While our work pertains to dimensions greater than two, employing a visualization that depicts the two-dimensional case can aid in comprehending the concept.
    Figure \ref{fig:proof_idea_illustration} illustrates an example of valid points
    $\{\bfx_i\}_{i=1}^m$ in red and the covering in green of $A^{(q, p)}$ for $d=2$.
    Note that, for $d=2$ and any $i\in \{1, ..., m\}$, $\bfx _i$ can be chosen inductively as the rotation of $\bfx _{i-1}$ of angle $\theta= \mathrm{Arcsin}(r(\bfx )/\|\bfx \|_2)$ around the origin.
    
    As the balls $\{B(\bfx_i , r(\bfx ))\}_{i=1}^m$ are disjoint and contained in $A^{(q, p)}$, we have
    \begin{align*}
        \bbP
        \left( \varepsilon F^{(0,1)}_{\mathcal{P}}(\mathbf{0})  \in A^{(q, p)} \right)
        &\geq
        \bbP \left( \varepsilon F^{(0,1)}_{\mathcal{P}}(\mathbf{0})  \in \bigcup_{i=1}^m B(\bfx _i, r(\bfx ))\right)
        \\
        &=\sum_{i=1}^m
        \bbP \left(
        \varepsilon
        F^{(0,1)}_{\mathcal{P}}( \mathbf{0})  \in B(\bfx _i, r(\bfx ) )
        \right)
    \end{align*}
    By isotropy of the law of $F^{(0,1)}_{\mathcal{P} } (\mathbf{0})$, we obtain
    \begin{equation}
        \label{eq:F_with_m}
        \bbP\left(\varepsilon F^{(0,1)}_{\mathcal{P} }(\mathbf{0}) \in A^{(q, p)} \right)
        \geq m \bbP \left(\varepsilon F^{(0,1)}_{\mathcal{P} }(\mathbf{0}) \in B(\bfx_1 , r(\bfx ))\right).
    \end{equation}
    We will now proceed to find a suitable lower bound for $m$.
    To accomplish this, we apply the mean value theorem to
    \begin{align*}
        h \colon [\|\bfx \|_2 - r(\bfx ), \|\bfx \|_2 + r(\bfx )] &\to \mathbb{R} \\
        x &\mapsto x^d,
    \end{align*}
    and we get that there exists $c  \in (-1, 1)$
    such that
    $$\left(\|\bfx \|_2 + r(\bfx )\right)^d - (\|\bfx \|_2 - r(\bfx ))^d = 2d r(\bfx ) (\|\bfx \|_2 + c r(\bfx ))^{d-1} \geq 2d r(\bfx ) (\|\bfx \|_2 - r(\bfx ))^{d-1}.$$
    Hence
    \begin{align*}
        m &
        \geq \floor*{\frac{\leb{A^{(q, p)}}}{C_1}}\\
        &= \floor*{\frac{\kappa_d \left((\|\bfx \|_2 + r(\bfx ))^d - (\|\bfx \|_2 - r(\bfx ))^d \right)}{\kappa_d (3r(\bfx ))^d}}
        \\
        &\geq
        \floor*{\frac{\kappa_d\left(2d r(\bfx ) (\|\bfx \|_2 - r(\bfx ))^{d-1}\right)}{\kappa_d (3r(\bfx ))^d}} \geq \max \left(1,C_2\left(\frac{\|\bfx \|_2}{r(\bfx )} -1 \right)^{d-1} -1 \right),
    \end{align*}
    where $C_2= (2d)/3^{d}$.
    Thus Equation \eqref{eq:F_with_m} leads to
    \begin{align*}
        \bbP\left(\varepsilon F^{(0,1)}_{\mathcal{P} }( \mathbf{0}) \in B(\bfx_1 , r(\bfx ))\right)
        % & \leq
        % \frac{\bbP\left( \varepsilon F^{(0,1)}_{\mathcal{P} }(\mathbf{0}) \in A^{(q, p)} \right)}{m}
        % \\
        &\leq
        \frac{\bbP \left(\varepsilon F^{(0,1)}_{\mathcal{P} }(\mathbf{0}) \in A^{(q, p)} \right)}{\max \left(1,C_2\left(\frac{\|\bfx \|_2}{r(\bfx )} -1 \right)^{d-1} -1 \right)}
        \\
        &\leq
        \frac{\bbP\left(  \|\varepsilon F^{(0,1)}_{\mathcal{P} }(\mathbf{0})\|_2 > \|\bfx \|_2 -  r(\bfx ) \right)}{\max \left(1,C_2\left(\frac{\|\bfx \|_2}{r(\bfx )} -1 \right)^{d-1} -1 \right)}.
    \end{align*}
    Since $F^{(0,1)}_{\mathcal{P} } = F_{\mathcal{P} } - F^{(1,\infty)}_{\mathcal{P} }$, we get
    \begin{align}
        \bbP\left(\varepsilon F^{(0,1)}_{\mathcal{P} }(\mathbf{0}) \in B(\bfx_1 , r(\bfx ))\right)& \leq
        \frac{\bbP\left(|\varepsilon| \|F_{\mathcal{P} }(\mathbf{0}) - F^{(1, \infty)}_{\mathcal{P} }(\mathbf{0})\|_2 > \|\bfx \|_2 -  r(\bfx ) \right)}{\max \left(1,C_2\left(\frac{\|\bfx \|_2}{r(\bfx )} -1 \right)^{d-1} -1 \right)}
        \nonumber
        \\
        &\leq
        \frac{\bbP\left(\|F_{\mathcal{P} }(\mathbf{0}) \|_2 >
            \frac{\|\bfx \|_2 -  r(\bfx )}{2 |\varepsilon|}\right) + \bbP\left(\|F^{(1, \infty)}_{\mathcal{P} }(\mathbf{0})\|_2 > \frac{\|\bfx \|_2 -  r(\bfx )}{2 |\varepsilon|} \right)}{\max \left(1,C_2\left(\frac{\|\bfx \|_2}{r(\bfx )} -1 \right)^{d-1} -1 \right)}
        \label{eq:eq_1}
    \end{align}
    By Lemma \ref{prop:distribution_of_F}, for any $0<\gamma< 1/(d-1) $ we have $\bbE  [\|F_{\mathcal{P} }(\mathbf{0})\|_2^{1 + \gamma}] < \infty$.
    Applying Markov's inequality to the first part of Equation  \eqref{eq:eq_1}, and using Lemma \ref{lem:bound_external_force} with $t= (\|\bfx \|_2 -  r(\bfx ))/(2 |\varepsilon|)$ for the last part, we obtain the existence of positive constants $c_1$, $c_2$, and $c_3$ such that
    \begin{align*}
        \bbP & \left(\bfx  + \varepsilon F^{(0,1)}_{\mathcal{P} }(\mathbf{0}) \in B(\mathbf{0}, r(\bfx ))\right)
        \\
        % &\leq
        % \frac{\bbP\left(\|F_{\mathcal{P} }(\mathbf{0}) \|_2 > (\|\bfx \|_2 -  r(\bfx ))/(2 |\varepsilon|)\right) +
            % \bbP\left(\|F^{(1, \infty)}_{\mathcal{P} }(\mathbf{0})\|_2 > (\|\bfx \|_2 -  r(\bfx ))/(2 |\varepsilon|) \right)}{\max \left(1,C_2\left(\frac{\|\bfx \|_2}{r(\bfx )} -1 \right)^{d-1} -1 \right)}
        % \\
        &\leq
        \frac{(2 |\varepsilon|)^{1 + \gamma}\bbE \left[\|F_{\mathcal{P} }(\mathbf{0})\|_2^{ 1 + \gamma}\right]}
        {\max \left(1,C_2\left(\frac{\|\bfx \|_2}{r(\bfx )} -1 \right)^{d-1} -1 \right)\left(\|\bfx \|_2 -  r(\bfx )\right)^{1 + \gamma}}
        +
        \frac{c_1 \exp\left(-c_2\frac{\|\bfx \|_2 -  r(\bfx )}{2 |\varepsilon|} \log\left(c_3 \frac{\|\bfx \|_2 -  r(\bfx )}{2 |\varepsilon|} \right)\right) }
        {\max \left(1,C_2\left(\frac{\|\bfx \|_2}{r(\bfx )} -1 \right)^{d-1} -1 \right)}
        \\
        &\leq
        \frac{C_3 |\varepsilon|^{1 + \gamma}}
        {\max \left(1,C_2\left(\frac{\|\bfx \|_2}{r(\bfx )} -1 \right)^{d-1} -1 \right)\left(\|\bfx \|_2 -  r(\bfx )\right)^{1 + \gamma}}
        +
        C_4\exp\left(-C_5\frac{\|\bfx \|_2 -  r(\bfx )}{|\varepsilon|} \log\left( C_6 \frac{\|\bfx \|_2 -  r(\bfx )}{|\varepsilon|}\right)\right)
    \end{align*}
    with $C_3 = 2^{1 + \gamma} \bbE \left[\|F_{\mathcal{P} }(\mathbf{0})\|_2^{1 + \gamma}\right]$, $C_4=c_1$, $C_5 = c_2/2$ and $C_6 = c_3/2$.
    % Hence, for $\|\bfx \|_2$ large enough there exist $C_7, \dots, C_{11} >0$ s.t.
    % \begin{align*}
        %     \bbP( \varepsilon F^{(0,1)}_{\mathcal{P} }(\mathbf{0}) \in B(\bfx , f(\bfx )))
        %     &\leq
        %     \frac{C_7 \varepsilon^{1+ \frac{1}{2(d-1)}} f(\bfx )^{d-1 }}
        %     {\|\bfx \|_2^{d+ \frac{1}{2(d-1)}}}
        %     +
        %     \frac{C_{10} f(\bfx )^{d-1} \exp\big(-C_8 \frac{\|\bfx \|_2}{\varepsilon}  \log(C_9 \frac{\|\bfx \|_2}{\varepsilon})\big) }
        %     {\|\bfx \|_2^{d-1}}
        %     \\
        %     & \leq
        %     C_{11} \varepsilon^{1+ \frac{1}{2(d-1)}} \frac{f(\bfx )^{d-1 }}
        %     {\|\bfx \|_2^{d + \frac{1}{2(d-1)}}}
        % \end{align*}
    % which ends our proof.
\end{proof}
Note that under the assumptions and definitions of Lemmas \ref{lem:bound_external_force}, \ref{lem:bound_internal_force}, for $\bfx  \in B(\mathbf{0}, R^{1/\beta})^c$, there exists positive constants $c_1, \dots, c_5$ such that
\begin{equation}
    \label{eq:Proba_far_point_enter}
    \bbP(\bfx  + \varepsilon F_{\mathcal{P} }(\bfx) \in B(\mathbf{0}, R))
    \leq
    \frac{c_1}
    {\max \left(1,c_2\left(\frac{\|\bfx \|_2}{r(\bfx )} -1 \right)^{d-1} -1 \right) g(\bfx )^{1 + \gamma}}
    +
    c_3 \exp\left(-c_4g(\bfx )  \log( c_5 g(\bfx ))\right),
\end{equation}
for any $0< \gamma < 1/(d-1)$.
As expected, we observe that
$$\bbP(\bfx  + \varepsilon F_{\mathcal{P} } \in B(\mathbf{0}, R)) \xrightarrow[\varepsilon \rightarrow 0]{} 0.$$
We will later see that by selecting an appropriate function $r$, typically for $r(\bfx )=\|\bfx \|_2^\beta $ with $0<\beta<\frac{\gamma}{d-1} $, the bound in Equation \eqref{eq:Proba_far_point_enter} converges fast enough to zero allowing to bound the moments of $\sum_{\bfx  \in \Pi_{\varepsilon}\mathcal{P}} \mathds{1}_{B(\mathbf{0},R)}(\bfx )$.

The next result is a corollary of Lemma \ref{lem:bound_internal_force}.
\begin{corollary}
    \label{coro:existence_internal_moments}
    Consider a homogeneous Poisson point process $\mathcal{P} \subset \mathbb{R}^d$ of intensity $\rho$, with $d \geq 3$.
    Let $ R > 0$, $\varepsilon \in (-1,1)$ and $\beta \in (0,1)$.
    Set $r(\bfx ) = \lVert \bfx  \lVert_2^{\beta}$, $g(\bfx ) = (\|\bfx \|_2 -  r(\bfx ))/|\varepsilon|$ and denote
    \begin{equation}
        \label{eq:def_E_k}
        E_k \triangleq
        \bbE \left[\sum_{\bfx _1, \dots, \bfx _k \in \mathcal{P}\cap B(\mathbf{0}, R)^c}^{\neq}
        \mathds{1}_{B(\mathbf{0},r(\bfx _1))}(\bfx _1 + \varepsilon F^{(0,1)}_{\mathcal{P} }(\bfx _1))
        \dots
        \mathds{1}_{B(\mathbf{0},r(\bfx _k))}(\bfx _k + \varepsilon F^{(0,1)}_{\mathcal{P}}(\bfx _k))\right].
    \end{equation}
    For any positive integer $m$ and $0<\gamma<1/(d-1)$ there exists positive constants $a_1, \dots, a_{m-1}$ and $c_1, \dots, c_5$ such that
    \begin{align}
        \label{eq:E_1_coro_4}
        E_1
        \leq
        \int_{B(\mathbf{0}, R)^c}
        \frac{c_1 }
        {\max \left(1, c_2\left(\|\bfx \|_2^{1-\beta} -1\right)^{d-1} - 1\right)g(\bfx )^{1 + \gamma}}
        % \\
        % \nonumber
        % & \hspace*{2cm}
        +
        c_3\exp\left(-c_4g(\bfx ) \log\left( c_5 g(\bfx )\right)\right) \rho \d \bfx,
    \end{align}
    and
    \begin{equation}
        \label{eq:existence_of_internal_moments}
        \bbE \left[\left(\sum_{\bfx  \in \mathcal{P}\cap B(\mathbf{0}, R)^c} \mathds{1}_{B(\mathbf{0},f(\bfx ))}( \bfx +\varepsilon F^{(0,1)}_{\mathcal{P} }(\bfx ))\right)^{m} \right] \leq
        \sum_{k=1}^{m -1}\left(a_k E_k + b^{k-1} E_1E_{m-k}\right) +  b^{m-1} E_1 .
    \end{equation}
    where $b=2^d \kappa_d \rho$.
\end{corollary}
Taking $0<\beta<\frac{\gamma}{d-1} $, the integrand in Equation \eqref{eq:E_1_coro_4} is $ o(
\|\bfx \|_2^{-d})$ as $\|\bfx \|_2 \rightarrow \infty$ implying that
\begin{equation}
    \label{eq:existence_moment_internal_force}
    \bbE \left[\left(\sum_{\bfx  \in \mathcal{P}\cap B(\mathbf{0}, R)^c} \mathds{1}_{B(\mathbf{0},r(\bfx ))}( \bfx +\varepsilon F^{(0,1)}_{\mathcal{P} }(\bfx ))\right)^{m} \right] < \infty ,
\end{equation}
for any positive integer $m$. This proves that all terms in  \eqref{e:three_term_decomposition} are indeed finite, thereby concluding the proof of Proposition  \ref{prop:extistance_of_the_moments}.

\begin{proof}[Proof of Corollary \ref{coro:existence_internal_moments}]
    We will show the validity of Equation \eqref{eq:existence_of_internal_moments} by induction on $m \geq 1 $.
    
    To begin with, for $m=1$ using the extended Slivnyak-Mecke theorem \eqref{eq:slivnyak-mecke} and that the law of $F^{(0,1)}_{\mathcal{P} }(\bfx )$ is translation-invariant (w.r.t. $\bfx $) we have
    $$ E_1 = \bbE
    \left[
    \sum_{\bfx  \in \mathcal{P}\cap B(\mathbf{0},R)^c}
    \mathds{1}_{B(\mathbf{0},r(\bfx ))}(\bfx
    +
    \varepsilon F^{(0,1)}_{\mathcal{P} \setminus \{\bfx\} }(\bfx )) \right]
    =
    \int_{B(\mathbf{0}, R)^c}
    \bbP
    \left( \bfx
    +
    \varepsilon F^{(0,1)}_{\mathcal{P} }(\mathbf{0}) \in B(\mathbf{0} , r(\bfx))
    \right)
    \rho \d \bfx .$$
    By Lemma \ref{lem:bound_internal_force}, for any $0<\gamma<1/(d-1)$ there exists positive constants $c_1, \dots , c_5$ such that
    \begin{align*}
        E_1
        \leq
        \int_{B(\mathbf{0}, R)^c}
        \frac{c_1}
        {\max \left(1, c_2\left(\|\bfx \|_2^{1-\beta} -1\right)^{d-1} - 1\right)g(\bfx )^{1 + \gamma}}
        +
        c_3 \exp\left(-c_4g(\bfx) \log( c_5 g(\bfx )) \right) \rho \d \bfx .
    \end{align*}
    
    Next, suppose that Equation \eqref{eq:existence_of_internal_moments} is valid until $m$, and let's verify that it holds for $m+1$.
    There exists a sequence of constants $(d_k)_{k=1}^m$ such that
    \begin{align}
        &\bbE \left[\left(\sum_{\bfx  \in \mathcal{P}\cap B(\mathbf{0}, R)^c} \mathds{1}_{B(\mathbf{0},r(\bfx ))}(\bfx  + \varepsilon F^{(0,1)}_{\mathcal{P} }(\bfx ))\right)^{m+1} \right]
        =
        \sum_{k=1}^m d_k E_k + E_{m+1}.
        \label{eq:decomposition_moment}
    \end{align}
    %By the induction hypothesis, all the terms except $A$ in the above Equation are finite.
    %It remains to show that $A< \infty$.
    We only need to focus on finding an upper bound of $E_{m+1}$.
    
    To simplify the notations, we denote $\mathcal{P} \setminus \{\bfx _1, \dots, \bfx _{k}\}$ by $\hat{\mathcal{P}}^{(k)}$ for any $k \in \{1, \dots, m\}$ and sometimes omit to remind that $\bfx _1, \dots, \bfx _{m} \in \mathcal{P}\cap B(\mathbf{0},  R )^c$ when it is clear from the context.
    
    To begin, we break down the sum that defines $E_{m+1}$ according to the count $k$ of the points of $\mathcal{P}\cap B(\mathbf{0}, R)^c$ that fall within a ball of radius $2$ centered at ${\bf x_{1}}$ as follows
    \begin{align*}
        E_{m+1}&=
        \underbrace{
            \bbE \left[\sum_{\substack{\bfx _1, \dots, \bfx _{m+1} \\
                    \max_{1 < i } \limits \|\bfx _i - \bfx _1\|_2 \leq 2} }^{\neq}
            \mathds{1}_{B(\bfx _1,r(\bfx _1))}(- \varepsilon F^{(0,1)}_{\mathcal{P} }(\bfx _1))
            \dots
            \mathds{1}_{B(\bfx _{m+1},r(\bfx _{m+1}))}(- \varepsilon F^{(0,1)}_{\mathcal{P} }(\bfx _{m+1}))\right]}_{A_{m+1}}
        \\
        & \hspace*{6cm} +
        \\
        &
        \sum_{k=1}^{m}
        \underbrace{
            \bbE \left[\sum_{\substack{\bfx _1, \dots, \bfx _{m+1} \\
                    \max_{1<i \leq k} \limits \|\bfx _i - \bfx _1\|_2 \leq 2 \\
                    \min_{k<i \leq m+1 } \limits \|\bfx _i - \bfx _1\|_2 > 2
            }}^{\neq}
            \mathds{1}_{B(\bfx _1,r(\bfx _1))}( - \varepsilon F^{(0,1)}_{\mathcal{P}}(\bfx _1))
            \dots
            \mathds{1}_{B(\bfx _{m+1},r(\bfx _{m+1}))}(- \varepsilon F^{(0,1)}_{\mathcal{P}}(\bfx _{m+1}))\right]
        }_{A_{k}   }.
    \end{align*}
    First, for $A_{m+1}$ the extended Slivnyak-Mecke theorem \eqref{eq:slivnyak-mecke} followed by a change of variables gives
    \begin{align*}
        A_{m+1} &
        \leq \bbE \left[\sum_{\substack{\bfx _1, \dots, \bfx _{m+1} \\
                \max_{1<i} \limits \|\bfx _i + \bfx _1\| \leq 2} }^{\neq}
        \mathds{1}_{B(\mathbf{0},r(\bfx _1))}\left(\bfx _1 + \varepsilon F^{(0,1)}_{\hat{\mathcal{P}}^{(m+1)} }(\bfx _1) + \varepsilon \sum_{j=2}^{m+1} \frac{ \bfx _1 - \bfx _j}{\|\bfx _1 - \bfx _j\|_2^d} \right)
        \right]
        \\
        &=
        \int_{B(\mathbf{0}, R)^c}\int_{ B(\bfx _1, 2)^{m}}
        \bbP \left(
        \bfx _1 + \varepsilon F^{(0,1)}_{\mathcal{P}}(\bfx _1) \in B \left(\varepsilon \sum_{j=2}^{m+1} \frac{\bfx _j - \bfx _1}{\|\bfx _j - \bfx _1\|_2^d} ,r(\bfx _1)\right)
        \right)
        \rho^{m+1} \d \bfx _{m+1} \dots \d \bfx _{1}
        \\
        &=
        \int_{ B(\mathbf{0}, 2)^{m}\times B(\mathbf{0}, R)^c}
        \bbP \left(
        \bfx _1 + \varepsilon F^{(0,1)}_{\mathcal{P}}(\bfx _1) \in B \left(\varepsilon \sum_{j=1}^m \frac{\mathbf{t}_j}{\|\mathbf{t}_j\|_2^d} ,r(\bfx _1)\right)
        \right)
        \rho^{m+1} \d \bfx _{1} \d \mathbf{t}_{1} \dots \d \mathbf{t}_m.
    \end{align*}
    By the definition of the first intensity measure \eqref{eq:campbell_theorem} the last equation is equal to
    $$\int_{ B(\mathbf{0}, 2)^{m}}
    \bbE \left[
    \sum_{\bfx  \in \mathcal{P} \cap B(\mathbf{0}, R)^c}
    \mathds{1}_{B \left(\varepsilon \sum_{j=1}^m \frac{\mathbf{t}_j}{\|\mathbf{t}_j\|_2^d} ,r(\bfx )\right)}
    (\bfx  + \varepsilon F^{(0,1)}_{\mathcal{P} \setminus \{\bfx \}}(\bfx ))
    \right]
    \rho^{m} \d \mathbf{t}_{1} \dots \d \mathbf{t}_m. $$
    Employing further the stationarity of $\Pi_{\varepsilon}^{(0,1)} \mathcal{P}$, we get
    \begin{align}
        A_{m+1} &
        \leq
        \int_{ B(\mathbf{0}, 2)^{m}}
        \bbE \left[
        \sum_{\bfx  \in \mathcal{P} \cap B(\mathbf{0}, R)^c}
        \mathds{1}_{B (\mathbf{0} ,r(\bfx ))}
        (\bfx  + \varepsilon F^{(0,1)}_{\mathcal{P}}(\bfx ))
        \right]
        \rho^{m} \d \mathbf{t}_{1} \dots \d \mathbf{t}_m
        \nonumber
        \\
        &= (2^d \kappa_d \rho)^m E_1 .
        \label{eq:bound_A_lastone}
    \end{align}
    
    Second, for $k \in \{ 2, \dots, m \}$, we have
    \begin{align*}
        A_k
        \leq
        \bbE \Bigg [&
        \sum_{\substack{\bfx _1 \in B(\mathbf{0}, R)^c\\
                \bfx _2, \dots, \bfx _k \in B(\bfx _1, 2)
        }}^{\neq}
        \Bigg(
        \mathds{1}_{B (\mathbf{0} ,r(\bfx _1))}
        \big(\bfx _1 + \varepsilon F^{(0,1)}_{\hat{\mathcal{P}}^{(k)}}(\bfx _1) + \varepsilon \sum_{j=2}^k \limits \frac{ \bfx _1 - \bfx _j}{\|\bfx _1 - \bfx _j\|_2^d}\big)
        \\
        &
        \sum_{\substack{\bfx _{k+1}, \dots, \bfx _{m+1} \in  B(\bfx _1, 2)^c
        }}^{\neq}
        \mathds{1}_{B(\mathbf{0},r(\bfx _{k+1}))}(\bfx _{k+1} + \varepsilon F^{(0,1)}_{\hat{\mathcal{P}}^{(k+1)}}(\bfx _{k+1}))
        \\
        &\hspace{5cm}
        \dots
        \mathds{1}_{B(\mathbf{0},r(\bfx _{m+1}))}(\bfx _{m+1} + \varepsilon F^{(0,1)}_{\hat{\mathcal{P}}^{(k)} \setminus \{\bfx _{m+1}\}}(\bfx _{m+1}))
        \Bigg)
        \Bigg].
    \end{align*}
    Applying the extended Slivnyak-Mecke theorem \eqref{eq:slivnyak-mecke} then, using the independence of $F^{(0,1)}_{\mathcal{P} }(\bfx )$ and $F^{(0,1)}_{\mathcal{P}}(\bfy)$ whenever $\|\bfx  - \bfy\|_2>2$ yield
    \begin{align*}
        A_k
        &\leq
        \int_{B(\mathbf{0}, R)^c\times B(\bfx _1, 2)^{k-1}}
        \bbE \Bigg [
        \mathds{1}_{B (\mathbf{0} ,r(\bfx _1))}
        \left(\bfx _1 + \varepsilon F^{(0,1)}_{\mathcal{P}}(\bfx _1) + \varepsilon \sum_{j=2}^k \limits \frac{\bfx _1 - \bfx _j}{\|\bfx _1 - \bfx _j\|_2^d}\right)
        \\
        &\hspace{3cm}
        \sum_{\substack{\bfx _{k+1}, \dots, \bfx _{m+1} \in B(\bfx _1, 2)^c
        }}^{\neq}
        \mathds{1}_{B(\mathbf{0},r(\bfx _{k+1}))}\left(\bfx _{k+1} + \varepsilon F^{(0,1)}_{ \mathcal{P} \setminus \{\bfx _{k+1}\}}(\bfx _{k+1})\right)
        \\
        &\hspace{5cm}
        \dots
        \mathds{1}_{B(\mathbf{0},r(\bfx _{m+1}))}\left(\bfx _{m+1} + \varepsilon F^{(0,1)}_{\mathcal{P} \setminus \{\bfx _{m+1}\}}(\bfx _{m+1})\right) \Bigg]
        \rho^k
        \d \bfx _{k}\dots \d \bfx _{1}
        \\
        &=
        \int_{B(\mathbf{0}, R)^c\times B(\bfx _1, 2)^{k-1}}
        \bbE \left[
        \mathds{1}_{B (\mathbf{0} ,r(\bfx _1))}
        \left(\bfx _1 + \varepsilon F^{(0,1)}_{\mathcal{P}}(\bfx _1) + \varepsilon \sum_{j=2}^k \limits \frac{ \bfx _1 - \bfx _j}{\|\bfx _1 - \bfx _j\|_2^d}\right) \right]
        \\
        &\hspace{3cm}
        \bbE \Big[\sum_{\substack{\bfx _{k+1}, \dots, \bfx _{m+1} \in
                B(\bfx _1, R)^c
        }}^{\neq}
        \mathds{1}_{B(\mathbf{0},r(\bfx _{k+1}))}\left(\bfx _{k+1} + \varepsilon F^{(0,1)}_{\mathcal{P}}(\bfx _{k+1})\right)
        \\
        &\hspace{5cm}
        \dots
        \mathds{1}_{B(\mathbf{0},r(\bfx _{m+1}))}\left(\bfx _{m+1} + \varepsilon F^{(0,1)}_{\mathcal{P}}(\bfx _{m+1})\right) \Big]
        \rho^k
        \d \bfx _{k}\dots \d \bfx _{1}
        \\
        & =
        E_{m + 1 -k}
        \times
        \int_{B(\mathbf{0}, R)^c\times B(\bfx _1, 2)^{k-1}}
        \bbE \left[
        \mathds{1}_{B (\mathbf{0} ,r(\bfx _1))}
        \left(\bfx _1 + \varepsilon F^{(0,1)}_{\mathcal{P}}(\bfx _1) + \varepsilon \sum_{j=2}^k \limits \frac{\bfx _1 - \bfx _j}{\|\bfx _1 - \bfx _j\|_2^d}\right) \right]
        \rho^k
        \d \bfx _{k} \dots \d \bfx _{1} .
    \end{align*}
    Following the same technique used to bound $A_{m+1}$ we get
    \begin{align}
        A_k &\leq
        %E_{m + 1 -k} \int_{B(\mathbf{0}, 2)^{k-1} \times B(\mathbf{0}, R)^c}
        % \bbE \left[
        % \mathds{1}_{B (\mathbf{0} ,r(\bfx _1))}
        % \left(\bfx _1 + \varepsilon F^{(0,1)}_{\mathcal{P}}(\bfx _1) + \varepsilon \sum_{j=1}^{k-1} \limits \frac{\mathbf{t}_j}{\|\mathbf{t}_j\|_2^d}\right) \right]
        % \rho^k
        % \d \bfx _{1} \d\mathbf{t}_1 \dots \d\mathbf{t}_{k-1}
        %\\
        %&=
        E_{m + 1 -k} \int_{B(\mathbf{0}, 2)^{k-1} }
        \bbE \left[ \sum_{\bfx  \in \mathcal{P} \cap B(\mathbf{0}, R)^c}
        \mathds{1}_{B (\varepsilon \sum_{j=1}^{k-1} \limits \frac{\mathbf{t}_j}{\| \mathbf{t}_j\|_2^d},r(\bfx ))}
        \left(\bfx  + \varepsilon F^{(0,1)}_{\mathcal{P}}(\bfx )\right) \right]
        \rho^{k-1}
        \d\mathbf{t}_1 \dots \d\mathbf{t}_{k-1}
        \nonumber
        \\
        &=
        E_{m + 1 -k} \int_{B(\mathbf{0}, 2)^{k-1} }
        \bbE \left[ \sum_{\bfx  \in \mathcal{P} \cap B(\mathbf{0}, R)^c}
        \mathds{1}_{B (\mathbf{0},r(\bfx ))}
        \left(\bfx  + \varepsilon F^{(0,1)}_{\mathcal{P}}(\bfx )\right)\right]
        \rho^{k-1}
        \d\mathbf{t}_1 \dots \d\mathbf{t}_{k-1}
        \nonumber
        \\
        &=
        (2^d \kappa_d \rho)^{k-1} E_1 E_{m + 1 -k}.
        \label{eq:bound_A_k}
    \end{align}
    Inserting \eqref{eq:bound_A_lastone}, and \eqref{eq:bound_A_k} in \eqref{eq:decomposition_moment} we get
    \begin{align*}
        \bbE \left[\left(\sum_{\bfx  \in \mathcal{P}\cap B(\mathbf{0}, R)^c} \mathds{1}_{B(\mathbf{0},r(\bfx ))}( \bfx +\varepsilon F^{(0,1)}_{\mathcal{P}}(\bfx ))\right)^{m+1} \right]
        & \leq
        &=
        A_{m+1} + \sum_{k=1}^m d_k E_k + A_{k}\\
        c_1^{m} E_1 +
        \sum_{k=1}^m \left(d_k E_k + c_1^{k-1} E_1E_{m+1 -k}\right) ,
    \end{align*}
    with $c_1=2^d \kappa_d \rho$, which concludes the proof.
\end{proof}
%------------------------------------
%------------------------------------
\subsection{Proof of Theorem \ref{thm:Variance_reduction} (variance reduction)} % (fold)
\label{subs:proof_variance_reduction}

In this section we prove the main result of the paper, Theorem \ref{thm:Variance_reduction}.
We start with is a lemma from the theory of harmonic functions and a lemma related to bounding moments of the counting measure of a Poisson point processes.
Then we state and prove a few technical lemmas on the behavior of various quantities relevant to the variance of linear statistics under $\Pi_\varepsilon\calP$.
Finally, we put these lemmas in action in the proof of Theorem~\ref{thm:Variance_reduction}.

\begin{lemma}
    \label{lem:integral_with_the_harmonic_fct}
    Let $d>2$ and $g\in C^2(\mathbb{R}^d)$ have compact support.
    For all $\bfx  \in \mathbb{R}^d$, we have
    \begin{equation}
        \label{eq:integral_with_force}
        g(\bfx )  = \frac{1}{(2-d)d \kappa_d} \int_{\mathbb{R}^d} \Delta g(\bfy) \frac{1}{\|\bfx -\bfy\|^{d-2}_2} \d \bfy
        = \frac{1}{d \kappa_d} \int_{\mathbb{R}^d} \nabla g(\bfy) \cdot \frac{\bfx -\bfy}{\|\bfx -\bfy\|^d_2} \d \bfy.
    \end{equation}
\end{lemma}
The first equality in \eqref{eq:integral_with_force} can be found in \cite[Chapter 9]{AxlBourRam2001}.
We thus only prove the second equality, which results from an appropriate integration by parts.

\begin{proof}[Proof of Lemma \ref{lem:integral_with_the_harmonic_fct}]
    First, recall Green's ``integration by parts'' formula \citep[Chapter 1]{AxlBourRam2001}
    \begin{align}
        \label{eq:eq_1_proof_lem_4}
        \int_{A} f(\bfx ) \Delta g(\bfx ) \d \bfx &
        = \int_{\partial A} (f \nabla g . \mathbf{n} + g \nabla f . \mathbf{n})~ \d S - \int_{A} \nabla f(\bfx ) \cdot \nabla g(\bfx ) \d \bfx,
    \end{align}
    where $A$ is a bounded subset of $\mathbb{R}^d$ with smooth boundary, $f$ and $g$ are $C^2$ on a neighborhood of $\bar{A}$, $\mathbf{n}$ is the outward unit normal vector, and $S$ is the surface-area measure on $\partial A$.
    Let the support of $g$ be a compact $K \subset A$, \eqref{eq:eq_1_proof_lem_4} simplifies and we have
    \begin{equation}
        \label{eq:integration_by_part_high_dim}
        \int_{A} f(\bfx ) \Delta g(\bfx ) \d \bfx
        = \int_{\partial A} f \nabla g . \mathbf{n} ~ \d S - \int_{A} \nabla f(\bfx ) \cdot \nabla g(\bfx ) \d \bfx .
    \end{equation}
    Let $\bfx  \in \mathbb{R}^d$, and let $R>0$ be large enough such that $\bfx $ and $K$ are contained in $B(\mathbf{0}, R)$.
    Then using the dominated convergence we get
    \begin{equation}
        \frac{1}{(2-d)d \kappa_d} \int_{\mathbb{R}^d} \frac{\Delta g(\bfy )}{\|\bfx -\bfy \|^{d-2}_2} \d \bfy  = \lim_{\varepsilon \rightarrow 0} \frac{1}{(2-d)d \kappa_d} \int_{B(\mathbf{0}, R) \setminus B(\bfx , \varepsilon)}  \frac{\Delta g(\bfy )}{\|\bfx -\bfy \|^{d-2}_2} \d \bfy.
        \label{eq:tool}
    \end{equation}
    Now, using \eqref{eq:integration_by_part_high_dim} with $A = B(\mathbf{0}, R) \setminus B(\bfx , \varepsilon)$ and $f: \bfy\mapsto \frac{1}{\|\bfx -\bfy \|^{d-2}_2}$, it comes
    %and that the support of $g$ is a subset of $\partial B(\mathbf{0}, R)$
    \begin{align*}
        \int_{A} \frac{\Delta g(\bfy )}{\|\bfx -\bfy \|^{d-2}_2} \d \bfy
        % \int_{\partial B(\bfx , \varepsilon)}  \frac{1}{\|\bfx -\bfy \|^{d-2}_2}  \nabla g(\bfy ) . \mathbf{n} ~ \d S
        % - \int_{A} \nabla g(\bfy ) .\nabla  \frac{1}{\|\bfx -\bfy \|^{d-2}_2}  \d \bfy \\
        &= \int_{\partial B(\bfx , \varepsilon)}  \frac{1}{\|\bfx -\bfy \|^{d-2}_2}  \nabla g(\bfy ) . \mathbf{n} \,\d S
        + (2-d)\int_{A} \nabla g(\bfy ) .\frac{\bfx -\bfy }{\|\bfx -\bfy \|^{d}_2}  \d \bfy  \\
        &= \frac{1}{\varepsilon^{d-2}}\int_{\partial B(\bfx , \varepsilon)}    \nabla g(\bfy ) . \mathbf{n} \,\d S
        + (2-d)\int_{A} \nabla g(\bfy ) .\frac{\bfx -\bfy }{\|\bfx -\bfy \|^{d}_2}  \d \bfy  \\
        &= \varepsilon \int_{\partial B(\bfx , 1)}  \nabla g(\varepsilon \bfy ) . \mathbf{n} \,\d S
        + (2-d)\int_{A} \nabla g(\bfy ) .\frac{\bfx -\bfy }{\|\bfx -\bfy \|^{d}_2}  \d \bfy.
        % & \xrightarrow[\varepsilon \rightarrow 0]{} (2-d)\int_{\mathbb{R}^d} \nabla g(\bfy ) .\frac{\bfx -\bfy }{\|\bfx -\bfy \|^{d}_2}  \d \bfy
    \end{align*}
    Plugging into \eqref{eq:tool} evaluating the limit, we obtain the desired limit.
    % Using the first equality in Equation \eqref{eq:integral_with_force} we get
    % $$g(\bfx ) = \lim_{\varepsilon \rightarrow 0} \frac{1}{(2-d)d \kappa_d} \int_{B(\mathbf{0}, R) \setminus B(\bfx , \varepsilon)}  \frac{\Delta g(\bfy )}{\|\bfx -\bfy \|^{d-2}_2} \d \bfy  = \frac{1}{d \kappa_d}  \int_{\mathbb{R}^d} \nabla g(\bfy ) .\frac{\bfx -\bfy }{\|\bfx -\bfy \|^{d}_2}  \d \bfy .$$
\end{proof}

    Before working on bounding the variance of linear statistics under the repelled Poisson point process, we need an elementary bound on the fractional moments of the counts of a Poisson point process. 

\begin{lemma}
    \label{lem:moment_poisson}
    Consider a homogeneous Poisson point process $\mathcal{P} \subset \mathbb{R}^d$ of intensity $\rho > 1$ with $d \geq 1$. For any $R > 0$ and $q > 1$, we have
    \begin{equation}\label{eq:moment_poisson}
        \bbE \left[\left(1 + \sum_{x \in \mathcal{P}} \mathds{1}_{B(\mathbf{0}, R)}(\bfx)\right)^q \right]^{1/q} \leq C(q, R) \rho,
    \end{equation}
    where $C(q, R) < \infty$ is a constant depending only on $q$ and $R$.
\end{lemma}
\begin{proof}[Proof of Lemma \ref{lem:moment_poisson}]
    We first prove \eqref{eq:moment_poisson} when $q$ is an integer. 
    To begin with, recall that the moment of order $q$ of a Poisson distribution $X$ of parameter $\lambda > 0$ is given by  
    $$
    \bbE[X^q] = \sum_{n = 0}^q \stir{q}{n} \lambda^n,
    $$ 
    where $\stir{q}{n}$ denote the Stirling numbers of the second kind. 
    To wit, $\stir{q}{n}$ counts the number of ways to partition a set of $q$ elements into $n$ non-empty subsets \citep[Equation 6]{Philipson1963}.
    For $\rho > 1$, using the binomial expansion and the above observation, we get
    \begin{equation}
        \label{eq:bound_poisson_for_int}
        \bbE\left[\left(1 + \sum_{x \in \mathcal{P}} \mathds{1}_{B(\mathbf{0}, R)}(\bfx)\right)^q \right] 
        \leq 
        2^q \left(1 + \sum_{n = 0}^q \stir{q}{n} \left(\left|B(\mathbf{0}, R)\right| \rho\right)^n\right) 
        \leq 
        \underbrace{2^q \left(1 + \sum_{n = 0}^q \stir{q}{n} \left|B(\mathbf{0}, R)\right|^n\right)}_{C(q, R)^q}
        \rho^q,
    \end{equation}
    where we used that any binomial coefficient $(_n^q) \leq 2^q$.
    Thus, Equation \eqref{eq:bound_poisson_for_int} proves Lemma \ref{eq:bound_poisson_for_int} when $q$ is an integer. 
    
    Finally, to deduce the general case for non-integer $q > 1$, we use Jensen's inequality to go back to the integer case and apply~\eqref{eq:bound_poisson_for_int},
    \begin{equation*}
        \bbE \left[\left(1 + \sum_{x \in \mathcal{P}} \mathds{1}_{B(\mathbf{0}, R)}(\bfx)\right)^q \right] 
        \leq 
        \bbE \left[\left(1 + \sum_{x \in \mathcal{P}} \mathds{1}_{B(\mathbf{0}, R)}(\bfx)\right)^{\left \lceil{q}\right \rceil} \right]^{q/\left \lceil{q}\right \rceil} \leq C(\left \lceil{q}\right \rceil, R)^{q/\left \lceil{q}\right \rceil}\rho^q.
    \end{equation*}
    
\end{proof}

Lemmas~\ref{lem:integral_with_the_harmonic_fct} and \ref{lem:moment_poisson} will be instrumental in proving the subsequent lemmas, which bring us closer to establishing control over the variance of linear statistics under the repelled Poisson point process.
As we will see, the proof of Theorem~\ref{thm:Variance_reduction} hinges on performing a Taylor expansion with an integral remainder for \( f(\bfx + \varepsilon F_{\mathcal{P}}(\bfx)) \) with respect to \( \varepsilon \). 
Consequently, the upcoming lemmas analyze each term in this expansion and aim to derive corresponding bounds for them.

\begin{lemma}
    \label{lem:variance_from_internal_term_A}
    Consider a homogeneous Poisson point process $\mathcal{P} \subset \mathbb{R}^d$ of intensity $\rho > 1$, with $d \geq 3$.
    Let $f \in C^2(\mathbb{R}^d)$ of compact support $K \subset B(\mathbf{0}, R)$ with $R>0$.
    For any $R^{\prime}\geq R$, and $0 < \delta < d/(d-1)$, we have
    \begin{equation}\label{eq:first_term_target_integral_pop_out}
        \bbE \left[\left(\sum_{\bfx  \in \mathcal{P}\cap B(\mathbf{0}, R^{\prime})} f\left(\bfx\right) \mathds{1}_{B(\mathbf{0}, R)}(\bfx + \varepsilon F_{\mathcal{P} }(\bfx )) \right)^2\right]
        -
        \bbE \left[\left(\sum_{\bfx  \in \mathcal{P}} f(\bfx)  \right)^2\right]
        =   O \left(\rho^2 |\varepsilon|^{\delta}\right).
    \end{equation}
\end{lemma}

\begin{proof}[Proof of Lemma \ref{lem:variance_from_internal_term_A}]
    Let $R'\geq R$, $\varepsilon \in (-1,1)$ and $0 < \delta < d/(d-1)$. Since, $f$ has support $K\subset B(0, R)$, proving equation \eqref{eq:first_term_target_integral_pop_out} is equivalent to proving 
    $$
    A  \triangleq  \bbE \left[\left(\sum_{\bfx  \in \mathcal{P}\cap B(\mathbf{0}, R)} f\left(\bfx\right) \mathds{1}_{B(\mathbf{0}, R)}(\bfx + \varepsilon F_{\mathcal{P} }(\bfx )) \right)^2 - \left(\sum_{\bfx  \in \mathcal{P} \cap B(\mathbf{0}, R)} f(\bfx)  \right)^2\right] = 
    O\left(\rho^2 \vert\varepsilon\vert^{\delta}\right).
    $$
    %The idea of the proof is to make appear $\mathds{1}_{B(\mathbf{0}, R)^c}(\bfx + \varepsilon F_{\mathcal{P} }(\bfx ))$ and then to use the decay of the tail of $F_{\mathcal{P} }(\bfx )$. 
    To begin with, observe that
    %we use that  $|a^2 - b^2| = |(a-b)(a+b)| \leq |a -b|\left(|a| + |b|\right)$, where $a, b$ are real, to get
    \begin{align}
        \label{eq:interm_variance_from_internal_term_A}
        \left|A\right| 
        & \leq 
        \bbE\left[ \sum_{\bfx  \in \mathcal{P}\cap B(\mathbf{0}, R)} \left|f(\bfx)\right|\mathds{1}_{B(\mathbf{0}, R)^c}(\bfx + \varepsilon F_{\mathcal{P} }(\bfx )) \left(\left|\sum_{\bfx  \in \mathcal{P}\cap B(\mathbf{0}, R)} f\left(\bfx\right) \mathds{1}_{B(\mathbf{0}, R)}(\bfx + \varepsilon F_{\mathcal{P} }(\bfx ))\right| + \left|\sum_{\bfx  \in \mathcal{P}\cap B(\mathbf{0}, R)} f\left(\bfx\right)\right|\right) \right] \nonumber
        \\
        & \leq  
        2 \|f\|_{\infty}^2 \bbE\left[ \sum_{\bfx  \in \mathcal{P}\cap B(\mathbf{0}, R)} \mathds{1}_{B(\mathbf{0}, R)^c}(\bfx + \varepsilon F_{\mathcal{P} }(\bfx )) \sum_{\bfx  \in \mathcal{P}\cap B(\mathbf{0}, R)} 1 \right] 
        \\
        & = 
        2 \|f\|_{\infty}^2 
        \bbE\left[ 
        \sum_{\bfx  \in \mathcal{P} \cap B(\mathbf{0}, R)}\mathds{1}_{B(\mathbf{0}, R)^c}(\bfx + \varepsilon F_{\mathcal{P}\setminus\{\bfx\}}(\bfx )) 
        \left(1 + \sum_{\bfy  \in \mathcal{P}\setminus\{\bfx\} } \mathds{1}_{B(\mathbf{0}, R)}(y) \right)\right]. 
        \nonumber
    \end{align}
    For the second step, let $p$ and $\gamma$ such that $\delta < \delta p < \gamma < d/(d-1)$, and $q = p/(p-1)$ the conjugate exponent of $p$. 
    Note that these two choices are possible since $\delta < d/(d-1)$.
    Using the extended Slivnyak-Mecke theorem \eqref{eq:slivnyak-mecke}, then H\"older's inequality and finally Lemma \ref{lem:moment_poisson}, we obtain
    \begin{align*}
        |A| \leq
        & 2 \|f\|_{\infty}^2\rho \int_{B(\mathbf{0}, R)} \bbE\left[\mathds{1}_{B(\mathbf{0}, R)^c}(\bfx + \varepsilon F_{\mathcal{P}}(\bfx ))\left(1+ \sum_{\bfy  \in \mathcal{P} } \mathds{1}_{B(\mathbf{0}, R)}(y)\right)\right] \d\bfx
        \\
        &\leq 
        2 \|f\|_{\infty}^2\rho \int_{B(\mathbf{0}, R)} \bbP\left[\|\bfx + \varepsilon F_{\mathcal{P}}(\bfx )\|_2 \geq R\right]^{1/p} \bbE\left[\left(1+ \sum_{\bfy  \in \mathcal{P} } \mathds{1}_{B(\mathbf{0}, R)}(y)\right)^q\right]^{1/q} \d\bfx
        \\ 
        &\leq
        2 \|f\|_{\infty}^2 C(q, R) \rho^2 \int_{B(\mathbf{0}, R)} \bbP\left[\|F_{\mathcal{P}}(\bfx )\|_2 \geq \frac{R - \|\bfx\|_2}{\varepsilon}\right]^{1/p} \d\bfx.
    \end{align*}
    Now, using Markov's inequality we get
    \begin{align}\label{eq:in_proofs_div_Cmu}
        |A| \leq 
        2 \|f\|_{\infty}^2 C(q, R)^{1/q}  
        \bbE\left[\|F_{\mathcal{P}}({\bf0} )\|_2^{\gamma}\right]^{1/p}
        \rho^2 |\varepsilon|^{\gamma/p} 
        \int_{B(\mathbf{0}, R)} \left(R - \|\bfx\|_2\right)^{- \gamma/p} \d \bfx . 
    \end{align}
    As $\gamma < d/(d-1)$, we have $\bbE\left[\|F_{\mathcal{P}}({\bf0} )\|_2^{\gamma}\right] < \infty$ (Proposition \ref{prop:distribution_of_F}). Moreover, $\gamma/p < \gamma < d/(d-1) < d$, so the integral is finite. Finally, $\gamma/p > \delta$, hence $|\varepsilon|^{\gamma/p} \leq |\varepsilon|^{\delta}$, which concludes the proof. 
\end{proof}

Note that in the above proof the only property of the force $F_{\mathcal{P}}$ that has been used is that it admits moment of order strictly smaller than $d/(d-1)$. 
Since the truncated version of the force $F_{\mathcal{P}}^{(0, p)}$, with  $p > 0$, also satisfies this property, Lemma \ref{lem:variance_from_internal_term} remains valid when replacing $F_{\mathcal{P}}$ with it truncated counterpart $F_{\mathcal{P}}^{(0, p)}$.	

The next Lemma takes us closer to proving the variance reduction theorem, making the term $-2d \kappa_d \rho \varepsilon \Var\left[\widehat{I}_{ \calP}(f)\right]$ observed in Theorem \ref{thm:Variance_reduction} appear.

\begin{lemma}
    \label{lem:variance_from_internal_term}
    Consider a homogeneous Poisson point process $\mathcal{P} \subset \mathbb{R}^d$ of intensity $\rho > 1$, with $d \geq 3$.
    Let $f \in C^2(\mathbb{R}^d)$ of compact support $K \subset B(\mathbf{0}, R)$ with $R>0$.
    For any $R^{\prime}\geq R$ and $0 < \delta < d/(d-1)$, we have
    \begin{equation}
        \label{eq:the_target_integral_pops_out}
        %\lim_{\varepsilon \to 0} {\varepsilon}^{-1} \bbE \left[\left(\sum_{\bfx  \in \mathcal{P}\cap B(\mathbf{0}, R^{\prime})} f\left(\bfx + \varepsilon F_{\mathcal{P} }(\bfx )\right) \right)^2 -
        %\left(\sum_{\bfx  \in \mathcal{P}} f(\bfx )\right)^2\right]
        %=
        %-2 d \kappa_d \rho^2 I_K(f^2),
        \varepsilon 
        \bbE 
        \left[
        \sum_{\bfx  \in \mathcal{P}\cap B(\mathbf{0}, R^{\prime})} 
        f\left(\bfx\right) \mathds{1}_{B(\mathbf{0}, R)}(\tilde{\bfx}) 
        \sum_{\bfy  \in \mathcal{P}\cap B(\mathbf{0}, R^{\prime})} 
        \int_0^1 \nabla f\left(\tilde{\bfy}\right) 
        \cdot 
        F_{\mathcal{P}}(\bfy) \mathds{1}_{B(\mathbf{0}, R)}(\tilde{\bfy}) 
        \d \tau 
        \right] 
        = 
        -d \kappa_d \rho^2 \varepsilon I_K(f^2) 
        +
        O(\rho^2 |\varepsilon|^{\delta}),
    \end{equation}
    where $\tilde{\bfx} \triangleq \bfx + \varepsilon F_{\mathcal{P} }(\bfx )$, $\tilde{\bfy} \triangleq \bfy + \tau \varepsilon F_{\mathcal{P} }(\bfy )$, and $I_K$ is defined in \eqref{eq:target_integral}.
\end{lemma}

\begin{proof}[Proof of Lemma \ref{lem:variance_from_internal_term}]

    Let $R'\geq R$, $\varepsilon \in (-1,1)$ and $0 < \delta < d/(d-1)$.
    We first show that
    \begin{equation}
        \label{eq:D}
        D \triangleq \varepsilon \bbE \left[
        \sum_{\bfx  \in \mathcal{P}\cap B(\mathbf{0}, R)} f(\bfx )
        \sum_{\bfy \in \mathcal{P}\cap B(\mathbf{0}, R)} \nabla f(\bfy) \cdot F_{\mathcal{P} }(\bfy)\right] = - 2d \kappa_d \rho^2 I_K(f^2),
    \end{equation}
    so proving \eqref{eq:the_target_integral_pops_out} is equivalent to proving that 
    \begin{equation}
        \label{eq:the_target_integral_pops_out_bis}
        A - D = O(\rho^2 |\varepsilon|^{\delta}),
    \end{equation}  
    where $A$ denotes the left-hand side of equation \eqref{eq:the_target_integral_pops_out}.
To see why \eqref{eq:D} holds, expand the sum in $D$ to obtain
\begin{align}
    %\bbE \left[
    %\sum_{\bfx  \in \mathcal{P}\cap B(\mathbf{0}, R)} f(\bfx )
    %\sum_{\bfy \in \mathcal{P}\cap B(\mathbf{0}, R)} \nabla f(\bfy) %\cdot F_{\mathcal{P} }(\bfy)\right] 
    D &=
    \bbE \left[
    \sum_{\bfx  \in \mathcal{P}\cap B(\mathbf{0}, R)} f(\bfx )
    \nabla f(\bfx ) \cdot F_{\mathcal{P} }(\bfx )
    +
    \sum_{\bfx , \bfy \in \mathcal{P} \cap B(\mathbf{0}, R)}^{ \neq }f(\bfx ) \nabla f(\bfy) \cdot F_{\mathcal{P} }(\bfy)\right]
    \nonumber
    \\
    &=
    \bbE
    \left[\sum_{\bfx  \in \mathcal{P}\cap B(\mathbf{0}, R) }f(\bfx ) \nabla f(\bfx ) \cdot F_{\mathcal{P} \setminus\{\bfx \}}(\bfx )\right]
    -
    \bbE
    \left[
    \sum_{\bfx , \bfy \in \mathcal{P}\cap B(\mathbf{0}, R)}^{ \neq }f(\bfx ) \nabla f(\bfy) \cdot \frac{\bfx  - \bfy}{\|\bfx  - \bfy\|_2^d}
    \right]
    \\
    &
    +
    \bbE \left[\sum_{\bfx , \bfy \in \mathcal{P}\cap B(\mathbf{0}, R) }^{\neq }f(\bfx ) \nabla f(\bfy) \cdot F_{\mathcal{P} \setminus\{\bfy, \bfx \}}(\bfy)
    \right].
    \label{eq:eq_valid_after_truncation}
\end{align}
Using Slivnyak-Mecke \eqref{eq:slivnyak-mecke} and Equation \eqref{eq:def_moment_poisson}, we obtain
\begin{align}
    \label{eq:last_step}
    D &=
    \int_{ B(\mathbf{0}, R)}
    \bbE
    \left[
    f(\bfx ) \nabla f(\bfx ) \cdot F_{\mathcal{P}}(\bfx )
    \right]
    \rho \d \bfx
    -
    \int_{B(\mathbf{0}, R) \times B(\mathbf{0}, R)}
    f(\bfx ) \nabla f(\bfy) \cdot \frac{\bfx  - \bfy}{\|\bfx
        -
        \bfy\|_2^d}\rho^2\d \bfx  \d \bfy
    +
    \\
    &
    \int_{B(\mathbf{0}, R) \times B(\mathbf{0}, R)}
    \bbE
    \left[
    f(\bfx ) \nabla f(\bfy) \cdot F_{\mathcal{P}}(\bfy)
    \right]
    \rho^2
    \d \bfx  \d \bfy
    .
\end{align}
As the first and last terms in \eqref{eq:last_step} are zero due to $F_\calP$ being a centered process, and the second term is $- 2d \kappa_d \rho^2 I_K(f^2)$ by Lemma~\ref{lem:integral_with_the_harmonic_fct}, we obtain Equation \eqref{eq:D}.
Now we can focus on proving Equation \eqref{eq:the_target_integral_pops_out_bis} to show the lemma. 

    We write $A-D$ as the sum of three terms
    \begin{equation}
        \label{e:three_terms}
        A-D = (A-B) + (B-C) + (C-D),
    \end{equation}
    where, 
    \begin{equation*}
        B  \triangleq  \varepsilon \bbE \left[\sum_{\bfx  \in \mathcal{P}\cap B(\mathbf{0}, R^{\prime})} f\left(\bfx\right) \mathds{1}_{B(\mathbf{0}, R)}(\bfx + \varepsilon F_{\mathcal{P} }(\bfx )) \sum_{\bfy  \in \mathcal{P}\cap B(\mathbf{0}, R^{\prime})} \nabla f\left(\bfy\right) \cdot F_{\mathcal{P} }(\bfy ) \mathds{1}_{B(\mathbf{0}, R)}(\bfy + \varepsilon F_{\mathcal{P} }(\bfy )) \right],
    \end{equation*}
    and 
    \begin{equation*}
        C \triangleq  \varepsilon \bbE \left[\sum_{\bfx  \in \mathcal{P}\cap B(\mathbf{0}, R^{\prime})} f\left(\bfx\right) \sum_{\bfy  \in \mathcal{P}\cap B(\mathbf{0}, R^{\prime})} \nabla f\left(\bfy\right) \cdot F_{\mathcal{P} }(\bfy ) \mathds{1}_{B(\mathbf{0}, R)}(\bfy + \varepsilon F_{\mathcal{P} }(\bfy )) \right],
    \end{equation*}
    and we now show that each of these three terms is $O(\rho^2 |\varepsilon|^\delta)$.
    
    For bounding the first term in \eqref{e:three_terms}, we take $0 < \gamma < 1$ such that $\delta -1 < \gamma < d/(d-1)-1$.
    Note that such a $\gamma$ exists since $\delta - 1 < d/(d-1) -1 = 1/(d-1) < 1$. 
    Using Cauchy-Schwarz and the mean value theorem, we have 
    \begin{align*}
        \left|A - B\right| 
        &
        \leq 
        |\varepsilon| 
        \|f\|_{\infty} 
        \bbE
        \left[
        \sum_{\bfx  \in \mathcal{P}\cap B(\mathbf{0}, R^{\prime})}
        \sum_{\bfy  \in \mathcal{P}\cap B(\mathbf{0}, R^{\prime})} 
        \int_0^1 \|\nabla f\left(\bfy + \tau \varepsilon F_{\mathcal{P} }(\bfy )\right) - \nabla f\left(\bfy\right)\|_2 \|F_{\mathcal{P}}(\bfy)\|_2 \d\tau 
        \right]
        \\
        &\leq
        |\varepsilon| 
        \|f\|_{\infty} 
        c_d 
        \bbE
        \left[
        \sum_{\bfx  \in \mathcal{P}\cap B(\mathbf{0}, R^{\prime})}
        \sum_{\bfy  \in \mathcal{P}\cap B(\mathbf{0}, R^{\prime})} 
        \int_0^1 (2 \| \nabla f\|_{\infty})^{1 - \gamma} \| \nabla^2 f \|_{\infty}^{\gamma} \left(\tau |\varepsilon| \|F_{\mathcal{P} }(\bfy )\|_{2} \right)^{\gamma } \|F_{\mathcal{P}}(\bfy)\|_2 \d\tau
        \right]
        \\
        & 
        \leq 
        2^{1 - \gamma} 
        c_d 
        \|f\|_{\infty}  
        \| \nabla f\|_{\infty}^{1 - \gamma} 
        \| \nabla^2 f \|_{\infty}^{\gamma} 
        |\varepsilon|^{1+\gamma} 
        \bbE
        \left[
        \sum_{\bfx  \in \mathcal{P}\cap B(\mathbf{0}, R^{\prime})}
        \sum_{\bfy  \in \mathcal{P}\cap B(\mathbf{0}, R^{\prime})} \|F_{\mathcal{P} }(\bfy )\|_{2}^{\gamma +1}
        \right],
    \end{align*}
    where $c_d > 0$ is a constant that depends only on the dimension $d$, $\| \nabla f\|_{\infty} \triangleq \| \|\nabla f\|_2 \|_{\infty}$ and  $\| \nabla^2 f \|_{\infty} \triangleq  \| \|\nabla^2 f\|_2 \|_{\infty}$. 
    Note that the second inequality in the above computation can be derived as follows
    \begin{align}
        \label{eq:trik_bounding_force_moment}
        \|\nabla f\left(\bfy + \tau \varepsilon F_{\mathcal{P} }(\bfy )\right) - \nabla f\left(\bfy\right)\|_2 
        & =  
        \|\nabla f\left(\bfy + \tau \varepsilon F_{\mathcal{P} }(\bfy )\right) - \nabla f\left(\bfy\right)\|_2^{ 1 - \gamma}  
        \|\nabla f\left(\bfy + \tau \varepsilon F_{\mathcal{P} }(\bfy )\right) - \nabla f\left(\bfy\right)\|_2^{\gamma}
        \\
        & 
        \leq
        (2 \| \nabla f\|_{\infty})^{1 - \gamma} 
        \| \nabla^2 f \|_{\infty}^{\gamma} 
        \left(\tau |\varepsilon| \|F_{\mathcal{P} }(\bfy )\|_{2} \right)^{\gamma }
        \label{e:intermediate_step_AB}
    \end{align}
    where the last inequality follows from the mean value theorem. 
    The introduction of the parameter $\gamma$ is critical here to avoid terms involving higher moments of $F$, with order larger than $d/(d-1)$, as such moments are not well-defined. 
    This technique will also be used in subsequent proofs.
    
    Now, to bound the expectation appearing in \eqref{e:intermediate_step_AB}, take $p > 1$ small enough so that $p(1+\gamma) < d/(d-1)$ and let $q = p/(p-1)$. 
    Using Slivnyak-Mecke \eqref{eq:slivnyak-mecke} again, H\"older's inequality and Lemma \ref{lem:moment_poisson}, we obtain
    \begin{align}
        \label{eq:sum_sum_F_power_gammaplus1}
        \bbE
        \left[
        \sum_{\bfx  \in \mathcal{P}\cap B(\mathbf{0}, R^{\prime})}
        \sum_{\bfy  \in \mathcal{P}\cap B(\mathbf{0}, R^{\prime})} \|F_{\mathcal{P} }(\bfy )\|_{2}^{\gamma +1}\right] 
        & = 
        \rho 
        \int_{B(\mathbf{0}, R^{\prime})} 
        \bbE
        \left[
        \|F_{\mathcal{P} }(\bfy )\|_{2}^{\gamma +1} 
        \left(1 + \sum_{\bfx  \in \mathcal{P}} \mathds{1}_{B(\mathbf{0}, R')}(\bfx)\right)
        \right] \d\bfy 
        \nonumber
        \\	
        & \leq 	
        \rho  
        \int_{B(\mathbf{0}, R^{\prime})} 
        \bbE 
        \left[
        \|F_{\mathcal{P} }(\bfy )\|_{2}^{p(\gamma +1)}
        \right]^{1/p} 
        \bbE 
        \left[
        \left(1 + \sum_{\bfx  \in \mathcal{P}}  
        \mathds{1}_{B(\mathbf{0}, R')}(\bfx)\right)^q
        \right]^{1/q} 
        \d\bfy 
        \nonumber
        \\
        & \leq 
        \leb{B(\mathbf{0}, R^{\prime})} 
        C(q, R^{\prime}) 
        \bbE 
        \left[
        \|F_{\mathcal{P} }(\mathbf{0})\|_{2}^{p(\gamma +1)}
        \right]^{1/p} 
        \rho^2
        = O(\rho^2).
    \end{align}
    Note that $\bbE \left[\|F_{\mathcal{P} }(\mathbf{0})\|_{2}^{p(\gamma +1)}\right] < \infty$, by Proposition \ref{prop:distribution_of_F} and since $p(1+\gamma) < d/(d-1)$. Recalling that $1 + \gamma > \delta$ and $\varepsilon \in (-1, 1)$, we finally obtain
    \begin{equation}\label{eq:AB}
        |A - B| = O(\rho^2 |\varepsilon|^{\delta}).
    \end{equation}
    
    For the second term $B-C$ in \eqref{e:three_terms}, using the triangular inequality, we obtain a bound similar to \eqref{eq:interm_variance_from_internal_term_A},
    \begin{align*}
        \left|B - C\right| & \leq
        |\varepsilon| \|f\|_{\infty} \|\nabla f\|_{\infty} \bbE \left[\sum_{\bfx  \in \mathcal{P} \cap B(\mathbf{0}, R)} \mathds{1}_{B(\mathbf{0}, R)^c}(\bfx + \varepsilon F_{\mathcal{P} }(\bfx )) \sum_{\bfy  \in \mathcal{P}\cap B(\mathbf{0}, R)} \|F_{\mathcal{P} }(\bfy )\|_2 \mathds{1}_{B(\mathbf{0}, R)}(\bfy + \varepsilon F_{\mathcal{P} }(\bfy )) \right] 
        \\
        &\leq
        |\varepsilon| \|f\|_{\infty} \|\nabla f\|_{\infty} \bbE \left[\sum_{\bfx  \in \mathcal{P}\cap B(\mathbf{0}, R)}\mathds{1}_{B(\mathbf{0}, R)^c}(\bfx + \varepsilon F_{\mathcal{P} }(\bfx )) \sum_{\bfy  \in \mathcal{P}\cap B(\mathbf{0}, R)} \left(\|y\|_2 + R\right) \varepsilon^{-1} \right] 
        \\
        &\leq
        2R \|f\|_{\infty} \|\nabla f\|_{\infty} \bbE \left[\sum_{\bfx  \in \mathcal{P}\cap B(\mathbf{0}, R)} \mathds{1}_{B(\mathbf{0}, R)^c}(\bfx + \varepsilon F_{\mathcal{P} }(\bfx )) \sum_{\bfy  \in \mathcal{P}\cap B(\mathbf{0}, R)} 1 \right].
    \end{align*}
    Following the same methodology used in Lemma \ref{lem:variance_from_internal_term_A} to bound equation \eqref{eq:interm_variance_from_internal_term_A}, we get
    \begin{equation}
        \label{eq:BC}
        |B-C| = O(\rho^2 |\varepsilon|^{\delta}).
    \end{equation}
    
    For the last term $C- D$ in \eqref{e:three_terms}, we first use Cauchy-Schwarz inequality and the Slivnyak-Mecke theorem~\eqref{eq:slivnyak-mecke} to get
    \begin{align*}
        \left|C - D\right| 
        & \leq 
        |\varepsilon| 
        \|f\|_{\infty} \|
        \nabla f \|_{\infty} 
        \bbE 
        \left[
        \sum_{\bfx  \in \mathcal{P}\cap B(\mathbf{0}, R)} 1 
        \sum_{\bfy  \in \mathcal{P}\cap B(\mathbf{0}, R)} \|F_{\mathcal{P} }(\bfy )\|_2 \mathds{1}_{B(\mathbf{0}, R)^c}(\bfy + \varepsilon F_{\mathcal{P} }(\bfy )) 
        \right] 
        \\ 
        & =	
        |\varepsilon| 
        \|f\|_{\infty} \|
        \nabla f \|_{\infty}  
        \rho 
        \int_{B(\mathbf{0}, R)} 
        \bbE 
        \left[
        \|F_{\mathcal{P} }(\bfy )\|_2 \mathds{1}_{B(\mathbf{0}, R)^c}(\bfy + \varepsilon F_{\mathcal{P} }(\bfy )) 
        \left(1 + \sum_{\bfx  \in \mathcal{P}\cap B(\mathbf{0}, R)} 1 \right)
        \right] 
        \d \bfy.
    \end{align*}
    Now, let $\delta^{\prime} > 1$ be such that $\delta < \delta^{\prime} < d/(d-1)$, $p$ such that $1 < p < d/((d-1)\delta')$, and $q = p/(p-1)$. 
    Using H\"older's inequality and Lemma \ref{lem:moment_poisson}, we get
    \begin{equation*}
        \left|C - D\right| 
        \leq 
        \|f\|_{\infty} \|
        \nabla f \|_{\infty} 
        C(q, R) 
        \rho^2 |\varepsilon| 
        \int_{B(\mathbf{0}, R)} 
        \bbE 
        \left[
        \|F_{\mathcal{P} }(\bfy )\|_2^p 
        \mathds{1}_{B(\mathbf{0}, R)^c}(\bfy + \varepsilon F_{\mathcal{P} }(\bfy ))
        \right]^{1/p} 
        \d\bfy.
    \end{equation*}
    Now, to bound the integral in the last inequality, using again H\"older's inequality with $a > 1$ and $b = a/(a-1)$ such that $pa < d/(d-1)$, followed by Markov's inequality with $1 < \gamma < d/(d-1)$, we get 
    \begin{align*}
        \left|C - D\right| 
        & \leq 
        \|f\|_{\infty} \|
        \nabla f \|_{\infty} 
        C(q, R') 
        \bbE 
        \left[
        \|F_{\mathcal{P} }(\mathbf{0} )\|_2^{ap}\right]^{1/(ap)} 
        \rho^2 |\varepsilon| 
        \int_{B(\mathbf{0}, R)} \bbP \left[\|\bfy + \varepsilon F_{\mathcal{P} }(\bfy )\|_2 \leq R\right]^{1/(bp)} 
        \d\bfy
        \\ 
        & \leq 
        \|f\|_{\infty} \|
        \nabla f \|_{\infty} 
        C(q, R') 
        \bbE 
        \left[
        \|F_{\mathcal{P} }(\mathbf{0} )\|_2^{ap}
        \right]^{1/(ap)} 
        \bbE 
        \left[\|F_{\mathcal{P} }(\mathbf{0} )\|_2^{\gamma}
        \right]^{1/(bp)} 
        \rho^2 
        |\varepsilon|^{1+\gamma/(bp)} 
        \int_{B(\mathbf{0}, R)} 
        \frac1{(R - \|x\|_2)^{\gamma/(bp)}} 
        \d\bfy.
    \end{align*}
    Note that the moments of the force $F_{\mathcal{P} }(\mathbf{0} )$ in the last inequality are well defined since $ap$ and $\gamma$ are strictly small than $d/(d-1)$, see Proposition \ref{prop:distribution_of_F}), while the integral is finite because $\gamma/(bp) < \gamma < d/(d-1) < d$.
    Finally, one can choose $a$ and $\gamma$ in order to have $1+\gamma/(bp) > \delta$ and we get
    \begin{equation}
        \label{eq:CD}
        |C - D| = O(\rho^2 |\varepsilon|^{\delta}).
    \end{equation}
    Indeed, take $a = (1 - \mu)d/(p(d-1))$ and $\gamma = (1 - \mu)d/(d-1)$ with $\mu < 1$ small enough in order to ensure $a > 1$ (using $p < d/(d-1)$) and $\gamma > 1$ (using $d/(d-1) > 1$). 
    With this choice, we get
    \begin{align*}
        1+\frac{\gamma}{bp} 
        = 
        1+ \gamma \frac{a-1}{a p} 
        = 
        1 + \frac{d(1-\mu)}{p(d-1)}
        \left(1 - \frac{p(d-1)}{d} \frac1{1- \mu}\right) 
        = 
        \frac{d}{p(d-1)}(1 - \mu) 
        > \delta'(1 - \mu) > \delta,
    \end{align*}
    for $\mu$ small enough since $\delta' > \delta$. 
    
    Gathering equations \eqref{eq:AB}, \eqref{eq:BC} and \eqref{eq:CD}, we obtain \eqref{eq:the_target_integral_pops_out_bis} which ends the proof.
\end{proof}

The key arguments of the above proof are Equation \eqref{eq:last_step} coupled with Lemma \ref{lem:integral_with_the_harmonic_fct}.
Remarkably, the proof of Lemma~\ref{lem:variance_from_internal_term} remains valid even if $F_{\mathcal{P}}$ is replaced by its truncated version $F_{\mathcal{P}}^{(0,p)}$, where $p \geq 2 R$. In particular, the choice of $p \geq 2 R$ is crucial to ensure that Equation \eqref{eq:eq_valid_after_truncation} remains valid. 

The next lemma will be useful in quantifying the integral remainder term of the Taylor expansion that shall appear in the proof of Theorem~\ref{thm:Variance_reduction}.

\begin{lemma}
    \label{lem:variance_integral_remainder}
    Consider a homogeneous Poisson point process $\mathcal{P} \subset \mathbb{R}^d$ of intensity $\rho > 1$, with $d \geq 3$.
    Let $f \in C^2(\mathbb{R}^d)$ of compact support $K \subset B(\mathbf{0}, R)$ with $R>0$.
    For any $R^{\prime}\geq R$ and $0 < \delta < d/(d-1)$,
    \begin{equation}
        \label{eq:variance_integral_remainder}
        \varepsilon^2 
        \bbE\left[\left(\sum_{\bfy  \in \mathcal{P}\cap B(\mathbf{0}, R^{\prime})} \int_0^1 \nabla f\left(\bfy + \tau \varepsilon F_{\mathcal{P} }(\bfy )\right)\cdot F_{\mathcal{P}}(\bfy) \mathds{1}_{B(\mathbf{0}, R)}(\bfy + \varepsilon F_{\mathcal{P} }(\bfy )) \d \tau\right)^2 \right] 
        = O(\rho^2 |\varepsilon|^{\delta}).
    \end{equation}
\end{lemma}

\begin{proof}[Proof of Lemma \ref{lem:variance_integral_remainder}]
    Let $R^{\prime}\geq R$, $\varepsilon\in(-1,1)$, and $0 < \delta < d/(d-1)$.
    We denote by $A$ the left-hand side of \eqref{eq:variance_integral_remainder}. Using three times  Cauchy-Schwarz, we obtain
    \begin{align*}
        A &\leq 
        \varepsilon^2 \bbE\left[\sum_{\bfx  \in \mathcal{P}\cap B(\mathbf{0}, R^{\prime})} 1 \sum_{\bfy  \in \mathcal{P}\cap B(\mathbf{0}, R^{\prime})} \left(\int_0^1 \nabla f\left(\bfy + \tau \varepsilon F_{\mathcal{P} }(\bfy )\right)\cdot F_{\mathcal{P}}(\bfy) \mathds{1}_{B(\mathbf{0}, R)}(\bfy + \varepsilon F_{\mathcal{P} }(\bfy )) 
        \d \tau\right)^2 \right] 
        \\& \leq
        \varepsilon^2 \bbE\left[\sum_{\bfx  \in \mathcal{P}\cap B(\mathbf{0}, R^{\prime})} 1 \sum_{\bfy  \in \mathcal{P}\cap B(\mathbf{0}, R^{\prime})} \int_0^1 \left|\nabla f\left(\bfy + \tau \varepsilon F_{\mathcal{P} }(\bfy )\right)\cdot F_{\mathcal{P}}(\bfy) \mathds{1}_{B(\mathbf{0}, R)}(\bfy + \varepsilon F_{\mathcal{P} }(\bfy ))\right|^2 \d\tau \right] 
        \\& \leq 
        \|\nabla f\|_{\infty}^2 \varepsilon^2 \bbE\left[\sum_{\bfx  \in \mathcal{P}\cap B(\mathbf{0}, R^{\prime})} 1 \sum_{\bfy  \in \mathcal{P}\cap B(\mathbf{0}, R^{\prime})}  \|F_{\mathcal{P}}(\bfy)\|_2^2 \mathds{1}_{B(\mathbf{0}, R)}(\bfy + \varepsilon F_{\mathcal{P} }(\bfy )) \right].
    \end{align*}
    Let $\gamma^{\prime}> 0$ be such that $2 - d/(d-1) < \gamma^{\prime} < 2 - \delta$. 
    Using the same trick as in Equation \eqref{eq:trik_bounding_force_moment}, we obtain
    \begin{equation*}
        A \leq \|\nabla f\|_{\infty}^2 (R+R^{\prime}) |\varepsilon|^{2 - \gamma^{\prime}} \bbE\left[\sum_{\bfx  \in \mathcal{P}\cap B(\mathbf{0}, R^{\prime})} 1 \sum_{\bfy  \in \mathcal{P}\cap B(\mathbf{0}, R^{\prime})} \|F_{\mathcal{P}}(\bfy)\|_2^{2 - \gamma^{\prime}} \right].
    \end{equation*}
    Finally, using \eqref{eq:sum_sum_F_power_gammaplus1} with $\gamma = 1 - \gamma^{\prime} \in (\delta-1,  d/(d-1)-1)$, as required, we get $A = O(\rho^2 |\varepsilon|^{2 - \gamma^{\prime}}) = O(\rho^2 |\varepsilon|^{\delta})$, since $2 - \gamma^{\prime} > \delta$.
\end{proof}

In the next two lemmas, we work on bounding the contribution to the variance from the points of the Poisson process located outside the support of integrand $f$.
\begin{lemma}
\label{lem:zero_variance_of_external_term}
Consider a homogeneous Poisson point process $\mathcal{P} \subset \mathbb{R}^d$ of intensity $\rho > 1$, with $d \geq 3$.
Let $f\in C^2(\mathbb{R}^d)$ of compact support $K \subset B(\mathbf{0}, R )$ with $R>0$.
Let further $0<\beta < \frac{1}{2(d-1)^2}$ and $ 0 < \delta < d/(d-1)$.
For any $R^{\prime} > (2R + 2)^{1/\beta}$
\begin{equation*}
    %\label{eq:eq_2_proof_prop_3}
    %\lim_{\varepsilon \to 0} {\varepsilon}^{-1}\bbE \left[\left(\sum_{\bfx  \in \mathcal{P}\cap B(\mathbf{0}, R^\prime)^c} f(\bfx + \varepsilon F_{\mathcal{P}  }(\bfx ))\right)^2 \right] = 0.
    \bbE \left[\left(\sum_{\bfx  \in \mathcal{P}\cap B(\mathbf{0}, R^\prime)^c} f(\bfx + \varepsilon F_{\mathcal{P}  }(\bfx ))\right)^2 \right] 
    = 
    O \left( \rho^2 |\varepsilon|^{\delta}\right).
\end{equation*}
\end{lemma}
The proof is based on Corollary \ref{coro:existence_external_moments} and \ref{coro:existence_internal_moments}.
\begin{proof}[Proof of Lemma \ref{lem:zero_variance_of_external_term}]
Let $R^{\prime} > (2R + 2)^{1/\beta}$, $\varepsilon \in (-1,1)$, and
$$
Y \triangleq \bbE \left[\left(\sum_{\bfx  \in \mathcal{P}\cap B(\mathbf{0}, R^\prime)^c} f(\bfx + \varepsilon F_{\mathcal{P} }(\bfx ))\right)^2 \right].
$$
Setting $r(\bfx ) = \|\bfx \|_2^\beta$ with $0<\beta < \frac{1}{2(d-1)^2}$, we use the same splitting technique as in \eqref{eq:inclusion_event_a_point_enter}, and write
\begin{align*}
    |Y| &\leq
    \|f\|_\infty\bbE \left[ \left(\sum_{\bfx  \in \mathcal{P}\cap B(\mathbf{0}, R^{\prime})^c} \mathds{1}_{B(\mathbf{0}, R)}(\bfx + \varepsilon F_{\mathcal{P} }(\bfx ))\right)^2 \right]
    \\
    & \leq
    \|f\|_\infty
    \bbE \left[ \left(\sum_{\bfx  \in \mathcal{P}\cap B(\mathbf{0}, R^{\prime})^c} \mathds{1}_{B(\mathbf{0}, r(\bfx ))}(\bfx + \varepsilon F^{(0,1)}_{\mathcal{P} }(\bfx ))
    +
    \mathds{1}_{\{ \|\varepsilon F^{(1,\infty)}_{\mathcal{P} }(\bfx )\|_2 > r(\bfx ) - R\}}(\bfx)\right)^2
    \right]
    \\
    & \leq
    2\|f\|_\infty
    \underbrace{\bbE \left[ \left(\sum_{\bfx  \in \mathcal{P}\cap B(\mathbf{0}, R^{\prime})^c} \mathds{1}_{B(\mathbf{0}, r(\bfx ))}(\bfx + \varepsilon F^{(0,1)}_{\mathcal{P} }(\bfx ))\right)^2 \right]}_{\triangleq E_{\text{int}}(\varepsilon)}
    +
    \\
    &
    2\|f\|_\infty \underbrace{\bbE \left[ \left(\sum_{\bfx  \in \mathcal{P}\cap B(\mathbf{0}, R^{\prime})^c}
        \mathds{1}_{\{ \|\varepsilon F^{(1,\infty)}_{\mathcal{P} }(\bfx )\|_2 > r(\bfx ) - R\}}(\bfx)\right)^2
        \right]}_{\triangleq E_{\text{ext}}(\varepsilon)} .
\end{align*}
By Corollary \ref{coro:existence_internal_moments} with $m=2$, there exists $a_1 >0$ such that for any $\gamma > 0$ satisfying $\delta-1 < \gamma < 1/(d-1)$,
\begin{align}
    \label{eq:E_int}
    E_{\text{int}}(\varepsilon) & \leq a_1 E_1 + E_1^2 \leq
    a_1 |\varepsilon|^{1 + \gamma} + 2 |\varepsilon|^{2 + 2\gamma} + h_1(\varepsilon) + 2h_1(\varepsilon)^2,
\end{align}
where $E_1$ is defined by Equation \eqref{eq:def_E_k}, for $k=1$ and
\begin{align*}
    h_1(\varepsilon)&= \rho\int_{B(\mathbf{0}, R^{\prime})^c}
    a_2\exp \left(
    - a_3 \frac{\|\bfx \|_2 - r(\bfx )}{|\varepsilon|}\log\left(
    a_4 \frac{\|\bfx \|_2 - r(\bfx )}{|\varepsilon|}\right)
    \right)
    ~\d \bfx\\
    & = \rho\int_{B(\mathbf{0}, R^{\prime})^c}
    a_2\exp \left(
    - a_3 \frac{\|\bfx \|_2^{\beta}}{|\varepsilon|}\left(\|\bfx \|_2^{1-\beta} - 1\right)\log\left(
    a_4 \frac{\|\bfx \|_2^\beta}{|\varepsilon|} \left(\|\bfx \|_2^{1-\beta} - 1\right)\right)
    \right)
    ~\d \bfx,
\end{align*}
for some positive constants $a_2, \dots, a_4$.
But
\begin{align*}
    h_1(\varepsilon)
    &\leq
    \rho \int_{ B(\mathbf{0}, R^\prime)^c} a_2\exp \left(
    - a_3 \frac{\|\mathbf{x}\|_2^\beta}{ |\varepsilon|} \left({R^\prime}^{1-\beta} - 1\right)
    \log \left(
    a_4 \frac{\|\mathbf{x}\|_2^\beta}{| \varepsilon|}\left({R^\prime}^{1-\beta} - 1\right)\right)
    \right) ~\d \mathbf{x}\\
    &=
    \rho |\varepsilon|^{d/\beta}
    \int_{ B(\mathbf{0}, R^\prime |\varepsilon|^{-1/\beta})^c} a_2\exp \left(
    - a_3 \|\bfy\|_2^\beta \left({R^\prime}^{1-\beta} - 1\right)
    \log \left(
    a_4 \|\bfy\|_2^\beta \left({R^\prime}^{1-\beta} - 1\right)
    \right)
    \right) ~\d \bfy\\
    &\leq \rho |\varepsilon|^{d/\beta}
    \underbrace{\int_{ B(\mathbf{0}, 1)^c} a_2\exp \left(
        - a_3 \|\bfy\|_2^\beta \left({R^\prime}^{1-\beta} - 1\right)
        \log \left(
        a_4 \|\bfy\|_2^\beta \left({R^\prime}^{1-\beta} - 1\right)
        \right)
        \right) ~\d \bfy}_{\triangleq c}.
\end{align*}
where in the second line we used the change of variable $\bfy = | \varepsilon |^{-1/\beta} \bfx$, and in the last line we used that $R^\prime> |\varepsilon|^{1/\beta}$.
Thus
\begin{align*}
    0 \leq h_1(\varepsilon)
    &\leq c \rho |\varepsilon|^{d/\beta}.
\end{align*}
As $d/\beta> 2(d-1)^2 d > \delta$ and $\rho > 1$, we have $h_1(\varepsilon) = O(\rho |\varepsilon|^{\delta}) = O(\rho^2 |\varepsilon|^{\delta})$ and thus $h_1(\varepsilon)^2 = O(\rho^2 |\varepsilon|^{2\delta}) =  O(\rho^2 |\varepsilon|^{\delta})$. 
Subsequently, 
$$E_{\text{int}}(\varepsilon) =  O(|\varepsilon|^{1 + \gamma} + |\varepsilon|^{2(1 + \gamma)}) + O(\rho^2 |\varepsilon|^{\delta})=  O(\rho^2 |\varepsilon|^{\delta}),$$ 
since $2(1 + \gamma) > 1+\gamma >\delta$.

It remains to show that $E_{\text{ext}}(\varepsilon) = O(\rho^2 \varepsilon^{\delta})$.
By Corollary \ref{coro:existence_external_moments} with $m=2$ there exists $a_5, \dots, a_8>0$ such that
\begin{align}
    \label{eq:E_ext}
    E_{\text{ext}}(\varepsilon) & \leq a_5\left(h_2(\varepsilon) + h_2(\varepsilon)^2\right),
\end{align}
with
\begin{align*}
    h_2(\varepsilon)
    &= \rho \int_{B(\mathbf{0}, R^{\prime})^c}
    a_6 \exp \left(
    - a_7 \frac{ r(\bfx )- (R+1)}{|\varepsilon|}
    \log\left(
    a_8 \frac{r(\bfx )- (R+1)}{|\varepsilon|}\right)
    \right)
    ~\d \bfx.
\end{align*}
In particular
\begin{align*}
    h_2(\varepsilon)
    &=
    \rho \int_{B(\mathbf{0}, R^{\prime})^c}
    a_6 \exp \left(
    - a_7 \frac{ \|\bfx \|^\beta- (R+1)}{|\varepsilon|}\log \left(
    a_8 \frac{\|\bfx \|^\beta- (R+1)}{|\varepsilon|}\right)
    \right)
    ~\d \bfx
    \\
    &\leq \rho
    \int_{B(\mathbf{0}, R^{\prime})^c}
    a_6 \exp \left(
    - a_7 \frac{ \|\bfx \|^\beta}{2|\varepsilon|}\log\left(
    a_8 \frac{\|\bfx \|^\beta}{2|\varepsilon|}\right)
    \right)
    ~\d \bfx,
\end{align*}
where in the last line we used that $R^{\prime}> (2R + 2)^{1/\beta}$.
Following the same method used earlier to bound $h_1(\varepsilon)$, we can show that $h_2(\varepsilon) \leq C \rho |\varepsilon|^{d/\beta}$ for some constant $C$.
This implies that $h_2(\varepsilon) = O(\rho |\varepsilon|^{\delta})$ and 
$$E_{\text{ext}}(\varepsilon) = O(\rho^2 |\varepsilon|^{\delta}).$$
We conclude that $Y = O(\rho^2 |\varepsilon|^{\delta})$, which ends the proof.
\end{proof}

Again, the proof's validity is unaffected by replacing $F_{\mathcal{P}}$ with its truncated counterpart $F_{\mathcal{P}}^{(0,p)}$, where $p > 0$.

The upcoming lemma is the final tool required to demonstrate Theorem \ref{thm:Variance_reduction}.
\begin{lemma}
    \label{lem:zero_variance_of_product_term}
    Consider a homogeneous Poisson point process $\mathcal{P} \subset \mathbb{R}^d$ of intensity $\rho > 1$, with $d \geq 3$.
    Let $f$ be a $ C^2(\mathbb{R}^d)$ function of compact support $K\subset B(\mathbf{0}, R)$ with $R>0$ and $0 < \delta < d/(d-1)$.
    For any $R^{\prime} \geq R$ we have
    \begin{equation*}
        %\label{eq:eq_3_proof_prop_3}
        %\lim_{\varepsilon \to 0} { \varepsilon}^{-1} \bbE \left[ \sum_{\bfx  \in \mathcal{P}\cap B(\mathbf{0}, R^{\prime})^c} f\left(\bfx + \varepsilon F_{\mathcal{P}  }(\bfx )\right)
        %\sum_{\bfy  \in \mathcal{P}\cap B(\mathbf{0}, R^{\prime})} f\left(\bfy + \varepsilon F_{\mathcal{P} }(\bfy )\right) \right]=0.
        \bbE 
        \left[ 
        \sum_{\bfx  \in \mathcal{P}\cap B(\mathbf{0}, R^{\prime})^c} f\left(\bfx + \varepsilon F_{\mathcal{P}  }(\bfx )\right)
        \sum_{\bfy  \in \mathcal{P}\cap B(\mathbf{0}, R^{\prime})} f\left(\bfy + \varepsilon F_{\mathcal{P} }(\bfy )\right) 
        \right] 
        = 
        O\left(\rho^2 |\varepsilon|^{\delta}\right).
    \end{equation*}
\end{lemma}
\begin{proof}[Proof of Lemma \ref{lem:zero_variance_of_product_term}]
    Let $R^{\prime} \geq R$, $\varepsilon \in (-1,1)$, and denote
    $$Z \triangleq \bbE \left[\sum_{\bfx  \in \mathcal{P}\cap B(\mathbf{0}, R^{\prime})^c} f(\bfx + \varepsilon F_{\mathcal{P}  }(\bfx ))
    \sum_{\bfy  \in \mathcal{P}\cap B(\mathbf{0}, R^{\prime})} f(\bfy + \varepsilon F_{\mathcal{P}  }(\bfy )) \right]. $$
    We show that $Z = O(\rho^2 |\varepsilon|^{\delta})$.
    Since $f$ is $C^2$ with compact support, it is Lipschitz and we denote its Lipschitz constant by $L\geq 0$.
    For $  {\bf x}\in B^{c}(0,R')$, $f(\bfx) = 0$ so 
    $
    | f(\bfx + \varepsilon F_{\mathcal{P}  }(\bfx )) | = | f(\bfx + \varepsilon F_{\mathcal{P}  }(\bfx ))-f(\bfx) | \leq  L  \|\varepsilon F_{\mathcal{P}  }(\bfx )\|_2\mathds{1}_{B(\mathbf{0}, R^{\prime})}(\bfx + \varepsilon F_{\mathcal{P}  }(\bfx )).
    $
    Hence
    \begin{align*}
        \lvert Z \rvert
        & \leq
        \bbE \left[
        \sum_{\bfx  \in \mathcal{P}\cap B(\mathbf{0}, R^{\prime})^c}
        L \|\varepsilon F_{\mathcal{P}  }(\bfx )\|_2 \mathds{1}_{B(\mathbf{0}, R^{\prime})}(\bfx + \varepsilon F_{\mathcal{P}  }(\bfx ))\sum_{\bfy  \in \mathcal{P}\cap B(\mathbf{0}, R^{\prime}) } \|f\|_{\infty}
        \right]
        \\
        & \leq
        L\|f\|_{\infty}|\varepsilon|
        \bbE \left[
        \sum_{\bfx  \in \mathcal{P}\cap B(\mathbf{0}, R^{\prime})^c}
        \left( \|F_{\mathcal{P} \setminus \{\bfx\} }(\bfx )\|_2 \mathds{1}_{B(\mathbf{0}, R^{\prime})}(\bfx + \varepsilon F_{\mathcal{P} \setminus \{\bfx\} }(\bfx ))
        \sum_{\bfy \in \mathcal{P} \setminus \{\bfx \}} \mathds{1}_{B(\mathbf{0}, R^{\prime})}(\bfy)\right)
        \right]
        \\
        & =
        L\|f\|_{\infty} |\varepsilon|
        \int_{B(\mathbf{0}, R^{\prime})^c}
        \bbE \left[
        \|F_{\mathcal{P} }(\bfx )\|_2 \mathds{1}_{B(\mathbf{0}, R)}(\bfx + \varepsilon F_{\mathcal{P} }(\bfx ))
        \sum_{\bfy \in \mathcal{P}} \mathds{1}_{B(\mathbf{0}, R^{\prime})}(\bfy)
        \right]
        \rho \d \bfx.
    \end{align*}
    The last equality was obtained using the extended Slivnyak-Mecke theorem \eqref{eq:slivnyak-mecke}.
    Next, we employ Hölder's inequality with $ p> 1$ such that $p< \min(1+ d/(d-1) - \delta, d/(d-1))$ and $q=p/(p-1)$, and Lemma \ref{lem:moment_poisson}, to obtain
    \begin{align*}
        \lvert Z \rvert & \leq
        L\|f\|_{\infty} |\varepsilon|
        \int_{B(\mathbf{0}, R^\prime)^c}
        \bbE \left[
        \|F_{\mathcal{P} }(\bfx )\|_2^p \mathds{1}_{B(\mathbf{0}, R^\prime)}(\bfx +\varepsilon F_{\mathcal{P} }(\bfx ))\right]^{1/p}
        \bbE \left[\left(\sum_{\bfy \in \mathcal{P}} \mathds{1}_{B(\mathbf{0}, R^\prime)}(\bfy)\right)^q
        \right]^{1/q}
        \rho \d \bfx
        \\
        &=
        \underbrace{  L\|f\|_{\infty} C(q, R^{\prime}) }_{C_1} \rho^2 |\varepsilon|
        \int_{B(\mathbf{0}, R^\prime)^c} 
        \bbE \left[
        \|F_{\mathcal{P} }(\bfx )\|_2^p \mathds{1}_{B(\mathbf{0}, R^\prime)}(\bfx + \varepsilon F_{\mathcal{P} }(\bfx ))\right]^{1/p}
        \d \bfx.
    \end{align*}
    When attempting to bound the last integral, a challenge arises as $\bbE [\|F_{\mathcal{P}}(\mathbf{x})\|_2^p]^{1/p}$ is not integrable over $B(\mathbf{0}, R^\prime)^c$.
    To address this, we will handle the cases of $\|F_{\mathcal{P}}(\mathbf{x})\|_2 > 1$ and $\|F_{\mathcal{P}}(\mathbf{x})\|_2 < 1$ separately.
    By doing so, we can manage the computation differently for each case, allowing us to remove the exponent $\frac{1}{p}$ in the first case and proceed with the calculations accordingly.
    For $\mathbf{x} \in B(\mathbf{0}, R')^c$, we have
    \begin{align*}
        \bbE &\left[
        \|F_{\mathcal{P} }(\bfx )\|_2^p \mathds{1}_{B(\mathbf{0}, {R^\prime})}(\bfx + \varepsilon F_{\mathcal{P} }(\bfx ))\right]^{1/p}
        \\
        &=
        \bbE \left[
        \|F_{\mathcal{P} }(\bfx )\|_2^p \mathds{1}_{B(\mathbf{0}, {R^\prime})}(\bfx + \varepsilon F_{\mathcal{P} }(\bfx )) \left(\mathds{1}_{\{\|F_{\mathcal{P} }(\bfx )\|_2 <1\}}(\bfx ) + \mathds{1}_{\{\|F_{\mathcal{P} }(\bfx )\|_2 \geq 1\}}(\bfx ) \right)\right]^{1/p}
        \\
        & \leq
        \bbE \left[
        \|F_{\mathcal{P} }(\bfx )\|_2^p \mathds{1}_{B(\mathbf{0}, {R^\prime})}(\bfx + \varepsilon F_{\mathcal{P} }(\bfx ))
        \mathds{1}_{\{\|F_{\mathcal{P} }(\bfx )\|_2 <1\}}(\bfx )\right]^{1/p}
        +
        \bbE \left[
        \|F_{\mathcal{P} }(\bfx )\|_2^p \mathds{1}_{B(\mathbf{0}, {R^\prime})}(\bfx + \varepsilon F_{\mathcal{P} }(\bfx ))
        \mathds{1}_{\{\|F_{\mathcal{P} }(\bfx )\|_2 \geq1\}}(\bfx )\right]^{1/p}
        \\
        &\leq
        \bbE \left[ \mathds{1}_{B(\mathbf{0}, {R^\prime} + \varepsilon)}(\bfx )\right]^{1/p}
        +
        \bbE \left[
        \|F_{\mathcal{P} }(\bfx )\|_2^p \mathds{1}_{B(\mathbf{0}, {R^\prime})}(\bfx + \varepsilon F_{\mathcal{P} }(\bfx ))
        \right]
        \\
        &=
        \mathds{1}_{B(\mathbf{0}, {R^\prime} + \varepsilon)}(\bfx )
        +
        \bbE \left[
        \|F_{\mathcal{P} }(\bfx )\|_2^p \mathds{1}_{B(\mathbf{0}, {R^\prime})}(\bfx + \varepsilon F_{\mathcal{P} }(\bfx ))
        \right].
    \end{align*}
    Thus
    \begin{align}
        \label{eq:bound_Z}
        \lvert Z \rvert & \leq
        C_1 \rho^2 |\varepsilon| \int_{B(\mathbf{0}, {R^\prime})^c}
        \mathds{1}_{B(\mathbf{0}, {R^\prime} + \varepsilon)}(\bfx )
        +
        \bbE \left[
        \|F_{\mathcal{P} }(\bfx )\|_2^p \mathds{1}_{B(\mathbf{0}, {R^\prime})}(\bfx + \varepsilon F_{\mathcal{P} }(\bfx ))
        \right]
        \d \bfx
        \nonumber
        \\
        &=
        C_1  \rho^2 |\varepsilon|\left( \leb{B(\mathbf{0}, {R^\prime})^c \cap B(\mathbf{0}, {R^\prime} + \varepsilon)}
        +
        \int_{B(\mathbf{0}, {R^\prime})^c}
        \bbE \left[
        \|F_{\mathcal{P} }(\mathbf{0})\|_2^p \mathds{1}_{B(\mathbf{0}, {R^\prime})}(\bfx + \varepsilon F_{\mathcal{P} }(\mathbf{0}))
        \right]
        \d \bfx\right).
    \end{align}
    Now, to bound Equation \eqref{eq:bound_Z} first observe that 
    $$\leb{B(\mathbf{0}, {R^\prime})^c \cap B(\mathbf{0}, {R^\prime} + \varepsilon)} = O(|\varepsilon|).$$
    Indeed for $\varepsilon > 0$,
    \begin{align}
        \label{eq:intersection_ball}
        \leb{B(\mathbf{0}, {R^\prime})^c \cap B(\mathbf{0}, {R^\prime} + \varepsilon)} & = \leb{B(\mathbf{0}, {R^\prime} + \varepsilon)} - \leb{B(\mathbf{0}, {R^\prime})} = \leb{B(\mathbf{0}, 1)}\left((R^{\prime} + \varepsilon)^d - (R^{\prime})^d\right) \nonumber \\& = \leb{B(\mathbf{0}, 1)} \varepsilon \sum_{n = 0}^{d-1} (R^{\prime})^n \varepsilon^{d-1 -n} \leq  \leb{B(\mathbf{0}, 1)}\sum_{n = 0}^{d-1} (R^{\prime})^n \varepsilon = O(\varepsilon).
    \end{align}
    Arguing similarly for $\varepsilon < 0$, we get the same result $\leb{B(\mathbf{0}, {R^\prime})^c \cap B(\mathbf{0}, {R^\prime} + \varepsilon)} = O(|\varepsilon|)$. 
    It remains to bound the second term of Equation \eqref{eq:bound_Z}.
    
    Denote by $f_{F_{\mathcal{P} }(\mathbf{0})}$ the density function of $F_{\mathcal{P} }(\mathbf{0})$ and observe that
    \begin{align*}
        \int_{B(\mathbf{0}, {R^\prime})^c}
        \bbE \left[
        \|F_{\mathcal{P} }(\mathbf{0})\|_2^p \mathds{1}_{B(\mathbf{0}, {R^\prime})}(\bfx + \varepsilon F_{\mathcal{P} }(\mathbf{0}))
        \right]
        ~ \d \bfx &
        =
        \int_{B(\mathbf{0}, {R^\prime})^c}
        \int_{\mathbb{R}^d}
        \|\mathbf{u}\|_2^p \mathds{1}_{B(\mathbf{0}, {R^\prime})}(\bfx + \varepsilon \mathbf{u})
        f_{F_{\mathcal{P} }(\mathbf{0})}(\mathbf{u})
        ~ \d \mathbf{u} ~
        \d \bfx
        \\
        &=
        \int_{\mathbb{R}^d}
        \|\mathbf{u}\|_2^p
        f_{F_{\mathcal{P} }(\mathbf{0})}(\mathbf{u}) \leb{B(\mathbf{0}, {R^\prime})^c \cap B(\varepsilon \mathbf{u}, {R^\prime})}
        ~ \d \mathbf{u}.
    \end{align*}
    Let $\max(0, \delta -1) < \gamma < \min(d/(d-1) - p, 1)$, which is possible since $p < d/(d-1) - \delta +1$ and $\delta < 1+1/(d-1) < 2$. 
    We now reason similarly to Equation \eqref{eq:intersection_ball} while using $\gamma$ 
    with the same trick as in Equation \eqref{eq:trik_bounding_force_moment}
    \begin{align*}
        \leb{B(\mathbf{0}, {R^\prime})^c \cap B(\varepsilon \mathbf{u}, {R^\prime})} & 
        \leq \leb{B(\varepsilon \mathbf{u}, {R^\prime})}^{1 - \gamma} \leb{B(\mathbf{0}, {R^\prime})^c \cap B(\mathbf{0}, {R^\prime} + |\varepsilon| \|\mathbf{u}\|_2)}^{\gamma} 
        \\ 
        & \leq \leb{B(\mathbf{0}, {R^\prime})}^{1 - \gamma} \left(\leb{B(\mathbf{0}, 1)}\sum_{n = 1}^{d-1} (R^{\prime})^n\right)^{\gamma} |\varepsilon|^{\gamma} \|\mathbf{u}\|_2^{\gamma} = O(|\varepsilon|^{\gamma} \| \bfu \|_2 ^{\gamma} ).
    \end{align*}
    Plugging this back into Equation \eqref{eq:bound_Z}, we get
    \begin{equation*}
        \lvert Z\rvert 
        \leq 
        C_1\rho^2 
        \left(
        O\left(|\varepsilon|^2\right) 
        + 
        O\left(|\varepsilon|^{1+\gamma}\right)
        \bbE \left[\|F_{\mathcal{P} }(\mathbf{0})\|_2^{p+\gamma}\right]
        \right).
    \end{equation*}
    As $p+\gamma < d/(d-1)$, $\bbE\left[\|F_{\mathcal{P} }(\mathbf{0})\|_2^{p+\gamma}\right] < \infty$ by Proposition \ref{prop:distribution_of_F}. 
    Moreover, $2 > \delta$ and $1 + \gamma > \delta$. 
    This yields $Z = O(\rho^2 |\varepsilon|^{\delta})$, as announced. 
\end{proof}
Note that the proof's validity is unaffected by replacing $F_{\mathcal{P}}$ with its truncated counterpart $F_{\mathcal{P}}^{(0,p)}$, where $p > 0$.

We are finally ready to proceed with the proof of Theorem \ref{thm:Variance_reduction}. 
We will use  Lemmas~\ref{lem:variance_from_internal_term_A}, \ref{lem:variance_from_internal_term}, \ref{lem:variance_integral_remainder}, \ref{lem:zero_variance_of_external_term} and \ref{lem:zero_variance_of_product_term}.

\begin{proof}[Proof of Theorem \ref{thm:Variance_reduction}]
Consider a function $f \in C^2(\mathbb{R}^d)$ of compact support $K\subset B(0, R)$ for some $R>0$.
First, Proposition \ref{prop:extistance_of_the_moments} implies the existence of $\Var\Big[\widehat{I}_{ \Pi_{\varepsilon}\mathcal{P} \cap K }(f) \Big]$ for any $\varepsilon \in (-1,1)$.
Now, fix
$$
0<\beta < \frac{1}{2(d-1)^2},
$$
and let $R^{\prime} \geq \max((2R + 2)^{1/\beta}, R)$ and  $0 < \delta < d/(d-1)$.
As $\bbE [\widehat{I}_{ \Pi_{\varepsilon}\mathcal{P} \cap K }(f)] = \bbE [\widehat{I}_{ \mathcal{P} \cap K }(f)]$ we have
\begin{align*}
    &\Var \left[\widehat{I}_{ \Pi_{\varepsilon}\mathcal{P} \cap K }(f) \right]
    -
    \Var\left[\widehat{I}_{ \mathcal{P} \cap K }(f)\right]
    \\
    &=
    \bbE \left[\left(\widehat{I}_{ \Pi_{\varepsilon}\mathcal{P} \cap K }(f)\right)^2\right] - \bbE \left[\left(\widehat{I}_{ \mathcal{P} \cap K }(f)\right)^2\right]\\
    &=
    \rho^{-2}\bbE \left[\left(\sum_{\bfx  \in \mathcal{P}} f(\bfx + \varepsilon F_{\mathcal{P}  }(\bfx ))\right)^2
    -
    \left(\sum_{\bfx  \in \mathcal{P}} f(\bfx )\right)^2\right]
    \\
    &=\rho^{-2} \bbE \left[\left(\sum_{\bfx  \in \mathcal{P}\cap B(\mathbf{0}, R^\prime)} f(\bfx + \varepsilon F_{\mathcal{P} }(\bfx ))\right)^2 -
    \left(\sum_{\bfx  \in \mathcal{P}} f(\bfx )\right)^2\right]
    +
    \rho^{-2}\bbE \left[\left(\sum_{\bfx  \in \mathcal{P}\cap B(\mathbf{0}, R^\prime)^c} f(\bfx + \varepsilon F_{\mathcal{P}}(\bfx ))\right)^2\right]
    \\
    &\quad +
    2 \rho^{-2} \bbE \left[\sum_{\bfx  \in \mathcal{P}\cap B(\mathbf{0}, R^\prime)} f(\bfx + \varepsilon F_{\mathcal{P}  }(\bfx ))\sum_{\bfx  \in \mathcal{P}\cap B(\mathbf{0}, R^\prime)^c} f(\bfx + \varepsilon F_{\mathcal{P}  }(\bfx )) \right].
\end{align*}
Using Lemmas \ref{lem:zero_variance_of_external_term}, and \ref{lem:zero_variance_of_product_term} we get

    \begin{align}
        \label{eq:exp_var_1}
        \Var \left[\widehat{I}_{ \Pi_{\varepsilon}\mathcal{P} \cap K }(f) \right]
        -
        \Var\left[\widehat{I}_{ \mathcal{P} \cap K }(f)\right]
        \leq \rho^{-2} \bbE \left[\left(\sum_{\bfx  \in \mathcal{P}\cap B(\mathbf{0}, R^\prime)} f(\bfx + \varepsilon F_{\mathcal{P} }(\bfx ))\right)^2 -
        \left(\sum_{\bfx  \in \mathcal{P}} f(\bfx )\right)^2\right]
        + C^1_{\delta} |\varepsilon|^{\delta},
    \end{align}
    where $C^1_{\delta}$ does not depend on $\rho$.
To simplify the first term inside the expectation in the last equality, first observe that as $f$ has support included in $B(\mathbf{0}, R)$, we can rewrite
\begin{align*}
    T 
    \triangleq 
    \bbE 
    \left[
    \left(
    \sum_{\bfx  \in \mathcal{P}\cap B(\mathbf{0}, R^\prime)} 
    f(\bfx + \varepsilon F_{\mathcal{P} }(\bfx ))
    \right)^2
    \right] 
    = 
    \bbE 
    \left[
    \left(
    \sum_{\bfx  \in \mathcal{P}\cap B(\mathbf{0}, R^\prime)} 
    f(\bfx + \varepsilon F_{\mathcal{P} }(\bfx )) 
    \mathds{1}_{B(\mathbf{0}, R)}(\bfx + \varepsilon F_{\mathcal{P} }(\bfx ))\right)^2\right].
\end{align*}
Now, using a Taylor expansion with integral remainder, we get
\begin{align*}
    T 
    &= 
    \bbE 
    \left[\left(
    \sum_{\bfx  \in \mathcal{P}\cap B(\mathbf{0}, R^\prime)} 
    \left(f(\bfx) + \varepsilon \int_0^1 \nabla f(\bfx + \tau \varepsilon F_{\mathcal{P} }(\bfx )) \cdot F_{\mathcal{P}}(\bfx) 
    \d\tau\right)
    \mathds{1}_{B(\mathbf{0}, R)}(\bfx + \varepsilon F_{\mathcal{P} }(\bfx )) \right)^2\right] 
    \\
    & 
    = 
    \bbE \left[\left(\sum_{\bfx  \in \mathcal{P}\cap B(\mathbf{0}, R^\prime)} f(\bfx)\mathds{1}_{B(\mathbf{0}, R)}(\bfx + \varepsilon F_{\mathcal{P} }(\bfx ))\right)^2\right]
    \\
    & 
    \qquad 
    + 
    2\varepsilon 
    \bbE\left[\sum_{\bfx  \in \mathcal{P}\cap B(\mathbf{0}, R^{\prime})} 
    f\left(\bfx\right) \mathds{1}_{B(\mathbf{0}, R)}(\bfx + \varepsilon F_{\mathcal{P} }(\bfx )) 
    \sum_{\bfy  \in \mathcal{P}\cap B(\mathbf{0}, R^{\prime})} \int_0^1 \nabla f\left(\bfy + \tau \varepsilon F_{\mathcal{P} }(\bfy )\right) 
    \cdot 
    F_{\mathcal{P}}(\bfy) \mathds{1}_{B(\mathbf{0}, R)}(\bfy + \varepsilon F_{\mathcal{P} }(\bfy )) 
    \d\tau \right] 
    \\
    &
    \qquad 
    + 
    \varepsilon^2 
    \bbE
    \left[\left(
    \sum_{\bfy  \in \mathcal{P}\cap B(\mathbf{0}, R^{\prime})} 
    \int_0^1 
    \nabla f\left(\bfy + \tau \varepsilon F_{\mathcal{P} }(\bfy )\right)
    \cdot 
    F_{\mathcal{P}}(\bfy) \mathds{1}_{B(\mathbf{0}, R)}(\bfy + \varepsilon F_{\mathcal{P} }(\bfy )) 
    \d\tau
    \right)^2 \right].
\end{align*}
Using Lemmas \ref{lem:variance_from_internal_term_A}, \ref{lem:variance_from_internal_term} and \ref{lem:variance_integral_remainder}, we get

    \begin{align*}
        \label{eq:exp_var_2}
        T \leq
        \bbE \left[ \left(\sum_{\bfx  \in \mathcal{P}} f(\bfx )\right)^2\right] 
        - 2d\kappa_d \rho^2 \varepsilon I_K\left(f^2\right)
        + C^2_{\delta}\rho^2|\varepsilon|^{\delta},
\end{align*}
where $I_K(f^2)$ is the integral of $f^2$ as defined is Equation \eqref{eq:target_integral} and $C^2_{\delta}$ does not depend on $\rho$.
Finally, substituting the last result in Equation \eqref{eq:exp_var_1}, we get,for all $\delta \in (0, 1+1/(d-1))$,
    \begin{align*}
        \Var \left[\widehat{I}_{ \Pi_{\varepsilon}\mathcal{P} \cap K }(f) \right]
        -
        \Var\left[\widehat{I}_{ \mathcal{P} \cap K }(f)\right] 
        &
        \leq
        \rho^{-2} T 
        -  
        \rho^{-2} \bbE \left[ \left(\sum_{\bfx  \in \mathcal{P}} f(\bfx )\right)^2\right] 
        + 
        C^1_{\delta}|\varepsilon|^{\delta}
        \\
        & 
        \leq
        -2d \kappa_d  \varepsilon I_K\left(f^2\right) 
        + 
        (C^1_{\delta}+C^2_{\delta})|\varepsilon|^{\delta}
        \\
        & 
        =
        -2d \kappa_d  \varepsilon \rho 
        \Var\left[\widehat{I}_{ \mathcal{P} \cap K }(f)\right] 
        + 
        C_{\delta}|\varepsilon|^{\delta}.
    \end{align*}
    with $C_{\delta} = C^1_{\delta}+C^2_{\delta}$, which concludes the proof.
\end{proof}

\begin{remark}
\label{rmk:remark_proof_singularity_importance}
Substituting $F_{\calP}$ with its truncated version $F^{(0,p)}_{\calP}$, where $p \geq \text{diam}(K)$, preserves the validity of the proof of Theorem~\ref{thm:Variance_reduction}. However, the same cannot be said for $F^{(q,p)}_{\calP}$ when $q > 0$, as it leads to the breakdown of the proof of Lemma \ref{lem:variance_from_internal_term}.
This highlights the possibility that the variance reduction may be attributed to the singularity of $F_{\calP}$ at the points of $\calP$.
\end{remark}
%-------------------------------------
%-------------------------------------
\bmsection{Additional experiments}
\label{app:experiments}
\bmsubsection{Further analysis of Section~4.2} % (fold)
\label{app:analysis_of_the_experiment_4.2}
\begin{table}[h!]
    \centering
    \resizebox{0.95\textwidth}{!}{
        \begin{tabular}{|l|c|c|c|c|c|c|c|c|c|c|}
            \hline
            {} & SW & CI & SW & CI & SW & CI & SW & CI & SW & CI
            \\
            \hline
            $f_1$
            &stat=0.9, p=0.04 & [-0.58, -0.46]
            &stat=0.96, p=0.48 & [-0.69, -0.51]
            &stat=0.95, p=0.34 & [-0.66, -0.54]
            &stat=0.95, p=0.39 & [-0.65, -0.53]
            &stat=0.91, p=0.06 & [-0.52, -0.28]
            \\
            \hline
            $f_2$
            &stat=0.98, p=0.88 & [-0.63, -0.51]
            &stat=0.94, p=0.26 & [-0.65, -0.53]
            &stat=0.93, p=0.17 & [-0.61, -0.49]
            &stat=0.96, p=0.61 & [-0.63, -0.51]
            &stat=0.94, p=0.30 & [-0.60, -0.48]
            \\
            \hline
            $f_3$
            &stat=0.97, p=0.65 & [-0.57, -0.45]
            &stat=0.98, p=0.92 & [-0.67, -0.55]
            &stat=0.97, p=0.84 & [-0.69, -0.57]
            &stat=0.93, p=0.15 & [-0.64, -0.52]
            &stat=0.99, p=0.99 & [-0.60, -0.48]
            \\
            \hline
            & \multicolumn{2}{|c|}{d=2} & \multicolumn{2}{|c|}{d=3} & \multicolumn{2}{|c|}{d=4}
            & \multicolumn{2}{|c|}{d=5} & \multicolumn{2}{|c|}{d=7}
            \\
            \hline
    \end{tabular}}
    \caption{Shapiro Wilk (SW) test for the residual of the OLS of the estimated log-standard deviation of $\widehat{I}_{\text{MCRB}}$ over $\log(N)$ and the confidence interval (CI) with  $99.7\%$ confidence level of the slopes for $f_1, \, f_2$ and $ f_3$ when $d \in \{2,3, 4, 5, 7\}$.}
    \label{table:shap_wilk_and_conf_inf}
\end{table}
%-----------------
\begin{figure}[!h]
    \centering
    \begin{subfigure}{\linewidth}
        \centering
        \includegraphics[width=0.65\linewidth]{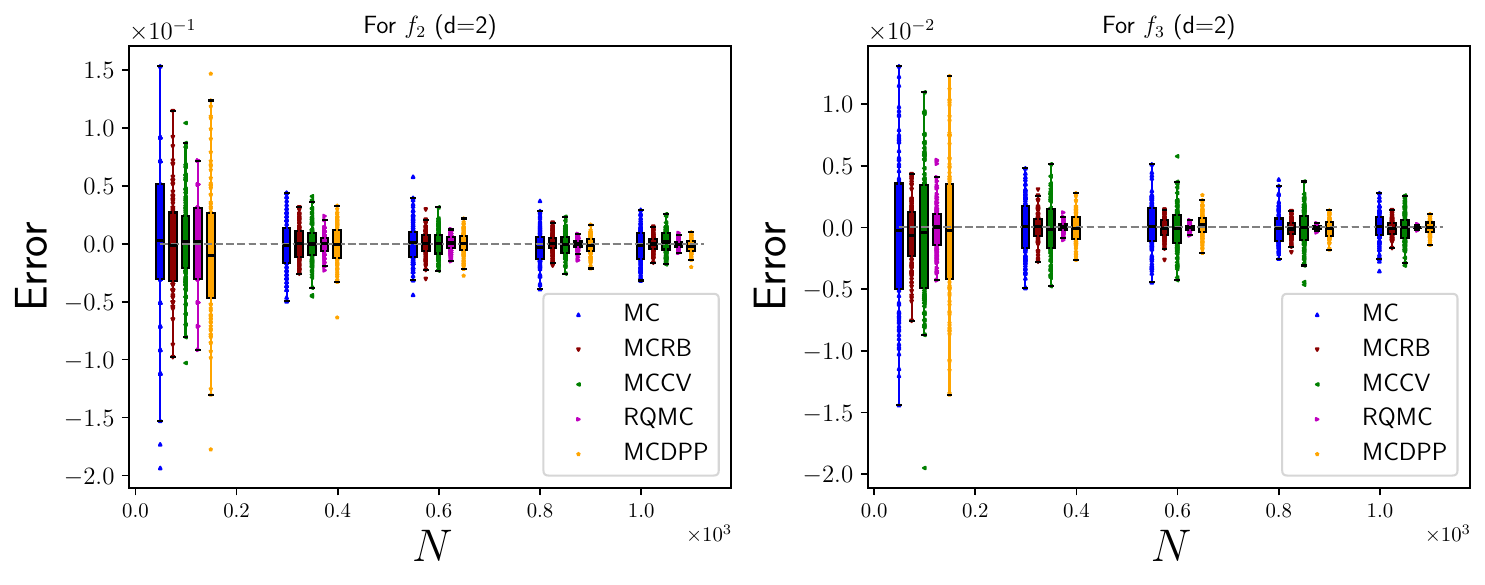}
    \end{subfigure}
    \begin{subfigure}{\linewidth}
        \centering
        \includegraphics[width=0.65\linewidth]{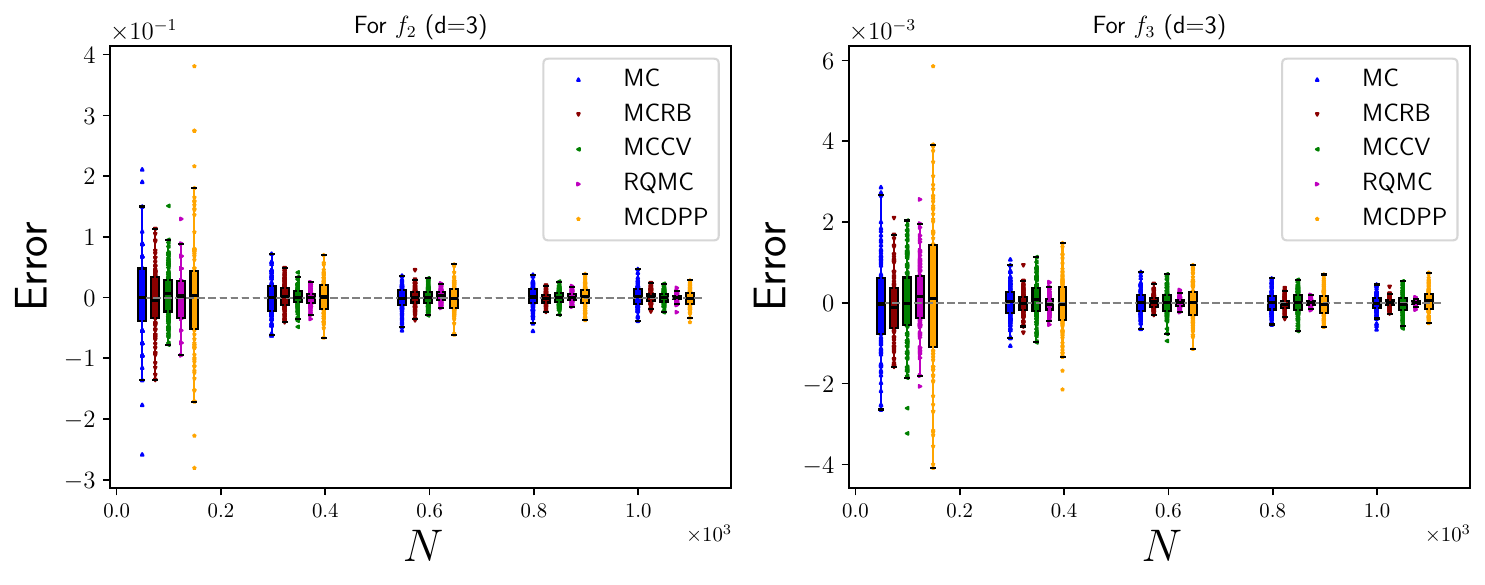}
    \end{subfigure}
    \begin{subfigure}{\linewidth}
        \centering
        \includegraphics[width=0.65\linewidth]{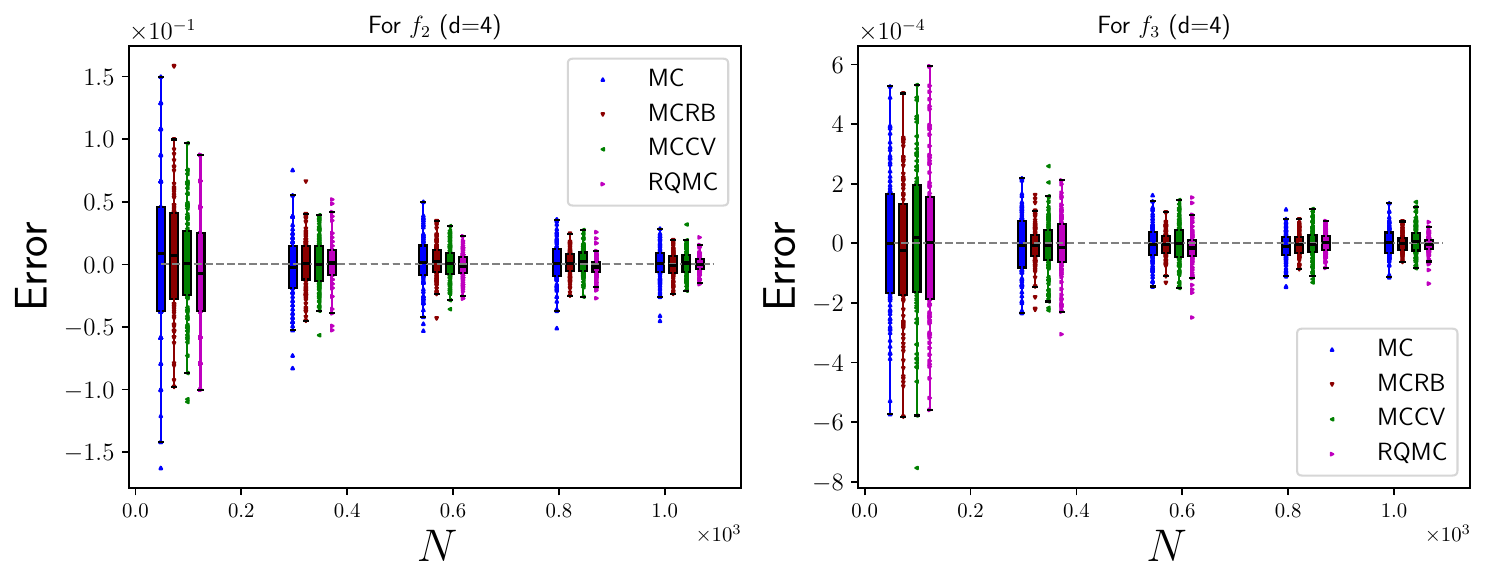}
    \end{subfigure}
    \begin{subfigure}{\linewidth}
        \centering
        \includegraphics[width=0.65\linewidth]{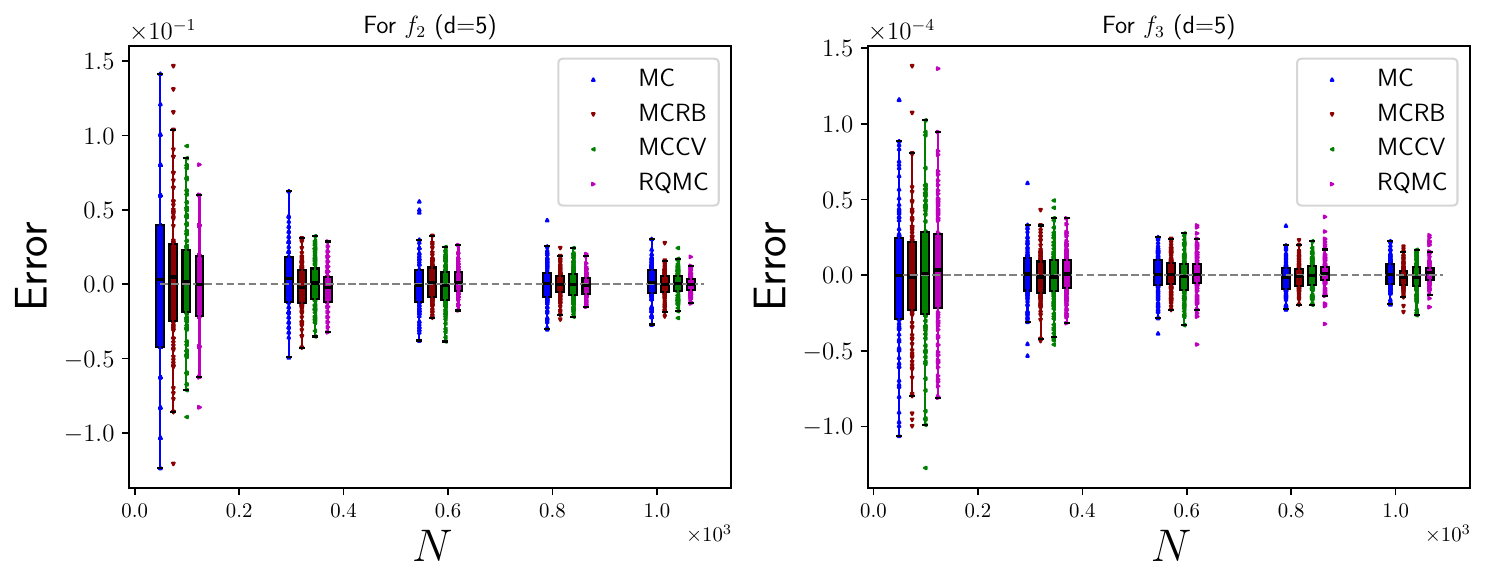}
    \end{subfigure}
    \begin{subfigure}{\linewidth}
        \centering
        \includegraphics[width=0.65\linewidth]{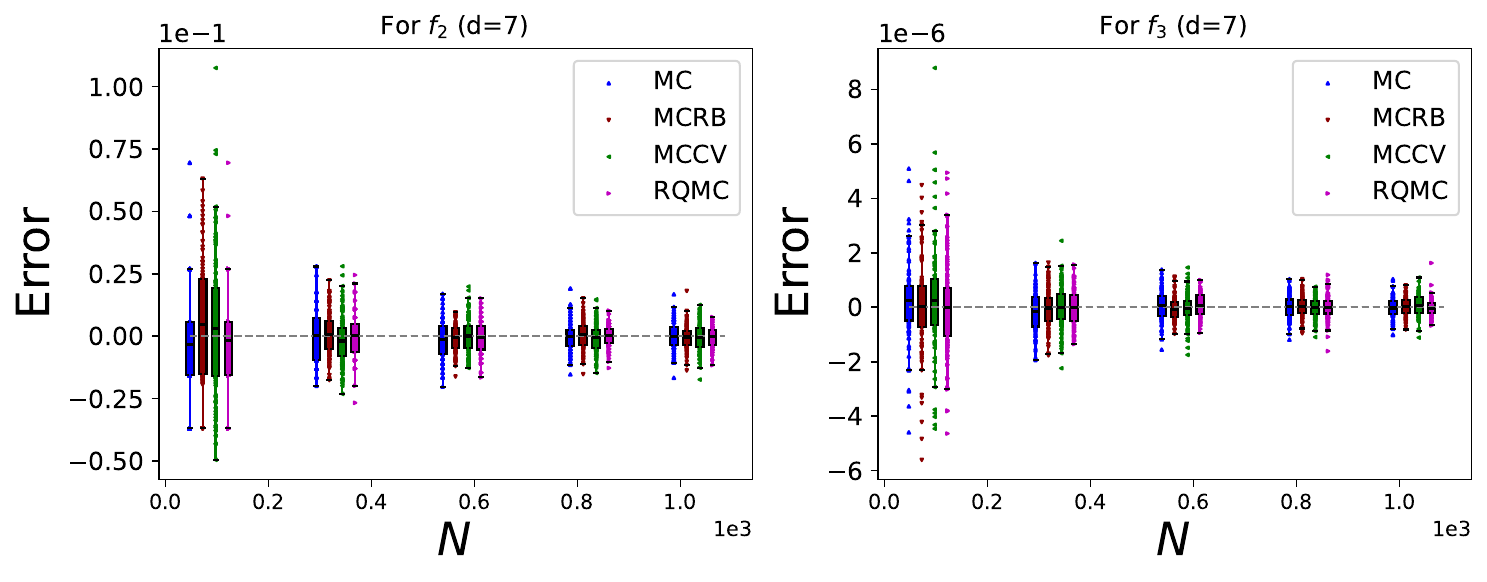}
    \end{subfigure}
    \caption{Estimated error of various Monte Carlo methods for $f_2$, and $f_3$ and $d \in \{2,3, 4, 5,7\}$.}
    \label{fig:square_error_mc}
\end{figure}
We used the Shapiro-Wilk test \citep{ShapWil1965} to assess whether the residuals of each OLS of $\log(\sigma)$ of $\widehat{I}_{\text{MCRB}}$ over $\log(N)$ shown in Figure 4 in the paper are Normally distributed, and computed confidence intervals for the slope corresponding to $\widehat{I}_{\text{MCRB}}$.
The $p$-values of the Shapiro-Wilk test corresponding to $\widehat{I}_{\text{MCRB}}$ are summarized in Table \ref{table:shap_wilk_and_conf_inf}.
The null hypothesis of the test is that the residuals are normally distributed.
As the statistics are close to 1 and the $p$-values are large (typically larger than 0.01), the distributions of the residuals are not significantly different from a normal distribution at a $99\%$ significance level.
We have also verified that the Quantile-Quantile plots of the residuals are compatible with the result of Shapiro-Wilk's test; for a comparison between Normality tests see \citep{NorBee2011}.
Hence, we can construct confidence intervals for the slopes using their estimated standard deviations.
Table \ref{table:shap_wilk_and_conf_inf} shows the $99.7\%$ confidence intervals (corresponding to three standard deviations) of the slopes.
For $d \in \{3, 4, 5\}$ the confidence intervals suggest that the variance of $\widehat{I}_{\text{MCRB}}$ may decrease slightly faster than the usual Monte Carlo rate.

Finally, to account for the slight possible bias of the estimator $\widehat{I}_{\text{MCRB}}$, we conclude this section by presenting the errors obtained in the experiment illustrated in Figure 4 in the paper.
Figure~\ref{fig:square_error_mc} displays the box plots of the error of the estimators $\widehat{I}_{\text{MC}}$, $\widehat{I}_{\text{MCRB}}$, $\widehat{I}_{\text{MCCV}}$, and $\widehat{I}_{\text{RQMC}}$, for the functions $f_2$ and $f_3$ for $d$ in $\{2, 3, 4, 5, 7\}$ and for $\widehat{I}_{\text{MCDPP}}$ when $d$ in $\{2, 3\}$.
There is no clear evidence in Figure~\ref{fig:square_error_mc} to suggest that $\widehat{I}_{\text{MCRB}}$ exhibits any notable bias.
%--------------------------------------
\bmsubsection{Further analysis of Section~5.1} % (fold)
\label{app:analysis_of_extreme_repelled}
Figure \ref{fig:extreme_repelled_additional} presents an extension of the iterative repulsion experiments discussed in Section 5.1. In this extended analysis, we explore a broader range of stopping times, specifically considering $t \in \{30, 80, 130, 180, 200\}$.
To set up the experiment, we consider a sample from a PPP $\calP$ of unit intensity.
Subsequently, we perform an iterative application of the repulsion operator $\Pi_{\varepsilon}\calP$ to the original point process $\calP$, yielding the resulting point processes $\Pi_{\varepsilon, t}\calP$ for $t \in \{30, 80, 130, 180, 200\}$.
In Figure \ref{fig:extreme_repelled_additional}, we present the outcomes of this iterative repulsion process. The left panels of the figure showcase the $t$-th iterates $\Pi_{\varepsilon, t}\calP$, specifically constrained to the observation window, where we set $\varepsilon = -\varepsilon_0$ (Equation (12) in the paper). On the other hand, the right panels display the results obtained with $\varepsilon = \varepsilon_0$.
\begin{figure}[!h]
    \centering
    \begin{subfigure}{\linewidth}
        \centering
        \includegraphics[width=0.75\linewidth]{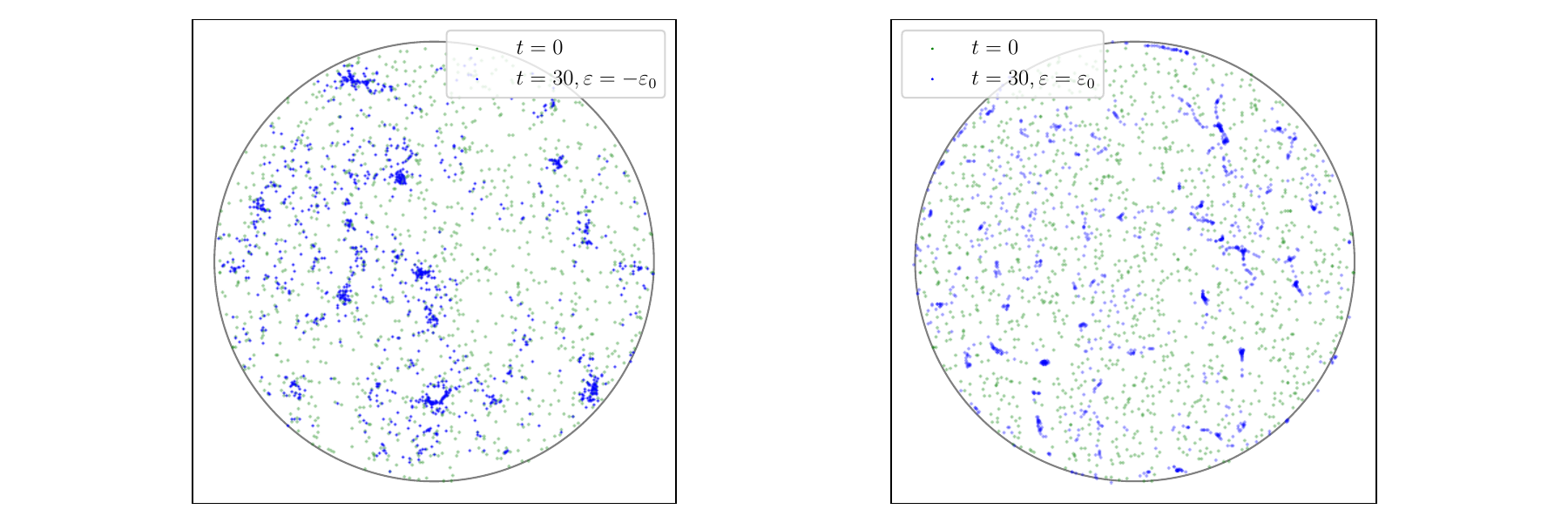}
    \end{subfigure}
    \begin{subfigure}{\linewidth}
        \centering
        \includegraphics[width=0.75\linewidth]{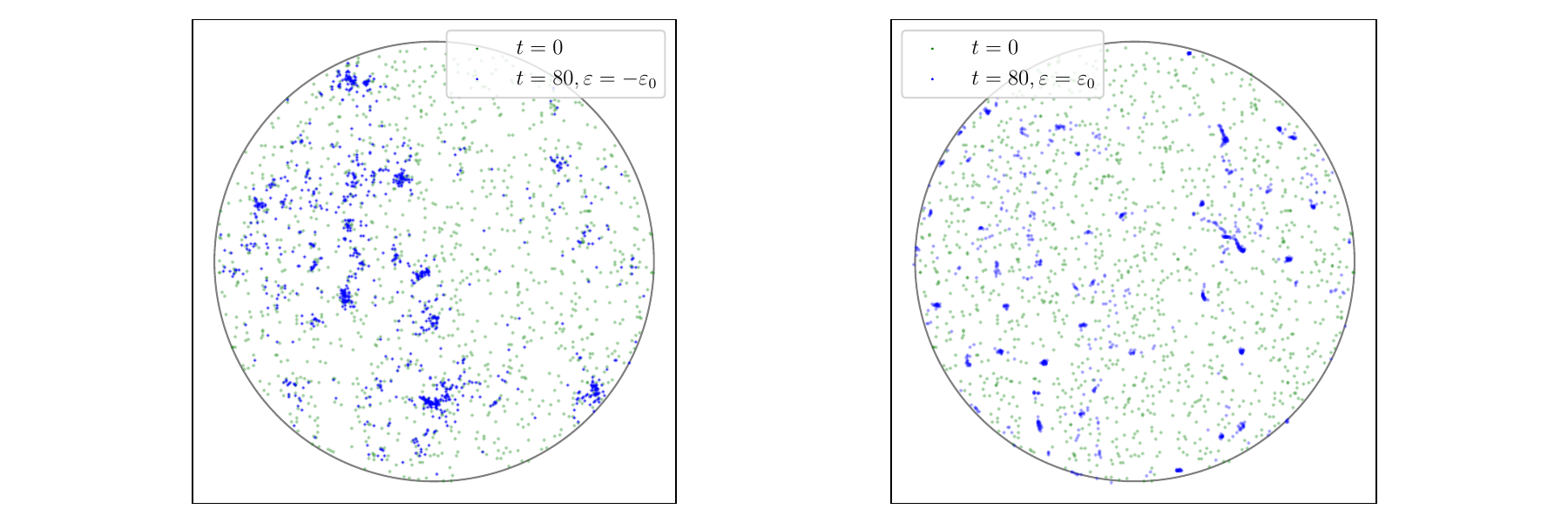}
    \end{subfigure}
    \begin{subfigure}{\linewidth}
        \centering
        \includegraphics[width=0.75\linewidth]{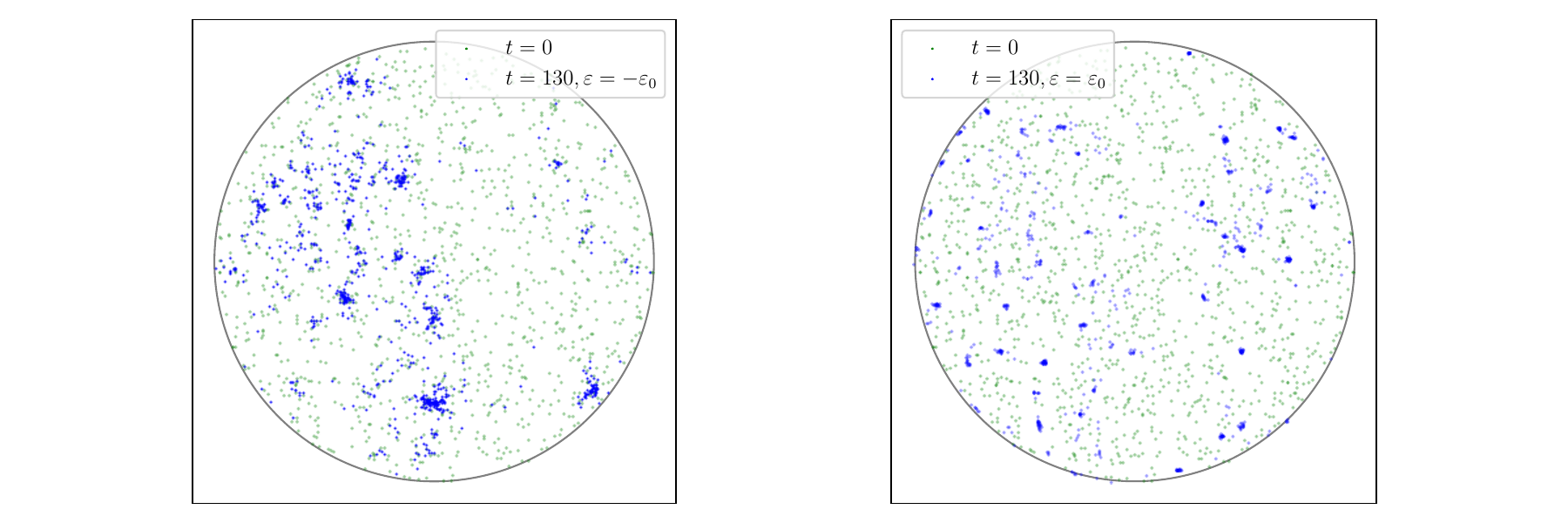}
    \end{subfigure}
    \begin{subfigure}{\linewidth}
        \centering
        \includegraphics[width=0.75\linewidth]{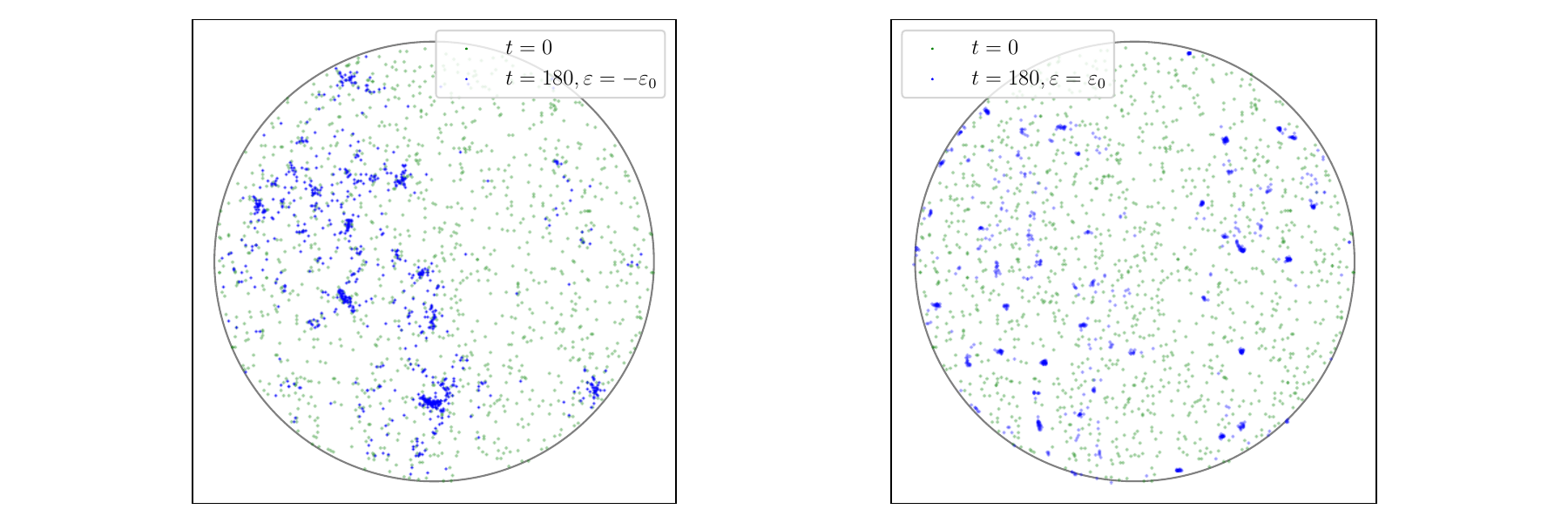}
    \end{subfigure}
    \begin{subfigure}{\linewidth}
        \centering
        \includegraphics[width=0.75\linewidth]{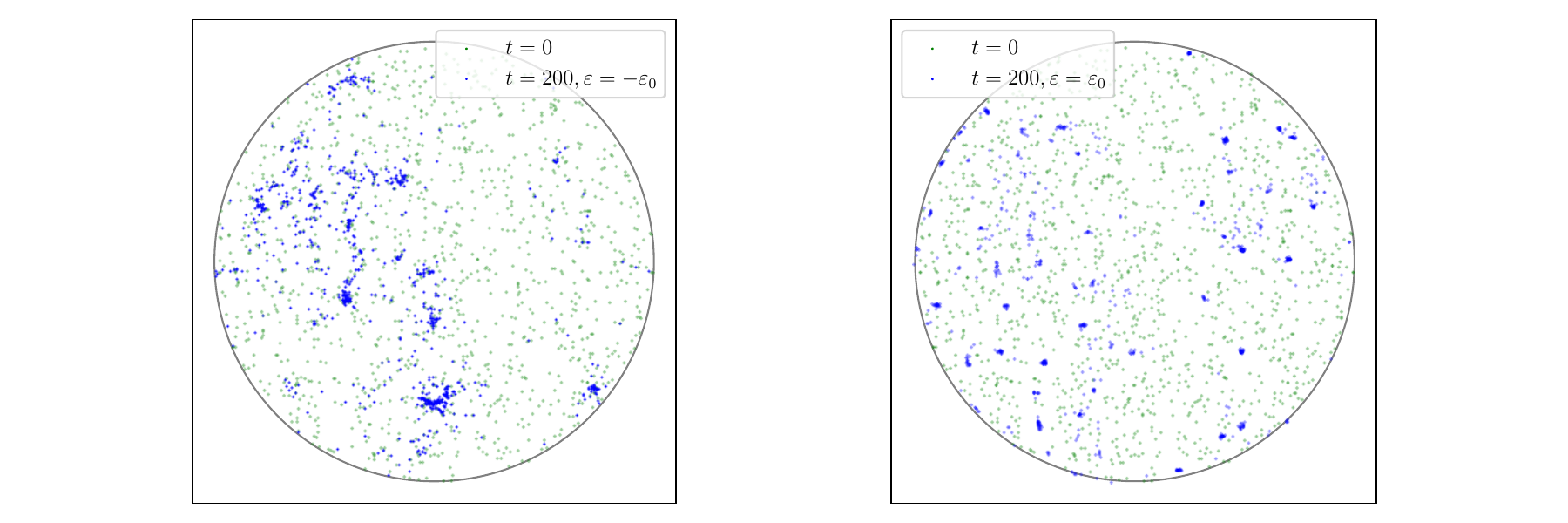}
    \end{subfigure}
    \caption{The green points represent a sample from a PPP $\calP$ of unit intensity, the blue points correspond to $\Pi_{-\varepsilon_0, t}\calP$ (left), and $\Pi_{+\varepsilon_0, t} \calP$ (right), for $t \in \{30, 80, 130, 180, 200\}$.}
    \label{fig:extreme_repelled_additional}
\end{figure}
%--------------------------------
%---------------------------
\bmsubsection{Further details on control variates} % (fold)
\label{a:comparison_of_CV_estimators}
In this section, we discuss three variants of the CV estimator \eqref{eq:estimator_MCCV}, and compare their estimated mean squared error in a numerical experiment in Figures~\ref{fig:exp_mccv_d_3} and \ref{fig:exp_mccv_d_12}.

The CV estimator \eqref{eq:estimator_MCCV} requires fixing $c$ and $(b_\alpha)$.
Were $(b_\alpha)$ known, we would usually estimate $c$ with a simple linear regression using the same sample $\mathcal{B}_N$ used in the evaluation of $\widehat{I}_{\mathrm{MC}}$ in \eqref{eq:estimator_MCCV}; see \citep[Equation 8.29 and below]{Owen2013}. 
More precisely, we would let
\begin{equation}
    \label{e:first_lsq}
    \left(\mu^{(1)}, c^{(1)}\right) \in \arg\min_{\mu, c\in\mathbb{R}} \sum_{\bfx\in\mathcal{B}_N} \left(f(\bfx) - \mu - c \sum_{\vert\alpha\vert\leq 2} b_{\alpha} (\mathbf{x}^\alpha-m_\alpha)\right)^2,
\end{equation} 
and take $c=c^{(1)}$ in \eqref{eq:estimator_MCCV}, actually resulting in $\widehat{I}_{\mathrm{MCCV}}(f) = \mu^{(1)}$.
Note that because we subtract the known integrals $m_\alpha$, there is no constant term in the polynomial in \eqref{e:first_lsq}, so that $\mu^{(1)}$ is well-defined.
Were the $(b_\alpha)$ unknown, the common practice would be to set $c=1$, and similarly estimate $(b_\alpha)$ by multiple linear regression of $f$ onto the monomials of degree up to $2$, using again the same sample $\mathcal{B}_N$.
Formally, let
\begin{equation}
    \label{e:second_lsq}
    \left(\mu^{(2)}, (b_\alpha^{(2)})\right) \in \arg\min_{\mu, (\beta_\alpha)} \sum_{\bfx\in\mathcal{B}_N} \left(f(\bfx) - \mu - \sum_{\vert\alpha\vert\leq 2} \beta_\alpha (\bfx^\alpha-m_\alpha)\right)^2.
\end{equation}
Setting $(b_\alpha) = (b_\alpha^{(2)})$ in \eqref{eq:estimator_MCCV} again yields $\widehat{I}_{\mathrm{MCCV}}(f) = \mu^{(2)}$. 
We label this estimator as \emph{MCCV\_OLS} in Figures~\ref{fig:exp_mccv_d_3} and \ref{fig:exp_mccv_d_12}.
This is often the default control-variate estimator; see \citep[Algorithm 8.3]{Owen2013}. 
However, relying on $\mathcal{B}_N$ both in the Monte Carlo estimator in \eqref{eq:estimator_MCCV} and to estimate either $c$ or $(b_\alpha)$ makes the estimators $\mu^{(1)}$ and $\mu^{(2)}$ in \eqref{e:first_lsq} and \eqref{e:second_lsq} biased.

For a more transparent comparison with the other estimators in the benchmark of Section~\ref{sub:Experminents}, which are all unbiased, we thus follow a recommendation of \citep[Chapter 8]{Owen2013} and estimate the coefficients $(b_\alpha)$ by regressing the evaluations of $f$ on an independent copy $\mathcal{B}'_N$ of $\mathcal{B}_N$.
Formally, let
\begin{equation}
    \label{e:third_lsq}
    (b_\alpha^{(3)}) \in \arg\min_{(\beta_\alpha)} \sum_{\bfx\in\mathcal{B}'_N} \left(f(\bfx) - \sum_{\vert\alpha\vert\leq 2} \beta_\alpha \bfx^\alpha\right)^2,
\end{equation}
and set $(b_\alpha)$ to $(b_\alpha^{(3)})$ and $c$ to $1$ in \eqref{eq:estimator_MCCV}.
The resulting estimator is unbiased, and we label it as \emph{MCCV\_2N} in Figures~\ref{fig:exp_mccv_d_3} and \ref{fig:exp_mccv_d_12}.
At the cost  of yet another independent copy $\mathcal{B}''_N$ of $\mathcal{B}_N$, one can further estimate $c$ in \eqref{eq:estimator_MCCV} and preserve unbiasedness. 
More precisely, define
\begin{equation}
    \label{e:fourth_lsq}
    c^{(4)} \in \arg\min_{c\in\mathbb{R}} \sum_{\bfx\in\mathcal{B}_N''} \left(f(\bfx) - c \sum_{\vert\alpha\vert\leq 2} b_{\alpha}^{(3)} \mathbf{x}^\alpha\right)^2,
\end{equation}
and set $(b_\alpha)$ to $(b_\alpha^{(3)})$ and $c$ to $c^{(4)}$ in \eqref{eq:estimator_MCCV}.
This is the MCCV estimator we consider in our experiments of Section~\ref{sub:Experminents}; we label it as \emph{MCCV} in Figures~\ref{fig:exp_mccv_d_3} and \ref{fig:exp_mccv_d_12}.
    
The estimated MSEs of the three estimators described here, on the same integration tasks as in Section~\ref{sub:Experminents}, are shown in Figures~\ref{fig:exp_mccv_d_3} and \ref{fig:exp_mccv_d_12}. 
The least squares problems are solved using standard NumPy routines, choosing by default the minimizer with smallest $2$-norm when the argmin is not a singleton. 
Note that since some of these estimators are biased, we report estimated mean squared errors rather than estimated variances.
We observe that the three variants of the estimator behave very similarly, to the point of being statistically indistinguishable in most of the regimes we are considering. 
The only noticeable fact is that, as dimension increases above $7$ in Figure~\ref{fig:exp_mccv_d_12}, crude Monte Carlo (in blue) becomes competitive again, while both \emph{MCCV\_OLS} and \emph{MCCV\_2N} become less accurate for smaller values of $N$. 
Intuitively, the corresponding values of the sample size are not large enough for the estimated polynomial coefficients to provide a useful approximation to the integrand, actually resulting in a worse MSE than crude Monte Carlo for these two variants.
In contrast, the three-sample estimator \emph{MCCV} recovers the behavior of crude Monte Carlo in that regime, thanks to the additional estimation of the parameter $c$ in \eqref{eq:estimator_MCCV}.
Overall, these additional experiments justify our choice of the \emph{MCCV} estimator as an idealized control-variate baseline in Section~\ref{sub:Experminents}. 		

%\bibliography{ref}
\begin{figure}[!h]
    \centering
    \begin{subfigure}{\linewidth}
        \centering
        \includegraphics[width=0.75\linewidth]{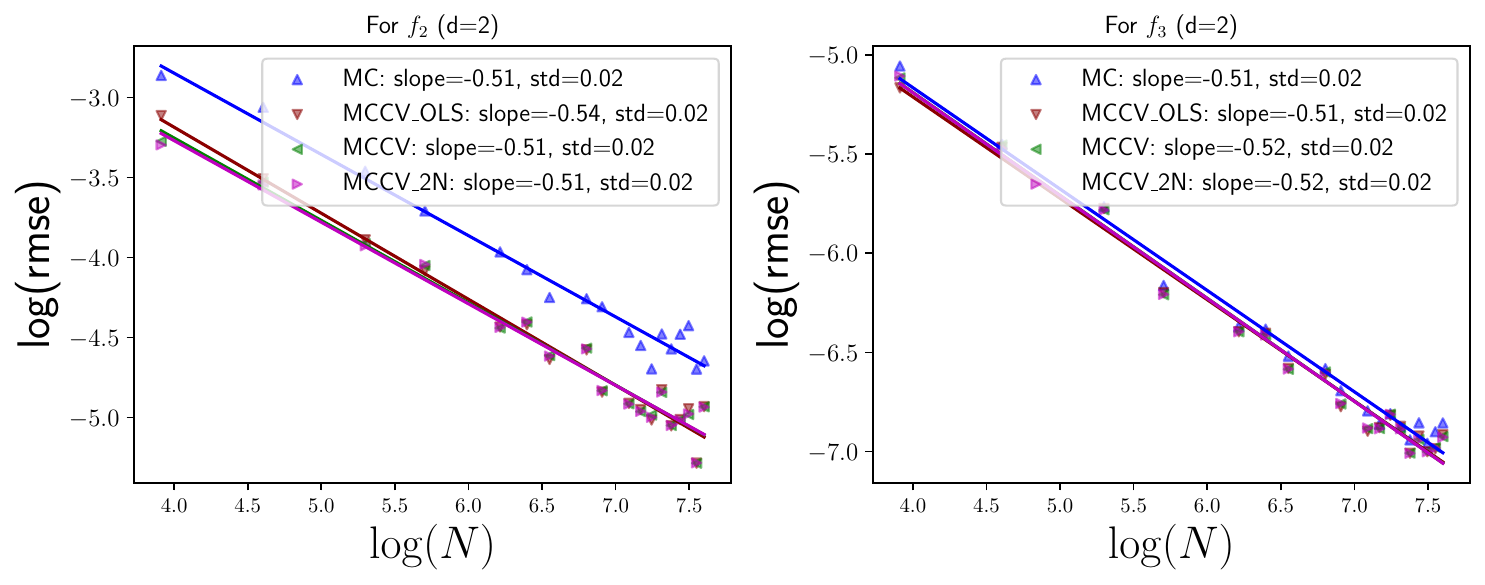}
    \end{subfigure}
    \begin{subfigure}{\linewidth}
        \centering
        \includegraphics[width=0.75\linewidth]{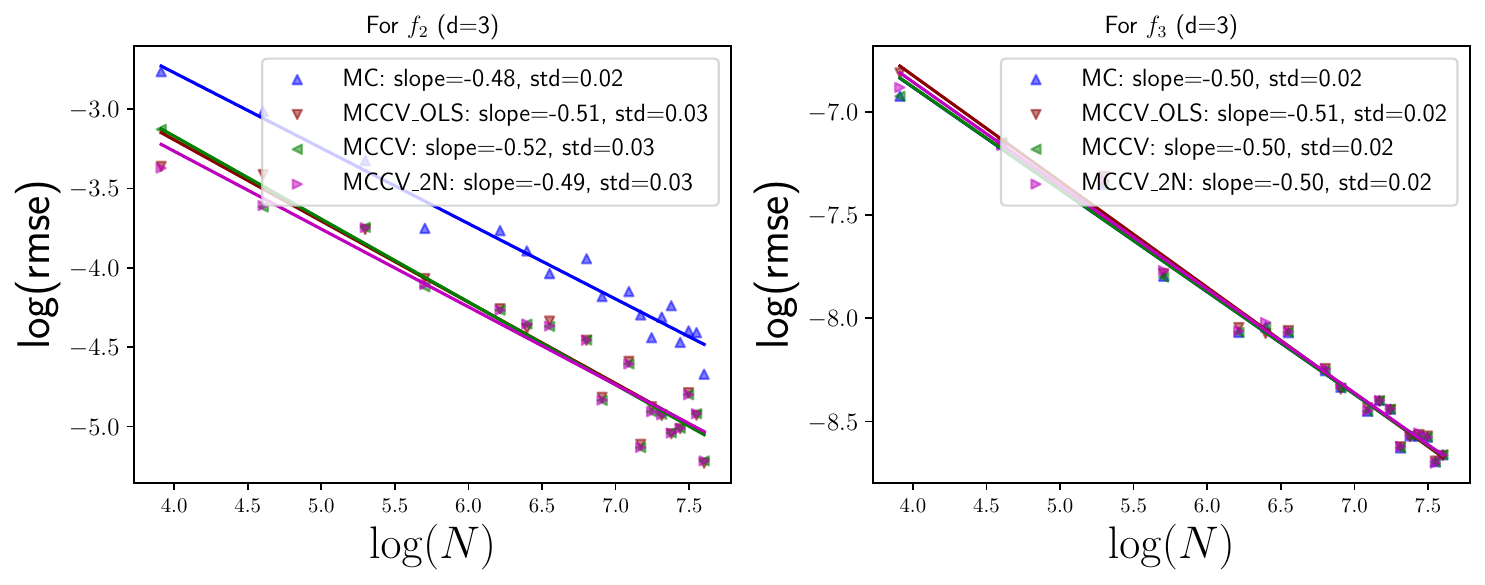}
    \end{subfigure}
    \begin{subfigure}{\linewidth}
        \centering
        \includegraphics[width=0.75\linewidth]{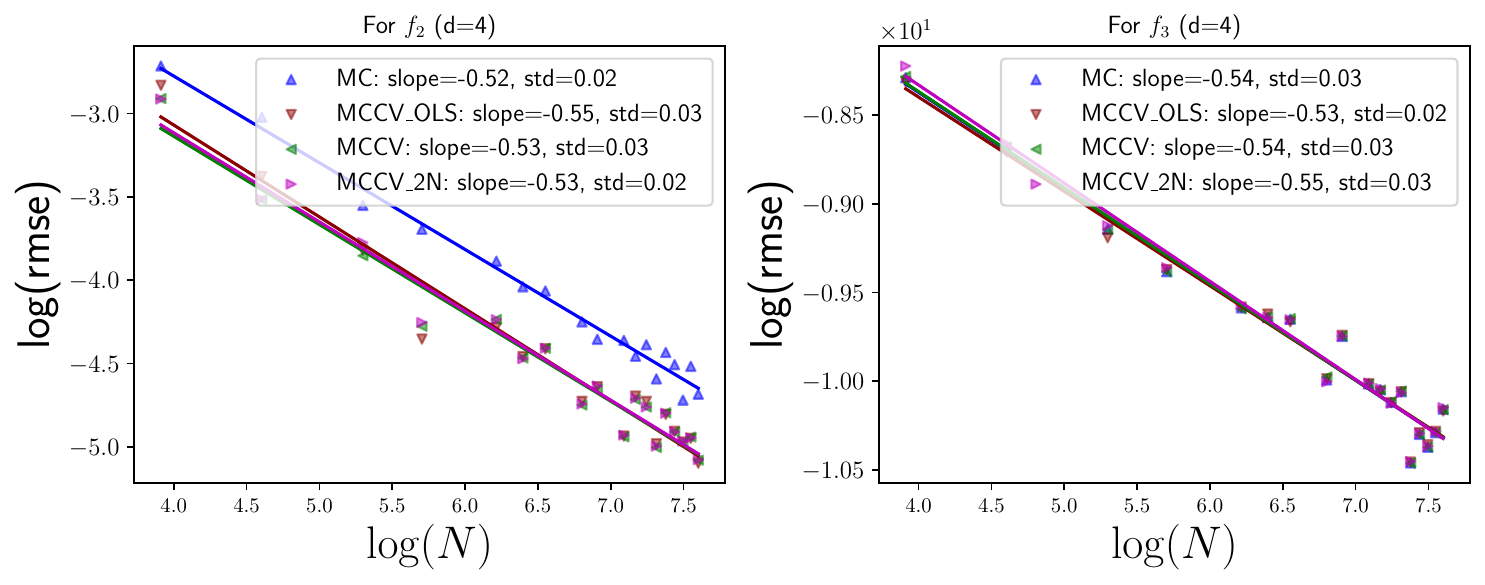}
    \end{subfigure}
    \caption{RMSE of the three differents approaches for MCCV for $d \in \{2, 3, 4\}$.}
    \label{fig:exp_mccv_d_3}
\end{figure}

\begin{figure}[!h]
    \centering
    \begin{subfigure}{\linewidth}
        \centering
        \includegraphics[width=0.75\linewidth]{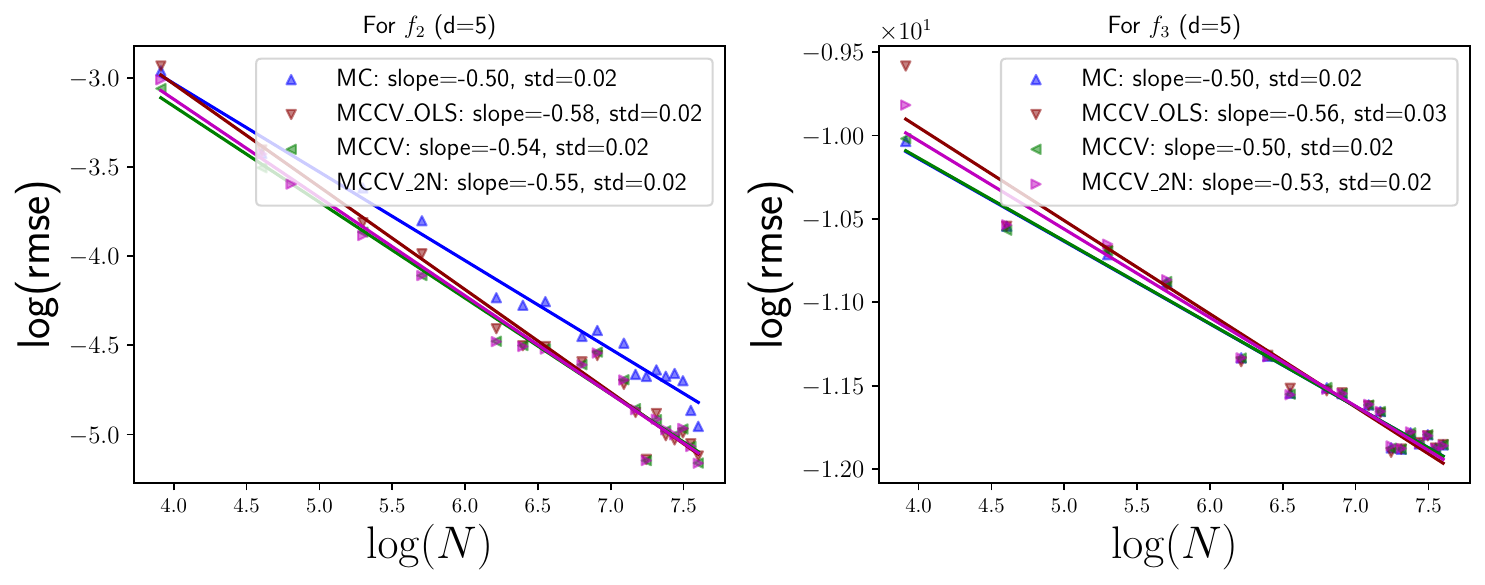}
    \end{subfigure}
    \begin{subfigure}{\linewidth}
        \centering
        \includegraphics[width=0.75\linewidth]{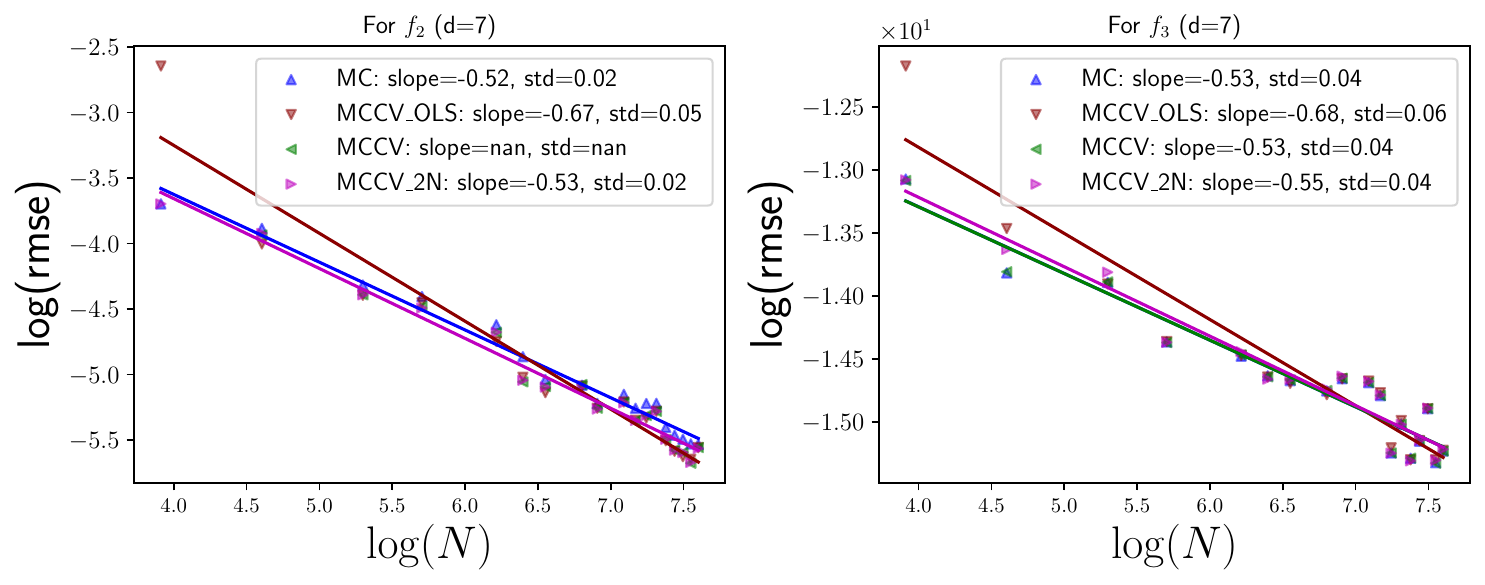}
    \end{subfigure}
    \begin{subfigure}{\linewidth}
        \centering
        \includegraphics[width=0.75\linewidth]{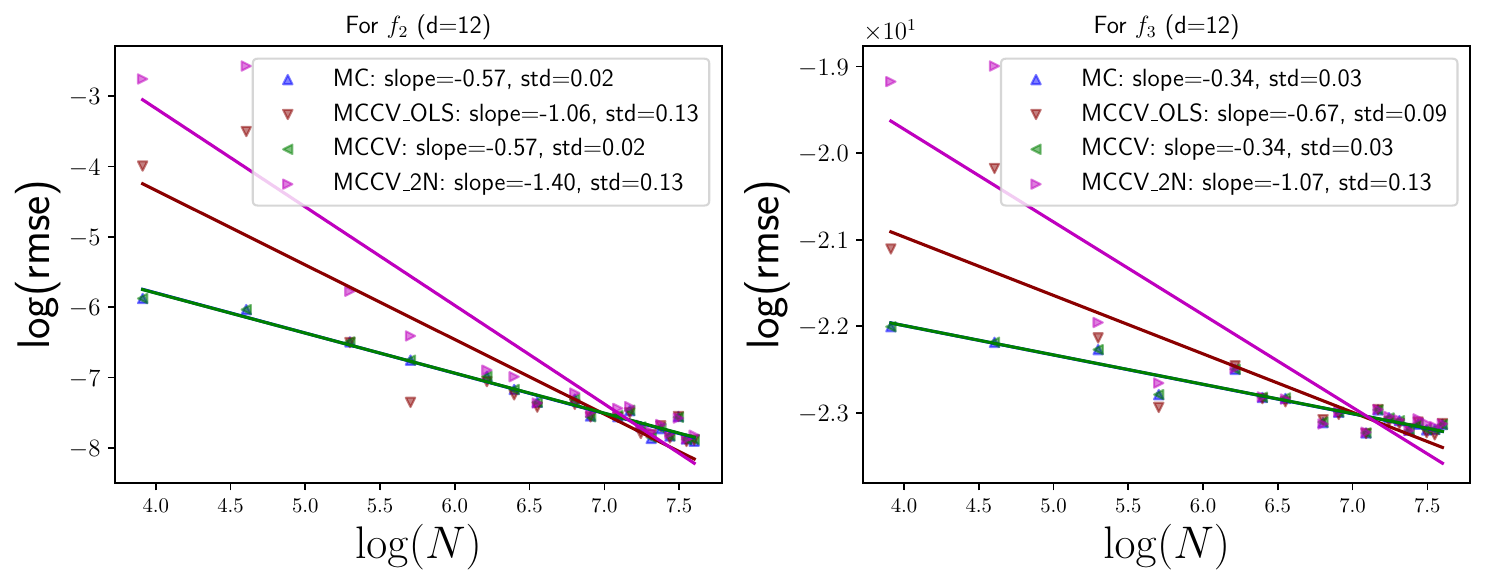}
    \end{subfigure}
    \caption{RMSE of the three differents approaches for MCCV $d \in \{5, 7, 12\}$. }
    \label{fig:exp_mccv_d_12}
\end{figure}

\end{document}